\documentclass[aps,pra,groupedaddress,amsmath,amssymb,tightenlines,notitlepage,nofootinbib,onecolumn,11pt]{revtex4-2}
\pdfoutput=1

\usepackage{graphicx}
\usepackage[dvipsnames,svgnames]{xcolor}
\usepackage{bm,bbm}%
\usepackage{braket}
\usepackage{amsthm,amsmath,amssymb,amsbsy}
\usepackage{enumerate}
\usepackage[caption=false]{subfig}
\usepackage{booktabs,multirow,moresize}
\usepackage{mathtools,mathrsfs,bbding}
\usepackage[percent]{overpic}
\usepackage{microtype}

\definecolor{quantumviolet}{HTML}{53257F}
\definecolor{quantumgray}{HTML}{555555}

\usepackage{newpxtext,newpxmath}

\let\counterwithout\relax
\let\counterwithin\relax
\usepackage{chngcntr}
\newcounter{parentnumber}

\usepackage[T1]{fontenc}
\usepackage[utf8]{inputenc}

\usepackage{array,adjustbox,afterpage}

\newtheoremstyle{red}{}{}{\itshape}{}{\color{red!80!black}\bfseries}{.}{ }{}

\definecolor{darkred}{rgb}{0.57,0,0.12}
\usepackage[colorlinks=true, linkcolor=quantumviolet, urlcolor=quantumviolet, citecolor=quantumviolet]{hyperref}

\usepackage{cleveref}
\usepackage[most,breakable]{tcolorbox}

\let\nc\newcommand

\usepackage{etoolbox}
\hypersetup{
    bookmarksnumbered=true, 
    breaklinks=true,
    unicode=false, 
    pdfstartview={FitH}, 
    pdftitle={Tight constraints on probabilistic convertibility of quantum states}, 
    pdfnewwindow=true,
    colorlinks=true, 
    allcolors=quantumviolet,
    linktoc=section,
}

\nc{\note}[1]{{\color{blue!90!black} #1}}
\nc{\noteb}[1]{{\color{red!80!black}\textbf{#1}}}

\nc{\incoh}[1]{{\mathcal{C}^{(#1-1)}}}
\nc{\produc}[1]{{\mathcal{P}^{(#1)}}}

\nc{\cmin}{c_{\min}}
\nc{\CG}{C_{\text{GME}}}

\nc{\gram}[1]{G^{(#1)}}

\nc{\pen}{^{-1}}

\DeclareMathOperator{\Tr}{Tr}

\DeclareMathOperator{\NN}{NN}
\DeclareMathOperator{\supp}{supp}

\DeclareMathOperator{\sspan}{span}

\DeclareMathOperator{\aff}{aff}

\let\Re\relax
\DeclareMathOperator{\Re}{Re}

\nc{\psic}{\psi^{c}}
\nc{\two}[1]{\underline{2^{d-#1}}}

\nc{\Hanc}{\mathcal{H}_{\text{anc}}}
\nc{\psianc}{\psi_{\text{anc}}}
\nc{\lampow}{\lambda^{1/d}}

\nc{\norm}[2]{\left\lVert#1\right\rVert_{\,#2}}
\nc{\proj}[1]{\ket{#1}\!\bra{#1}}
\nc{\pro}[1]{#1 #1^\dagger}
\nc{\lnorm}[2]{\left\lVert#1\right\rVert_{\ell_{#2}}}

\nc{\VV}{V^{\,\C}_\Sp}

\nc{\R}{\mathcal{R}}
\nc{\W}{\mathcal{W}}
\nc{\T}{\mathcal{T}}
\nc{\B}{\mathcal{B}}
\nc{\C}{\mathcal{C}}
\nc{\U}{\mathcal{U}}
\nc{\E}{\mathcal{E}}
\nc{\EE}{\mathscr{E}}
\nc{\K}{\mathcal{K}}
\nc{\V}{\mathcal{V}}
\nc{\X}{\mathcal{X}}
\nc{\F}{\mathbb{F}}
\nc{\D}{\mathcal{D}}
\nc{\Y}{\mathcal{Y}}

\nc{\M}{\mathcal{M}}
\nc{\N}{\mathcal{N}}
\nc{\I}{\mathcal{I}}
\nc{\Q}{\mathbb{Q}}
\nc{\RR}{\mathbb{R}}
\nc{\CC}{\mathbb{C}}

\nc{\HH}{\mathbb{H}}
\nc{\MM}{\mathbb{M}}
\nc{\DD}{\mathbb{D}}
\renewcommand{\NN}{\mathbb{N}}
\renewcommand{\O}{\mathbb{O}}
\nc{\J}{\mathcal{J}}

\let\FF\F
\let\OO\O

\newcommand\floor[1]{\left\lfloor#1\right\rfloor}

\let\succc\succ

\nc{\mleq}{\preceq}
\nc{\mgeq}{\succeq}
\nc{\mle}{\prec}
\nc{\mge}{\succc}

\renewcommand{\succ}{{\mathrm{succ}}}

\nc{\CPTP}{{\mathrm{CPTP}}}
\nc{\CP}{{\mathrm{CP}}}
\nc{\CPTN}{{\mathrm{CPTN}}}

\nc{\ext}[1]{\operatorname{ext}\left(#1\right)}

\renewcommand{\*}{\textup{*}}
\nc{\<}{\left\langle}
\renewcommand{\>}{\right\rangle}

\renewcommand{\bar}{\;\rule{0pt}{9.5pt}\right|\;}

\nc{\lset}{\left\{\left.}
\nc{\rset}{\right\}}

\DeclareMathOperator{\cone}{cone}

\nc{\ve}{\varepsilon}

\nc{\cbraket}[1]{\left|\braket{#1}\right|}

\nc{\id}{\mathbbm{1}}
\nc{\idc}{\mathrm{id}}

\nc{\mnorm}[1]{\norm{#1}{[m]}}
\nc{\knorm}[1]{\norm{#1}{(k)}}

\renewenvironment{boxed}[1]%
	{\expandafter\ifstrequal\expandafter{#1}{orange}{\begin{tcolorbox}[colback=red!15,colframe=orange!15,breakable,enhanced]}{\begin{tcolorbox}[colback=white,colframe=gray!10,breakable,enhanced]}}%
	{\end{tcolorbox}}
  
\nc{\red}[1]{{\color{red!50!black} #1}}
{\begin{boxed}{white}\vspace{-\baselineskip}\begin{equation}}%
{\end{equation}\end{boxed}}

\theoremstyle{plain}

\newtheorem{theorem}{Theorem}

\newtheorem{proposition}[theorem]{Proposition}
\newtheorem{corollary}[theorem]{Corollary}
\newtheorem{lemma}[theorem]{Lemma}

\theoremstyle{definition}

\newtheorem*{remark}{Remark}

\theoremstyle{red}

\let\oldproofname\proofname
\renewcommand{\proofname}{\rm\bf{\oldproofname}}



\nc{\lsetr}{\left\{\,}
\nc{\rsetr}{\right.\right\}}
\nc{\barr}{\;\rule{0pt}{9.5pt}\left|\;}
\nc{\ketbra}[2]{\ket{#1}\!\bra{#2}}

\nc{\prob}{\mathrm{prob}}

\nc{\eqt}[1]{\stackrel{\mathclap{\mbox{\scriptsize #1}}}{=}}

\nc{\wt}{\widetilde}

\nc{\Rmax}{R_{\max}}
\nc{\Pmax}{R_{\max}}
\nc{\Qmin}{Q_{\min}}
\nc{\Qmax}{V_{\FF}}
\nc{\Qmaxa}{V_{\aff(\FF)}}
\nc{\Rs}{R_{s}}

\nc{\RW}{\Omega_\F}

\nc{\din}{d_\mathrm{in}}
\nc{\dout}{d_\mathrm{out}}

\nc{\gbar}{,\;}

\nc{\Hh}{\wt{H}}

\nc{\HP}{\Xi}

\nc{\bb}{{\bullet\bullet}}
\nc{\phim}{{\phi_{\star}}}

\makeatletter
\def\l@f@section{%
 \addpenalty{\@secpenalty}%
 \addvspace{0.4em plus\p@}%
}%
\makeatother

\usepackage{changepage}

\usepackage{stackengine}
\stackMath
\newcommand\xxrightarrow[2][]{\mathrel{%
  \setbox2=\hbox{\stackon{\scriptstyle#1}{\scriptstyle#2}}%
  \stackunder[0pt]{%
    \xrightarrow{\makebox[\dimexpr\wd2\relax]{$\scriptstyle#2$}}%
  }{%
   \scriptstyle#1%
  }%
}}

\newcommand{\tends}{\xrightarrow{\vphantom{.}\smash{{\raisebox{-1.8pt}{\scriptsize{$n\!\to\!\infty$}}}}}}
\newcommand{\tendsk}{\xrightarrow{\vphantom{.}\smash{{\raisebox{-1.8pt}{\scriptsize{$k\!\to\!\infty$}}}}}}

\newcommand{\transf}{\xxrightarrow[\vphantom{a}\smash{{\raisebox{0pt}{\ssmall{$\OO$}}}}]{}}

\makeatletter
\patchcmd{\@ssect@ltx}
    {\addcontentsline{toc}{#1}{\protect\numberline{}#8}}
    {}
    {}
    {}
\makeatother

 \makeatletter

\def\frontmatter@title@above{\noindent}
\def\frontmatter@title@format{\normalfont\sffamily\huge\noindent\hyphenpenalty=5000}
\def\frontmatter@preabstractspace{1.5em}

\def\frontmatter@authorformat{\def\@makefnmark{\relax}{}\vspace*{1.5em}\sffamily\raggedright\large}
\def\frontmatter@above@affiliation{\vspace*{1.5em}}
\def\frontmatter@affiliationfont{\normalfont\sffamily\color{quantumgray}\footnotesize}
\def\frontmatter@title@produce{%
 \begingroup
  \frontmatter@title@above
  \frontmatter@title@format
  \parbox{\linewidth}{\@title}
  \unskip
  \phantomsection\expandafter\@argswap@val\expandafter{\@title}{\addcontentsline{toc}{title}}%
  \@ifx{\@title@aux\@title@aux@cleared}{}{%
   \expandafter\frontmatter@footnote\expandafter{\@title@aux}%
  }%
  \par
  \frontmatter@title@below
 \endgroup
}%
   \makeatother

\begin{document}

 \title{\texorpdfstring{\href{https://quantum-journal.org/?s=Tight\%20constraints\%20on\%20probabilistic\%20convertibility\%20of\%20quantum\%20states&reason=title-click}{\color{quantumviolet}Tight constraints on probabilistic convertibility of quantum states}}{Tight constraints on probabilistic convertibility of quantum states}}
 
 \author{\href{https://orcid.org/0000-0001-7225-071X}{\color{black}Bartosz Regula}}
 \email{bartosz.regula@gmail.com}
\affiliation{\parbox{\linewidth}{Department of Physics, Graduate School of Science, The University of Tokyo, Bunkyo-ku, Tokyo 113-0033, Japan}}


\begin{abstract}
We develop two general approaches to characterising the manipulation of quantum states by means of probabilistic protocols constrained by the limitations of some quantum resource theory.

First, we give a general necessary condition for the existence of a physical transformation between quantum states, obtained using a recently introduced resource monotone based on the Hilbert projective metric.
In all affine quantum resource theories (e.g.\ coherence, asymmetry, imaginarity) as well as in entanglement distillation, we show that the monotone provides a necessary and sufficient condition for one-shot resource convertibility under resource--non-generating operations, and hence no better restrictions on all probabilistic protocols are possible. We use the monotone to establish improved bounds on the performance of both one-shot and many-copy probabilistic resource distillation protocols.

Complementing this approach, we introduce a general method for bounding achievable probabilities in resource transformations under resource--non-generating maps through a family of convex optimisation problems. We show it to tightly characterise single-shot probabilistic distillation in broad types of resource theories, allowing an exact analysis of the trade-offs between the probabilities and errors in distilling maximally resourceful states.
We demonstrate the usefulness of both of our approaches in the study of quantum entanglement distillation.
\end{abstract}

\begin{adjustwidth*}{.8cm}{.8cm}
 \maketitle
\end{adjustwidth*}


\vspace*{-1em}
 \tableofcontents

\section{Introduction}

Establishing precise conditions for the existence of a physical transformation between two quantum states under different practical constraints has attracted a large body of work --- from the early contributions by Alberti and Uhlmann~\cite{alberti_1980}, to the characterisation of entanglement transformations by Nielsen~\cite{nielsen_1999} and Vidal~\cite{vidal_1999-1}, to more recent efforts to generalise such conditions in various ways~\cite{chefles_2004,buscemi_2012,reeb_2011,heinosaari_2012,horodecki_2013,gour_2015,alhambra_2016,buscemi_2017,gour_2017,gour_2018,takagi_2019,liu_2019,buscemi_2019,dallarno_2020,regula_2020,zhou_2020}.
Many problems of this type can be understood as instances of the question of resource convertibility within the formalism of quantum resource theories~\cite{horodecki_2013-3,chitambar_2019}. However, most works that dealt with resource transformations, especially in the context of general quantum resources, focused only on deterministic transformations~\cite{brandao_2015,gour_2017,takagi_2019,liu_2019,regula_2020,fang_2020,gonda_2019,zhou_2020,hsieh_2020,kuroiwa_2020,ferrari_2020,regula_2021-1,fang_2022}. This places strong limitations on the practical utility of such approaches.

In particular, probabilistic protocols --- representing processes that depend on the outcome of some measurements, and thus have a certain probability of failing to perform the desired transformation --- are ubiquitous in the exploitation of quantum resources in practice. They underlie, for instance, the majority of the commonly used schemes for the distillation of entanglement~\cite{bennett_1996,horodecki_2009} or magic (non-stabiliser) states~\cite{bravyi_2005,campbell_2017}. There are strong operational motivations to prefer such protocols to deterministic quantum channels: it is known that probabilistic transformations can be significantly more powerful than deterministic ones within the same resource-theoretic constraints~\cite{lo_2001,vidal_1999-1,dur_2000,horodecki_1999-1}. This fact is often used in practice to `gamble' with quantum resources~\cite{bennett_1996,lo_2001,horodecki_1999-1,rozpedek_2018,fang_2018}: although some desired performance thresholds might not be achievable deterministically, one can often take a risk by sacrificing some success probability in the hopes of obtaining a better performance stochastically.
In some settings, this gap becomes quite extreme, as is the case in the manipulation of pure multipartite entangled states, where deterministic transformations are almost never possible~\cite{devicente_2013,gour_2017-1,sauerwein_2018}, but probabilistic manipulation is nevertheless feasible~\cite{gour_2017-1,sauerwein_2018}.
Therefore, constraining only deterministic protocols does not provide complete information about our capability to manipulate a given resource in practice --- if no-go theorems that restrict deterministic conversion are easy to circumvent by simply allowing a small chance of failure, can they really be considered to place a meaningful limitation on practical transformation schemes? 
In order to understand the ultimate limits of physically realisable transformations of a given resource, we argue that it is necessary to know the extent of the advantage that probabilistic resource manipulation can provide. The issue then is that most of the previous approaches were not suited to such a problem, and no systematic way of constraining the performance of 
probabilistic transformation protocols was known.
Here we aim to bridge this gap and establish methods which allow for a comprehensive characterisation of probabilistic convertibility between quantum states in general quantum resource theories.

We approach the problem of understanding probabilistic convertibility in two different ways. First, we aim to establish limitations on \emph{all} possible probabilistic protocols --- no matter how small their probability of success is. Despite no general bounds of this type being previously known, we show that the problem can be cast in a straightforward mathematical framework that allows us to not only obtain very broadly applicable restrictions, but also prove that they are the best ones possible under some assumption on the given resource theory.
And although such a constraint on all transformations might seem like overkill, we show that it leads to non-trivial bounds that improve on and extend previously known restrictions in several cases. 
Nevertheless, should more fine-grained constraints be desired, we consider also a complementary approach that allows for a deeper understanding of the trade-offs between the errors incurred in transformation protocols and the probabilities with which they can be achieved.

Specifically, we first study the \emph{projective robustness} $\RW$, a resource monotone based on the Hilbert projective metric~\cite{bushell_1973,reeb_2011} that was recently introduced in~\cite{regula_2022}. We provide rigorous proofs of the properties and applications of this quantity in constraining one-shot resource transformations, expanding the results originally announced in~\cite{regula_2022}. We show in particular that $\RW$ is monotonic under all resource transformations, deterministic or probabilistic, which means that it provides a necessary condition for the existence of the most general physical conversion protocol between any two quantum states. 
Under a suitably general type of transformations known as resource--non-generating operations, we establish also a converse of this result, showing that the projective robustness together with a closely related monotone called the \emph{free projective robustness} $\RW^\FF$ can give sufficient conditions for resource convertibility. The latter result applies in particular to all affine resource theories --- a broad class that includes the theories of quantum coherence, asymmetry, athermality, or imaginarity --- as well as in the study of bi- and multi-partite entanglement distillation. Focusing on the task of probabilistic distillation (purification) of quantum resources, we show how $\RW$ can be used to establish strong limitations on the error achievable in such manipulation, which become tight in relevant cases and yield significant improvements over previous results of this type. The properties of $\RW$ allow us to extend such bounds to many-copy distillation protocols, constraining the overheads required in probabilistic resource distillation. We also use $\RW$ to show that in any convex resource theory there exist `isotropic-like' quantum states whose fidelity with a pure state cannot be increased by any one-shot transformation protocol.

Our results also shed light on recent no-go theorems for probabilistic resource manipulation in Ref.~\cite{fang_2020}: we demonstrate that $\RW$ leads to a significant improvement over the quantitative bounds on distillation errors shown in that work, and we prove that the restrictions derived there do not actually establish any trade-offs between probability of success and transformation accuracy.

The second approach of this work is concerned with the question: when a probabilistic transformation between given states \emph{is} possible, what is the best achievable probability of success?
To address this, we introduce a family of convex optimisation problems that can precisely delineate the trade-offs between transformation errors and the probability of successful conversion in general quantum resource theories. We obtain upper and lower bounds on the maximal achievable probability of resource transformations, as well as bounds on the least error achievable with given probability. We further show that, for several types of resources, our bounds become tight in the task of resource distillation, where one aims to purify a given state to a maximally resourceful one. This allows us to obtain an exact quantitative characterisation of both exact and approximate one-shot probabilistic resource distillation.

Throughout this work, we exemplify our results by applying them to distillation of quantum entanglement, showing how our methods can both extend previously known insights as well as provide new ones.

This paper is structured as follows. 

\begin{itemize}
  \item In Section~\ref{sec:prelim}, we overview the formalism of quantum resource theories, probabilistic resource conversion, and resource monotones, defining the notation and framework of the paper.
  \item Section~\ref{sec:projective_robustness} is concerned with defining the projective robustness monotones and establishing their basic properties. Here we present our main results in the form of general no-go theorems for resource transformations (Theorem~\ref{thm:nogo_monotonicity}), as well as establish conditions for their tightness (Theorems~\ref{thm:nogo_affine_full} and~\ref{thm:nogo_sufficient_fulldim_full}).
  \item In Section~\ref{sec:prob_dist_RW} we apply the projective robustness to resource distillation, establishing a general lower bound for the error of any physical distillation protocol (Theorem~\ref{thm:prob_error_full}). We derive corresponding upper bounds (Theorem~\ref{thm:RW_dist_achiev}) that are tight in relevant cases, and compare our methods with previous state-of-the-art results.\\[.75\baselineskip](Sections~\ref{sec:projective_robustness}--\ref{sec:prob_dist_RW} thus also serve as a technical companion to and an extension of Ref.~\cite{regula_2022}, where the projective robustness was first announced.)
  \item Section~\ref{sec:tradeoff} introduces our second approach, which allows for the study of trade-offs between transformation error and achievable probability. Here we give in particular tight quantitative constraints governing the probabilities and errors in probabilistic distillation protocols (Theorem~\ref{thm:prob_ve_bounds}). 
\end{itemize}

\section{Preliminaries}\label{sec:prelim}

Our discussion will pertain to self-adjoint operators acting on finite-dimensional Hilbert spaces. We will be working in the real Hilbert space obtained by endowing the vector space of such operators with the Hilbert--Schmidt inner product $\<A, B\> = \Tr A B$. We use $\RR_+$ ($\RR_{++}$) to denote the non-negative (positive) reals, $\HH$ to denote all self-adjoint operators, and $\HH_+$ to denote positive semidefinite operators. We use $\DD$ to denote the set of density operators (quantum states), i.e.\ positive semidefinite operators of unit trace.
Following conventions common in optimisation theory, we take $\inf \emptyset = \infty$ and we will allow basic algebraic operations involving infinity, such as $c + \infty = \infty$ and $c \infty = \infty$ for any $c > 0$. Quantities of the form $1/0$ will be understood as the corresponding limit over positive reals, i.e.\ $0^{-1} = \infty$ and $\infty^{-1}= 0$.

\subsection{Resource theories and resource--non-generating operations}

In a broad sense, the question that we concern ourselves with in this work is the (im)possibility of achieving the transformation from a state $\rho$ to another state $\rho'$, with some physical constraints imposed on the physical processes available to us. Such constraints can be naturally characterised in the formalism of quantum resource theories~\cite{horodecki_2013-3,chitambar_2019}. 
In this work, we always understand quantum resource theories to mean theories of \emph{state} manipulation, although more general formulations of this formalism can also be made~\cite{devetak_2008,coecke_2016,delrio_2015,takagi_2019,liu_2020,gour_2019-1}.

We thus assume that a certain set of quantum states $\FF$ is designated as the \emph{free states}, dividing the state space into states which carry no resources ($\rho \in \FF$) and states which are deemed resourceful ($\rho \notin \FF$).
Throughout this manuscript, we make the basic assumption that the set $\FF$ is topologically closed; some of our results will also require that $\FF$ be a convex set, which is also a natural assumption in most settings, as it simply means that probabilistically mixing two free states cannot produce a resourceful state. 
When studying transformations of states between different spaces, we will implicitly understand that each space has its own associated set of free states $\FF$.

In addition to the free states, a basic ingredient of any resource theory are the \emph{free operations} $\OO$ --- these are the transformations available to us in the given restricted setting, the convertibility of states under which we wish to understand.
The choice of the free operations can be motivated by physical considerations, representing the operations which are particularly `easy' or `cheap' to implement without expending any resources. There can, however, be other motivations to study the conversion of resources under various types of free operations. For instance, it may be the case that enlarging the allowed operations even slightly, e.g.\ by supplying some additional resources for free, can allow one to overcome the limitations that more restricted operations suffer from~\cite{eggeling_2001,audenaert_2003,ishizaka_2004-1,brandao_2010,brandao_2015,berta_2022,faist_2015}.
In other contexts, it might be the case that natural choices of easily implementable operations are actually extremely weak and trivialise the given theory, making almost all transformations impossible~\cite{gour_2017-1,sauerwein_2018,chitambar_2016,lami_2019-1,lami_2020-1}. Depending on the resource theory, it might not even be clear what a suitable choice of physical free operations should be~\cite{chitambar_2016}. In such cases, it is justified to look at transformations under various types of operations, and in particular to seek to understand general properties and restrictions shared by all feasible choices of `free' channels.

The approach we take in this work is to consider all \textbf{\emph{resource--non-generating operations}}, and hereafter we use $\OO$ to denote such maps. They constitute the largest set of free operations which is consistent with $\FF$ being the set of free states, in the sense that resourceful states cannot be created for free. Focusing for now on deterministic transformations (i.e.\ quantum channels), we define a channel $\E$ to be resource non-generating if $\E(\sigma) \in \FF$ for all free $\sigma \in \FF$.

These transformations have already found success in the description of resources such as quantum entanglement~\cite{brandao_2010,contreras-tejada_2019,lami_2021-1}, quantum coherence~\cite{chitambar_2016}, or quantum thermodynamics~\cite{faist_2015,faist_2018}, as well as in general resource-theoretic frameworks~\cite{chitambar_2019,liu_2019,regula_2020}. There are several reasons to study the transformations of resources under such operations. First, their crucial property is that any no-go result shown for the resource--non-generating operations $\OO$ will necessarily apply to all other physical types of free operations, by virtue of $\OO$ being the most permissive free transformations. Importantly, such restrictions can be considered to be robust: if a restriction is shown only for a subset $\wt\O \subsetneq \OO$, then it might be the case that some other, larger type of free operations can transcend it; however, a restriction valid for all transformations $\OO$ cannot be overcome by \emph{any} free operations. In this sense, resource--non-generating operations are perfectly suited to characterise the ultimate limitations of a given theory, applicable regardless of the specific choice of free operations.

An important reason for our choice of resource--non-generating operations is simply that they can be defined in \emph{any} resource theory: we do not need to assume almost anything about the structure of the considered resources, yet we will be able to make meaningful physical statements about their convertibility. From a technical standpoint, resource--non-generating operations are straightforward to describe mathematically, allowing for a simplified description of the operational aspects of theories where the more `natural' choice of free operations may be very difficult to describe. A trade-off for this mathematical simplicity is that such transformations might not have a form which is easily implementable in practice, meaning that they might require additional resources to be physically realised~\cite{faist_2015,chitambar_2019}. Nevertheless, they can always be understood as an `upper bound' on what is achievable in practice within the given resource-theoretic setting.

\interfootnotelinepenalty=10000

\subsection{Probabilistic resource manipulation}

A general probabilistic protocol can be represented by a quantum instrument~\cite{davies_1970,ozawa_1984} $\{\E_i\}$, that is, a collection of maps which are completely positive and such that the overall transformation $\sum_i \E_i$ is trace preserving. In this setting, $\Tr \E_i(\rho)$ can be understood as the probability corresponding to the $i$th outcome of the instrument, and thus the probability that the state $\rho$ is transformed into~$\frac{\E_i(\rho)}{\Tr \E_i(\rho)}$.

In the context of resource transformations, deterministic conversion between $\rho$ and $\rho'$ is equivalent to the existence of a quantum channel $\E\in\OO$ such that $\E(\rho) = \rho'$. The problem is more involved in the case of probabilistic conversion. Since we are concerned with the transformation between two fixed states, $\rho$ and $\rho'$, we can simply coarse-grain the outcomes of the instrument into two: either the conversion succeeded, or it failed. For such an instrument $\{\E_i\}_{i=0}^1$ to be considered resource non-generating, we then impose that each constituent map $\E_i$ must be resource non-generating, but only in a probabilistic sense: we require that for all $\sigma \in \FF$, it holds that $\E_i(\sigma) = [\Tr \E_i(\sigma)]\, \sigma'$ for some $\sigma' \in \FF$.\footnote{A slightly different way of defining free probabilistic protocols that is sometimes encountered in the study of specific quantum resources is to treat the instrument $\{\E_i\}$ as a single trace-preserving map $\wt\E (\cdot) = \sum_i \E_i(\cdot) \otimes \proj{i}$ and impose that the overall channel $\wt\E$ belongs to $\OO$, where the resource theory has been extended to include the ancillary classical system. In all of the settings considered explicitly here and in~\cite{regula_2022}--- in particular, in the resource theories of entanglement, coherence, and magic --- the two definitions are fully equivalent.
However, the definition based on the `flagged' channel $\wt\E$ can be different from ours e.g.\ in quantum thermodynamics, where pure classical states do not come for free. The definition employed in this work, which only uses the instrument $\{\E_i\}$, has the advantage that it can be applied to general resource theories, without having to explicitly consider how the ancillary classical states are handled.}

In fact, we will not need to consider instruments explicitly, and it will suffice to study the action of sub-normalised (probabilistic) free operations $\E$. These are completely positive and trace--non-increasing maps which are resource non--generating in the above sense. The equivalence between transformations under such maps and more general instruments is thanks to the fact that any probabilistic map can be understood as being part of a larger free instrument: given any sub-normalised free operation $\E$, we can define $\E'(\cdot) \coloneqq [\Tr(\cdot)-\Tr\E(\cdot)] \sigma'$ for some fixed $\sigma' \in \FF$, and the instrument $\{\E, \E'\}$ constitutes a valid free probabilistic instrument.

We thus use $\OO$ to denote both deterministic and probabilistic resource--non-generating maps:
\begin{equation}\begin{aligned}
  \OO = \big\{  \E : \HH \to \HH' \;\big|\;& \E \text{ completely positive and trace non-increasing},\\
   &\forall \sigma \in \FF,\, \exists \sigma' \in \FF,\; p\in [0,1] \text{ s.t.\ } \E(\sigma) = p \sigma' \big\}.
\end{aligned}\end{equation}
We then say that the transformation from $\rho$ to $\rho'$ is possible with probability $p$ if there exists an operation $\E \in \OO$ such that $\E(\rho) = p \rho'$. Here, the output space of the map can differ from its input space, and we assume that a corresponding set $\FF$ is defined in both spaces.

However, this is still not the most general way in which a probabilistic conversion can be realised. There exist cases where the transformation $\E(\rho) = p \rho'$ is impossible with any non-zero probability $p$, but the state $\rho'$ can nevertheless be approached arbitrarily closely~\cite{horodecki_1999-1}. In particular, it is possible that there exists a sequence of operations $(\E_n)_n \in \OO$ such that $\Tr \E_n(\rho) \tends 0$ but
\begin{equation}\begin{aligned}\label{eq:prob_asymp}
 \frac{\E_n(\rho)}{\Tr \E_n(\rho)} \tends \rho'.
\end{aligned}\end{equation}
That is, the transformation from $\rho$ to $\rho'$ might only be possible in the asymptotic limit where the probability of success vanishes as the fidelity approaches 1. We hereafter use the notation $\rho \transf \rho'$ to denote the existence of a free probabilistic protocol transforming $\rho$ to $\rho'$, whether it takes the form of a single operation $\E \in \OO$ or a sequence of free operations satisfying \eqref{eq:prob_asymp}. More precisely,
\begin{equation}\begin{aligned}
  \rho \transf \rho' \;\iff\; \rho' \,\in\, \operatorname{cl} \lset \frac{\E(\rho)}{\Tr \E(\rho)} \bar \E \in \OO, \; \Tr \E(\rho) > 0 \rset,
\end{aligned}\end{equation}
where cl denotes closure. This form represents the most general physical transformation of a quantum resource allowed in the framework of this work.

\subsection{Resource monotones}

A resource monotone (or resource measure) $M_\FF : \HH_+ \to \RR_+ \cup \{\infty\}$ is any function which is monotonic under deterministic free operations, that is, $M_\FF(\rho) \geq M_\FF(\E(\rho))$ for any channel $\E \in \OO$~\cite{vedral_1997,chitambar_2019}. The crucial consequence of this definition is that any monotone can be used to certify the impossibility of deterministic resource conversion: if $M_\FF(\rho) < M_\FF(\rho')$, then there cannot exist a free channel which transforms $\rho$ to $\rho'$ deterministically. When it comes to probabilistic transformations, a desirable feature of a resource monotone is its strong monotonicity~\cite{vedral_1998,vidal_2000}, which states that the monotone should, on average, decrease under the action of a free probabilistic instrument~$\{\E_i\}$:
\begin{equation}\begin{aligned}\label{eq:strong_monotonicity}
  M_\FF(\rho) \geq \sum_{i} \Tr \E_i(\rho)\, M_\FF\left( \frac{\E_i(\rho)}{ \Tr \E_i(\rho)} \right).
\end{aligned}\end{equation}
Such strong monotones can then be used to bound transformation probabilities of quantum resources. Specifically, assume that there exists a probabilistic operation $\E \in \OO$ such that $\E(\rho) = p \rho'$ with some probability $p>0$. Eq.~\eqref{eq:strong_monotonicity} then gives
\begin{equation}\begin{aligned}\label{eq:strong_monot_prob}
  p \leq \frac{M_\FF(\rho)}{M_\FF(\rho')},
\end{aligned}\end{equation}
and so any strong monotone provides a restriction on probabilistic transformations. However, most known monotones can never be used to completely rule out the possibility of probabilistic resource conversion: this would require showing that the transformation cannot be realised with \emph{any} probability $p$ (or equivalently that $p$ must be $0$), but this cannot be concluded from Eq.~\eqref{eq:strong_monot_prob} for finite and non-zero values of $M_\FF$. Even in the cases when tight restrictions on resource conversion can be obtained in such a setting, they typically require the optimisation over infinitely many monotones~\cite{vidal_2000}, potentially limiting the practical applicability of such approaches.

Let us now introduce two resource measures that serve as inspiration for the concepts used in this work. The \emph{(generalised) robustness}~\cite{vidal_1999,datta_2009} is defined as
\begin{equation}\begin{aligned}
  R_\FF(\rho) \coloneqq \inf \lset \lambda \bar \rho \leq \lambda \sigma, \; \sigma \in \FF \rset,
\end{aligned}\end{equation}
where the inequality is with respect to the positive semidefinite cone, in the sense that $A \leq B \iff B - A \in \HH_+$. This monotone found a number of operational uses --- it quantifies the advantages of resources in channel discrimination problems~\cite{takagi_2019-2}, and can be applied to benchmark the performance of resource distillation tasks~\cite{regula_2020,regula_2021-1}. A quantity which is, in a sense, dual to the robustness is the \emph{resource weight}~\cite{lewenstein_1998}
\begin{equation}\begin{aligned}
  W_\FF(\rho) \coloneqq \sup \lset \nu \bar \rho \geq \nu \sigma,\; \sigma \in \FF \rset.
\end{aligned}\end{equation}
Note that this quantity is technically an antitone rather than a monotone --- it satisfies $W_\FF(\rho) \leq W_\FF(\E(\rho))$ for free channels $\E$. 
The resource weight can also be used to study resource advantages in a different type of channel discrimination problems~\cite{uola_2020-1,ducuara_2020}, and it has recently been applied to establish strong bounds on deterministic resource distillation errors and overheads~\cite{regula_2021-1,fang_2022}. The similarity between the two measures is made precise by noticing that they can be both expressed in terms of the max-relative entropy $D_{\max}$~\cite{datta_2009}, which we define in its non-logarithmic form as
\begin{equation}\begin{aligned}
  \Rmax(\rho \| \sigma) \coloneqq 2^{D_{\max}(\rho\|\sigma)} = \inf \lset \lambda \bar \rho \leq \lambda \sigma \rset.
\end{aligned}\end{equation}
We then have
\begin{equation}\begin{aligned}
  R_\FF(\rho) = \min_{\sigma \in \FF} \Rmax(\rho \| \sigma),\qquad W_\FF(\rho) = \left[\min_{\sigma \in \FF} \Rmax(\sigma \| \rho)\right]^{-1}\!.
\end{aligned}\end{equation}

We also introduce a variant of the max-relative entropy $\Rmax$ as
\begin{equation}\begin{aligned}
  \Rmax^\FF (\rho \| \sigma) \coloneqq \inf \lset \lambda \bar \rho \leq_{\FF} \lambda \sigma\rset,
\end{aligned}\end{equation}
where now $\leq_\FF$ denotes inequality with respect to the cone generated by the free set $\FF$; precisely, $A \leq_\FF B \iff B - A \in \cone(\FF)$, where
\begin{equation}\begin{aligned}
  \cone(\FF) \coloneqq \lset \lambda \sigma \bar \lambda \in \RR_+, \; \sigma \in \FF \rset.
\end{aligned}\end{equation}
This quantity can be used to define another monotone known as the \emph{standard robustness} (or \emph{free robustness}) $R^\FF_\FF$~\cite{vidal_1999}, given by\footnote{We note that there are different notational conventions when it comes to the robustness and weight monotones: the original definitions of the generalised and standard robustness would correspond to $R_\FF(\rho)-1$ and $R_\FF^\FF(\rho)-1$ in our notation, respectively, while the weight measure is often defined as $1 - W_\FF(\rho)$. We chose slightly different definitions in order to more easily relate the measures with the max-relative entropy and to make their generalisations clearer.}
\begin{equation}\begin{aligned}
  R^\FF_\FF(\rho) \coloneqq \min_{\sigma \in \FF} \Rmax^\FF (\rho \| \sigma).
\end{aligned}\end{equation}
Let us remark here that the robustness measures introduced here are not necessarily finite without additional assumptions on the theory in consideration. It can be seen that $R_\FF(\rho) < \infty$ for all $\rho$ if and only if $\FF$ contains a state of full rank, which is satisfied in most theories of practical interest. For the standard robustness, it holds that $R_\FF^\FF(\rho) < \infty$ if and only if $\rho \in \sspan(\FF)$, meaning that this monotone will typically find use in resource theories such that $\sspan(\FF) = \HH$, referred to as full-dimensional theories.

We note also that robustness and weight quantifiers can be naturally defined not only for quantum states, but also for unnormalised positive semidefinite operators; in such cases, we get that $R_\FF(\lambda \rho) = \lambda R_\FF(\rho)$ for any $\lambda \in \mathbb{R}_{++}$, and analogously for $W_\FF$ and $R^\FF_\FF$.




\section{Projective robustness}\label{sec:projective_robustness}

Our main object of study in this section will be the \textbf{\textit{projective robustness}} $\RW$~\cite{regula_2022}, defined as
\begin{equation}\begin{aligned}
  \RW(\rho) \coloneqq& \min_{\sigma \in \FF}\, \Rmax(\rho \| \sigma) \, \Rmax(\sigma \| \rho)\\
  =& \inf \lset \frac{\lambda}{\nu} \bar \nu \sigma \leq \rho \leq \lambda \sigma,\; \sigma \in \FF \rset.
\end{aligned}\end{equation}
The motivation for the introduction of this quantity is twofold. On the one hand, the robustness and weight monotones are very useful in constraining deterministic resource transformations~\cite{brandao_2015,regula_2020,regula_2021-1,fang_2022}, but their applicability in probabilistic cases is significantly limited~\cite{regula_2021-1,fang_2022}; the idea is then to exploit the complementary character of $R_\FF$ and $W_\FF$ to introduce a monotone which \emph{combines} the properties of the two measures in order to establish much broader and tighter limitations. On the other hand, the projective robustness can be thought of as a minimisation of the Hilbert projective metric $\ln \Rmax(\rho \| \sigma) \, \Rmax(\sigma \| \rho)$~\cite{bushell_1973,reeb_2011}, a quantity which has already been used to study transformations of pairs of quantum states~\cite{reeb_2011,buscemi_2017}, but whose applications in more general resource-theoretic settings have not been explored.

Analogously, the \textbf{\textit{free projective robustness}}~\cite{regula_2022} is given by
\begin{equation}\begin{aligned}\label{eq:proj_free_def}
  \RW^\FF(\rho) \coloneqq& \min_{\sigma \in \FF}\, R_{\max}^\FF(\rho \| \sigma) \, \Rmax(\sigma \| \rho)\\
  =& \inf \lset \frac{\lambda}{\nu} \bar \nu \sigma \leq \rho \leq_\FF \lambda \sigma,\; \sigma \in \FF \rset.
\end{aligned}\end{equation}
Notice that this quantity combines $\Rmax$, which is an optimisation with respect to the positive semidefinite cone $\HH_+$, with $\Rmax^\FF$, which is defined using $\cone(\FF)$. It is not necessary to use $R_{\max}^\FF$ for both terms, and indeed this choice would trivialise the optimisation for all convex resource theories. The inclusion $\cone(\FF) \subseteq \HH_+$ entails that $\RW^\FF(\rho) \geq \RW(\rho)$ in general.

A geometric interpretation of the definition of the two measures is provided in Figure~\ref{fig:hilbert}.

\begin{figure}[t]
\centering
\includegraphics[width=10cm]{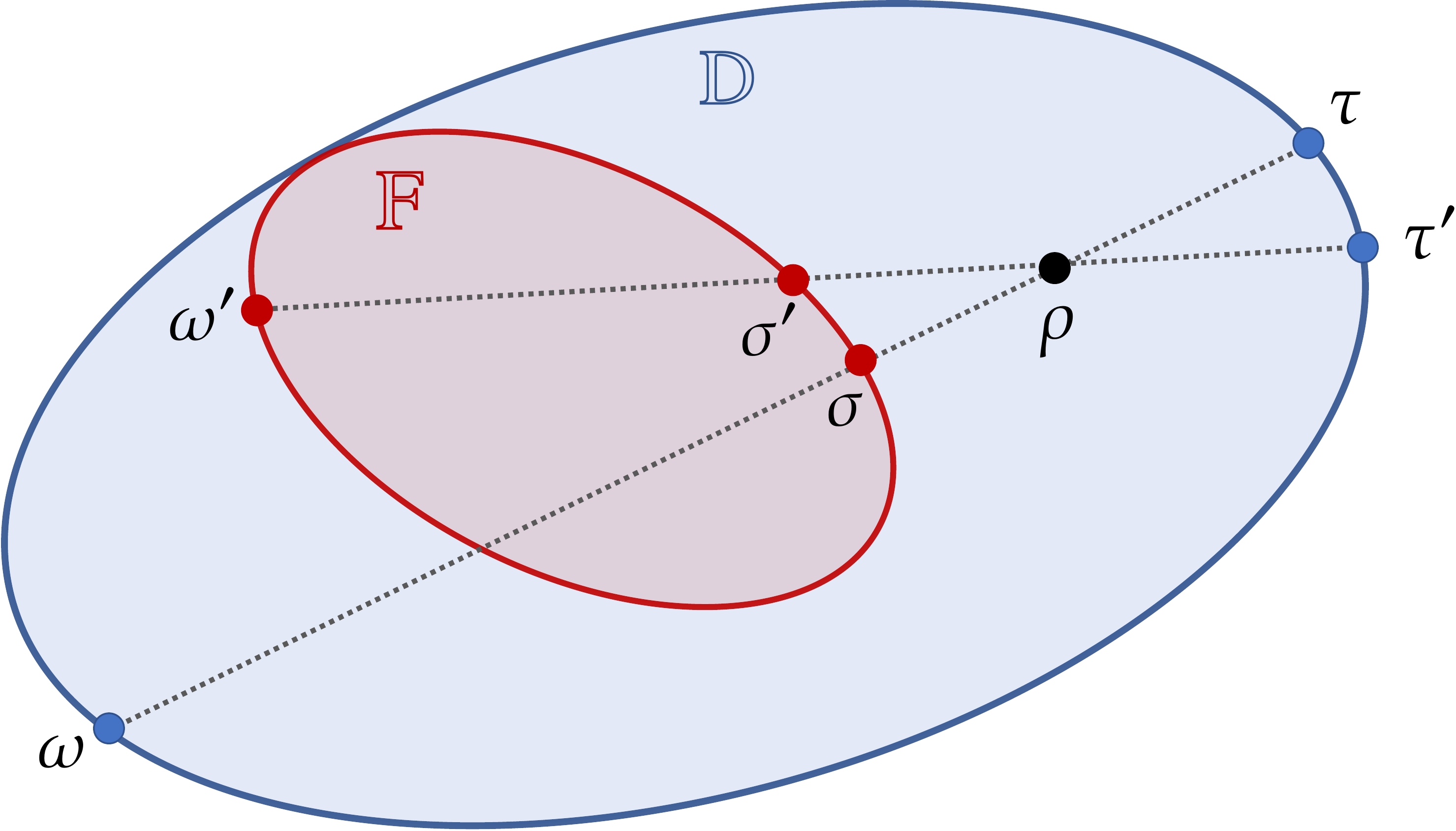}
\caption{\textbf{Geometric definition of the projective robustness measures.} Using the known geometric interpretation of the original Hilbert's projective metric~\cite{bushell_1973,kohlberg_1982}, the quantity $\RW(\rho)$ can be understood as the minimal cross-ratio of the points $\tau$, $\rho$, $\sigma$, and $\omega$ within the set of all density operators $\DD$, with $\sigma$ being a free state. Specifically, $\RW(\rho)$ is the minimal value of $\displaystyle \frac{\norm{\rho - \omega}{}}{\norm{\protect\vphantom{\rho}\sigma - \omega}{}}\frac{\norm{\sigma- \tau\protect\vphantom{\rho}}{}}{\norm{\rho - \tau}{}}$ such that $\tau, \omega \in \DD$, $\sigma \in \FF$, and the four states are all colinear as shown in the Figure. Here, the choice of the norm $\norm{\cdot}{}$ is arbitrary. This can be seen by noticing that for any feasible $\sigma$ such that $\nu \sigma \leq \rho \leq \lambda \sigma$ we can equivalently write $\rho + (\lambda - 1) \omega = \lambda \sigma$ and $\rho - (1-\nu) \tau = \nu \sigma$ for some states $\tau,\omega$; this then gives $\rho - \omega = \lambda (\sigma - \omega)$ and $\rho - \tau = \nu (\sigma - \tau)$, from which the cross-ratio form follows.
\newline The free projective robutness $\RW^\FF(\rho)$ is defined analogously as the least value of $\displaystyle \frac{\norm{\rho - \omega'}{}}{\norm{\protect\vphantom{\rho}\sigma' - \omega'}{}}\frac{\norm{\sigma' - \tau'\protect\vphantom{\rho}}{}}{\norm{\rho - \tau'}{}}$ where now both $\sigma'$ and $\omega'$ are constrained to be free states. It is not difficult to see that any optimal choice of such a decomposition will have all points $\tau, \sigma, \omega$ or $\tau', \sigma', \omega'$ lying on the boundary of the corresponding set ($\DD$ or $\FF$), since otherwise a better decomposition could be chosen along the same line.
}
\label{fig:hilbert}
\end{figure}

\subsection{Properties}

We begin with a study of the basic properties of the projective robustness $\RW$. For clarity, we for now only assume that $\FF$ is a topologically closed subset of density operators. 

\begin{boxed}{white}
\begin{theorem}\label{thm:proj_properties}
The projective robustness $\RW(\rho)$ of any state $\rho$ satisfies the following properties.
\begin{enumerate}[(i)]
\item\label{itm:supp} $\RW$ takes a finite value if and only if there exists a state $\sigma \in \FF$ such that $\supp \sigma = \supp \rho$. 
\item\label{itm:invariance} It is invariant under scaling, that is, $\RW(\lambda \rho) = \RW(\rho)$ for any $\lambda \in \RR_{++}$.
\item\label{itm:dual} $\RW$ can be computed as the optimal value of the optimisation problem
\begin{align}\label{eq:thm_convex_progA}
  \RW(\rho) &= \inf \lset \gamma \in \RR \bar \rho \leq \wt\sigma \leq \gamma \rho,\; \wt\sigma \in \cone(\FF) \rset,
  \end{align}
where the infimum is achieved as long as it is finite. When $\FF$ is a convex set, this is a conic linear optimisation problem, and it admits an equivalent dual formulation as
\begin{align}
  \RW(\rho) &= \sup \lset \< A, \rho \> \bar \< B, \rho \> = 1 ,\; B - A \in \cone(\FF)\*,\; A, B \in \HH_+ \rset\label{eq:thm_convex_progB}\\
  &= \sup \lset \frac{\< A, \rho \>}{\< B, \rho \>} \bar \frac{\< A, \sigma \>}{\< B, \sigma \>} \leq 1 \; \forall \sigma \in \FF,\; A, B \in \HH_+ \rset.\label{eq:thm_convex_progC}
\end{align}
Here, $\cone(\FF) = \lset \lambda \sigma \bar \lambda \in \RR_+, \; \sigma \in \FF \rset$ is the cone induced by the set $\FF$, and $\cone(\FF)\* = \lset X \in \HH \bar \< X, \sigma \> \geq 0 \; \forall \sigma \in \FF \rset$ is its dual cone.
\item\label{itm:faithful} $\RW$ is faithful, that is, for any state $\rho$ it holds that $\RW(\rho) = 1 \iff \rho \in \FF$.
\item\label{itm:submult} If the set of free states is closed under tensor product, in the sense that $\sigma_1, \sigma_2 \in \FF \Rightarrow \sigma_1 \otimes \sigma_2 \in \FF$, then $\RW$ is sub-multiplicative: for any states $\rho$ and $\omega$, it holds that
\begin{equation}\begin{aligned}
  \RW(\rho \otimes \omega) \leq \RW(\rho) \, \RW(\omega).
\end{aligned}\end{equation}
Alternatively, the weaker variant of sub-multiplicativity $\RW(\rho^{\otimes n}) \leq \RW(\rho)^n$ holds whenever $\sigma \in \FF \Rightarrow \sigma^{\otimes n} \in \FF \; \forall n \in \mathbb{N}$.
\item\label{itm:quasiconvex} When $\FF$ is a convex set, $\RW$ is quasiconvex: for any $t \in [0,1]$, it holds that 
\begin{equation}\begin{aligned}
  \RW(t \rho + (1-t) \omega) \leq \max \{ \RW(\rho) ,\, \RW(\omega) \}.
\end{aligned}\end{equation}
\item\label{itm:lsc} $\RW$ is lower semicontinuous, that is, $\displaystyle \RW(\rho) \leq \liminf_{n\to\infty}\, \RW(\rho_n)$ for any sequence $(\rho_n)_n$ of states converging to $\rho$. 
\item\label{itm:bounds} It can be bounded as
\begin{align}
  R_\FF(\rho) \, W_\FF(\rho)^{-1} \leq \;&  \RW(\rho) \leq R_\FF(\rho) \, \Rmax(\sigma^\star_R \| \rho)\\
  & \hphantom{\RW(\rho)} \leq R_\FF(\rho) \, \lambda_{\min}(\rho)^{-1},\nonumber\\
&   \RW(\rho) \leq W_\FF(\rho)^{-1} \, \Rmax(\rho \| \sigma^\star_W),
\end{align}
where $\lambda_{\min}(\rho)$ denotes the smallest eigenvalue of $\rho$, $\sigma^\star_R$ is an optimal state such that $R_\FF(\rho) = \Rmax(\rho \| \sigma^\star_R)$, and analogously $\sigma^\star_W$ is an optimal state such that $W_\FF(\rho) = \Rmax(\sigma^\star_W \| \rho)^{-1}$, whenever such states exist.
\end{enumerate}
\end{theorem}
\end{boxed}

Let us comment on the unusual property~\eqref{itm:supp} of $\RW$: it always diverges for states which do not share support with a free state. There are two immediate consequences of this fact: (1) as long as $\FF$ contains a state of full rank (which is the case in virtually all theories of interest), then $\RW(\rho) < \infty$ for all full-rank states; (2) for any pure state $\psi \notin \FF$, we have that $\RW(\psi) = \infty$. Although the latter point might make the applicability of $\RW$ seem limited, we will see in Section~\ref{sec:prob_dist_RW} that it can be easily circumvented by allowing a small error in the conversion, making $\RW$ useful also in the study of transformations involving pure states.
\begin{proof}
\begin{enumerate}[(i)]
\item A necessary and sufficient condition for $\Rmax(\rho \| \sigma) < \infty$ can be seen to be $\supp \rho \subseteq \supp \sigma$~\cite{douglas_1966}. Applying the same reasoning to $\Rmax(\sigma \| \rho)$ then means that the supports must be equal.

\item Follows from the fact that $\Rmax(\lambda \rho \| \mu \sigma) = \lambda^{-1} \mu \Rmax(\rho \| \sigma)$, which can be easily verified from the definition of $\Rmax$.

\item From the positivity of $\Rmax(\rho\|\sigma)$ for any states $\rho$ and $\sigma$, we can see that
\begin{align}
  \min_{\sigma \in \F} \Pmax(\rho\| \sigma) \, \Pmax(\sigma \| \rho) &= \min_{\sigma \in \FF} \left[ \inf \lset \lambda \bar \rho \leq \lambda \sigma \rset \inf \lset \mu \bar \sigma \leq \mu \rho \rset \right]\\
  &= \inf \lset \lambda \mu \bar \rho \leq \lambda \sigma,\; \sigma \leq \mu \rho,\; \sigma \in \F \rset\label{eq:dmaxdmax}.
\end{align}
Observe that any feasible solution to the problem 
\begin{equation}\begin{aligned}
\inf \lset \gamma \bar \rho \leq \wt\sigma \leq \gamma \rho,\; \wt\sigma \in \cone(\FF) \rset
\end{aligned}\end{equation}
gives a feasible solution to Eq.~\eqref{eq:dmaxdmax} as
\begin{equation}\begin{aligned}\label{eq:proj_sdp}
  \sigma = \frac{\wt\sigma}{\Tr \wt\sigma},\quad \lambda = \Tr \wt\sigma,\quad \mu = \frac{\gamma}{\Tr \wt\sigma}
\end{aligned}\end{equation}
with objective function value $\lambda \mu = \gamma$. Conversely, any feasible solution $\{\sigma,\lambda,\mu\}$ to Eq.~\eqref{eq:dmaxdmax} gives a feasible solution to Eq.~\eqref{eq:proj_sdp} as $\wt\sigma = \lambda \sigma$, $\gamma = \lambda \mu$. The two problems are therefore equivalent.

To show that the infimum is achieved when $\RW(\rho) < \infty$, take any optimal sequence $(\gamma_n)_n \in \RR_+$ such that $\gamma_n \tends \RW(\rho)$ and $\rho \leq \wt\sigma_n \leq \gamma_n \rho$. Since $\norm{\wt\sigma_n}{\infty} \leq \norm{\gamma_n \rho}{\infty} \leq \gamma_n$, the sequence $(\wt\sigma_n)_n$ is bounded, and so by the Bolzano–Weierstrass theorem it has a convergent subsequence $(\wt\sigma_{n_k})_k \tendsk \wt\sigma$. The closedness of $\cone(\FF)$ (itself a consequence of the assumed closedness of $\FF$) ensures that $\wt\sigma \in \cone(\FF)$. But then $(\wt\sigma_{n_k} - \rho)_k$ and $(\gamma_{n_k} \rho - \wt\sigma_{n_k})_k$ also converge as sums of convergent sequences, and by the closedness of $\HH_+$ they must converge to a positive semidefinite operator. Hence $\rho \leq \wt\sigma \leq \RW(\rho) \rho$ as desired.

The dual form (Eq.~\eqref{eq:thm_convex_progB}) is obtained through standard Lagrange duality arguments~\cite{ponstein_2004}. The fact that strong duality (equality between the primal and the dual) holds can be ensured by Slater's theorem~\cite[Thm.~28.2]{rockafellar_1970}, applicable since the choice of $B = \id$, $A = \ve \id$ for $\ve \in (0,1)$ is strongly feasible for the dual.

Eq.~\eqref{eq:thm_convex_progC} follows since any feasible solution to this program can be rescaled as $A \mapsto A / \< B, \rho \>$, $B \mapsto B / \< B, \rho \>$ to give a feasible solution to the dual, and vice versa. 
Here, we implicitly constrain ourselves to $B$ such that $\< B, \rho \> \neq 0$ and $\<B ,\sigma \> \neq 0 \; \forall \sigma \in \FF$ to make the expressions well defined; this can always be ensured by perturbing $B$ by a small multiple of the identity.

\item If $\rho \in \FF$, then $\RW(\rho) \leq \Rmax(\rho\|\rho)^2 = 1$, and this is the least value that $\RW$ can take for any state. Conversely, if $\RW(\rho) = 1$, then by~\eqref{itm:dual} we have that there exists a state $\sigma \in \FF$ such that $\sigma \leq \rho \leq \sigma$, which is only possible if $\rho = \sigma$.

\item If either of $\RW(\rho)$ or $\RW(\omega)$ is infinite, then the relation is trivial, so assume otherwise. Let  $\sigma \in \FF$ be an optimal state such that $\rho \leq \lambda \sigma$ and $\sigma \leq \mu \rho$ with $\lambda \mu = \RW(\rho)$, and analogously let $\sigma' \in \FF$ be such that $\omega \leq \lambda' \sigma'$ and $\sigma' \leq \mu' \omega$ with $\lambda'\mu' = \RW(\omega)$, which all exist by~\eqref{itm:dual}. Then
\begin{equation}\begin{aligned}
0 \leq \rho \otimes ( \lambda' \sigma' - \omega) + (\lambda \sigma - \rho) \otimes \lambda' \sigma' = (\lambda \sigma) \otimes (\lambda' \sigma') - \rho \otimes \omega
\end{aligned}\end{equation}
and analogously $(\lambda \sigma) \otimes (\lambda' \sigma') \leq (\lambda \mu \rho) \otimes (\lambda' \mu' \omega)$, which shows that $\sigma \otimes \sigma'$ is a feasible solution for the projective robustness of $\rho \otimes \omega$, yielding $\RW(\rho \otimes \omega) \leq \lambda \mu \lambda' \mu' = \RW(\rho) \, \RW(\omega)$. Weak sub-multiplicativity can be shown analogously, first choosing $\omega = \rho$ and then by induction on~$n$.

\item As in~\eqref{itm:submult}, we take $\sigma \in \FF$ such that $\rho \leq \lambda \sigma$ and $\sigma \leq \mu \rho$, and $\sigma' \in \FF$ such that $\omega \leq \lambda' \sigma'$ and $\sigma' \leq \mu' \omega$. We then have that
\begin{equation}\begin{aligned}
  t \rho + (1-t) \omega &\leq t \lambda \sigma + (1-t) \lambda' \sigma' = \left[t \lambda + (1-t) \lambda'\right] \frac{t \lambda \sigma + (1-t) \lambda' \sigma'}{t \lambda + (1-t) \lambda'}
\end{aligned}\end{equation}
where $\displaystyle \frac{t \lambda \sigma + (1-t) \lambda' \sigma'}{t \lambda + (1-t) \lambda'} \in \FF$ by convexity of $\FF$. Now,
\begin{equation}\begin{aligned}
  \frac{t \lambda \sigma + (1-t) \lambda' \sigma'}{t \lambda + (1-t) \lambda'} &\leq \frac{t \lambda \mu \rho + (1-t) \lambda' \mu' \omega}{t \lambda + (1-t) \lambda'} \leq \frac{\max\{ \lambda \mu, \lambda'\mu' \} \left[ t \rho + (1-t) \omega\right]}{t \lambda + (1-t) \lambda'}.
\end{aligned}\end{equation}
Using this state as a feasible solution for $\RW$ gives
\begin{equation}\begin{aligned}
  \RW(t \rho + (1-t) \omega) &\leq \left[t \lambda + (1-t) \lambda'\right] \frac{\max\{ \lambda \mu, \lambda'\mu' \}}{t \lambda + (1-t) \lambda'} = \max \{ \lambda \mu, \lambda' \mu' \}
\end{aligned}\end{equation}
as was to be shown.

\item Establishing lower semicontinuity is equivalent to showing that the sublevel sets $\lset \omega \bar \RW(\omega) \leq \gamma \rset$ are closed for all $\gamma \in \RR$~\cite[Thm.~7.1]{rockafellar_1970}. Consider then a sequence $\omega_n \tends \omega$ such that $\RW(\omega_n) \leq \gamma \; \forall n$ for some $\gamma$. By~\eqref{itm:dual}, this entails that there exists $\wt\sigma_n \in \cone(\FF)$ such that $\omega_n \leq \wt\sigma_n \leq \gamma \omega_n$ for each $n$. Since $(\wt\sigma_n)_n$ then forms a bounded sequence, we can assume that it converges as $\wt\sigma_n \tends \wt\sigma$, up to passing to a subsequence.
 The closedness of $\cone(\FF)$ ensures that $\wt\sigma \in \cone(\FF)$. The convergent sequences $(\wt\sigma_n - \omega_n)_n$ and $(\gamma \omega_n - \wt\sigma_n)_n$ then must converge to positive semidefinite operators by the closedness of the positive semidefinite cone. This gives $\omega \leq \wt\sigma \leq \gamma \omega$, showing that the sublevel set of $\RW$ is closed as desired. Since $\gamma$ was arbitrary, the desired statement follows.

\item The lower bound is obtained by noting that
\begin{equation}\begin{aligned}
  \RW(\rho) &= \min_{\sigma \in \FF}\, \Rmax(\rho \| \sigma) \, \Rmax(\sigma \| \rho)\\
  &\geq \left[ \min_{\sigma \in \FF} \Rmax(\rho \| \sigma) \right] \left[ \min_{\sigma \in \FF} \Rmax(\sigma \| \rho) \right]\\
  &= R_\FF(\rho) \, W_\FF(\rho)^{-1}.
\end{aligned}\end{equation}
The upper bounds follow by using $\sigma^\star_R$ and $\sigma^\star_W$ as feasible solutions in the definition of $\RW$. The relation with the smallest eigenvalue follows from the fact that
\begin{equation}\begin{aligned}
  \RW(\rho) &\leq R_\FF(\rho) \, \Rmax(\sigma^\star_R \| \rho)\\
  &\leq R_\FF(\rho) \, \max_{\omega \in \DD} \Rmax(\omega \| \rho)\\
  &= R_\FF(\rho)\,\lambda_{\min}(\rho)^{-1},
\end{aligned}\end{equation}
where we relaxed the optimisation to all density matrices $\DD$ and used Lemma~\ref{lem:weight_lambda} below.
\end{enumerate}
\end{proof}

We employed the following auxiliary lemma, first shown in \cite{haapasalo_2014}.
\begin{lemma}\label{lem:weight_lambda}
For any state $\rho$, it holds that $\displaystyle \max_{\omega \in \DD} \Rmax(\omega \| \rho) = \lambda_{\min}(\rho)^{-1}$.
\end{lemma}
\begin{proof}
If $\rho$ is not full rank, we can choose $\omega$ supported on $\ker(\rho)$ so that $\Rmax(\omega\|\rho) = \infty = \lambda_{\min}(\rho)^{-1}$. Assuming then that $\rho$ is full rank, for any $\omega \in \DD$ it holds that
\begin{equation}\begin{aligned}
  \omega &\leq \lambda_{\max}(\omega) \,\id \\
  &\leq \lambda_{\max}(\omega) \lambda_{\min}(\rho)^{-1} \,\rho
\end{aligned}\end{equation}
where we used that $\lambda_{\max}(\omega) = \min \lset \lambda \bar \omega \leq \lambda \id \rset$ and  $\lambda_{\min}(\rho) = \max \lset \lambda \bar \rho \geq \lambda \id \rset$. This constitutes a feasible solution to $\Rmax(\omega \| \rho)$, yielding
\begin{equation}\begin{aligned}
  \Rmax(\omega \| \rho) \leq \lambda_{\max}(\omega) \lambda_{\min}(\rho)^{-1} \leq \lambda_{\min}(\rho)^{-1}.
\end{aligned}\end{equation}
The proof is concluded by noticing that the bound is saturated by the choice of $\omega = \proj\psi$ with $\ket\psi$ being an eigenvector of $\rho$ corresponding to $\lambda_{\min}(\rho)$.
\end{proof}

The quantity $\RW^\FF$, defined in Eq.~\eqref{eq:proj_free_def}, obeys a similar set of properties. The proof is obtained in complete analogy with the above.
\begin{boxed}{white}
\begin{corollary}\label{cor:projf_properties}
The free projective robustness $\RW^\FF(\rho)$ satisfies the following properties.
\begin{enumerate}[(i)]
\item It is invariant under scaling, that is, $\RW^\FF(\lambda \rho) = \RW^\FF(\rho)$ for any $\lambda > 0$.
\item When $\FF$ is a convex set, $\RW^\FF$ can be computed as the optimal value of a conic linear optimisation problem:
\begin{align}\label{eq:thm_convex_progAfree}
  \RW^\FF(\rho) &= \inf \lset \gamma \in \RR \bar \rho \leq_\FF \wt\sigma \leq \gamma \rho,\; \wt\sigma \in \cone(\FF) \rset\\
  &= \sup \lset \< A, \rho \> \bar \< B, \rho \> = 1 ,\; B - A \in \cone(\FF)\*,\; A \in \cone(\FF)\*,\; B \in \HH_+ \rset\label{eq:thm_convex_progBfree}\\
  &= \sup \lset \frac{\< A, \rho \>}{\< B, \rho \>} \bar \frac{\< A, \sigma \>}{\< B, \sigma \>} \leq 1 \; \forall \sigma \in \FF,\; A \in \cone(\FF)\*,\; B \in \HH_+ \rset.\label{eq:thm_convex_progCfree}
\end{align}
Furthermore, the infimum in \eqref{eq:thm_convex_progAfree} is achieved as long as it is finite.
\item $\RW^\FF$ is faithful, that is, $\RW^\FF(\rho) = 1 \iff \rho \in \FF$.
\item When $\FF$ is a convex set, $\RW^\FF$ is quasiconvex: $\RW^\FF(t \rho + (1-t) \omega) \leq \max \{ \RW^\FF(\rho) ,\, \RW^\FF(\omega) \}$ for any $t \in [0,1]$.
\item $\RW^\FF$ is lower semicontinuous, that is, $\displaystyle \RW^\FF(\rho) \leq \liminf_{n\to\infty}\, \RW^\FF(\rho_n)$ for any sequence $(\rho_n)_n$ which converges to $\rho$.
\item It can be bounded as
\begin{align}
  R^\FF_\FF(\rho) \, W_\FF(\rho)^{-1} \leq \;&  \RW^\FF(\rho) \leq R_\FF^\FF(\rho) \, \Rmax(\sigma^\star_R \| \rho),\\
  & \hphantom{\RW^\FF(\rho)} \leq R_\FF^\FF(\rho) \, \lambda_{\min}(\rho)^{-1}\nonumber\\
&   \RW^\FF(\rho) \leq W_\FF(\rho)^{-1} \, \Rmax^\FF(\rho \| \sigma^\star_W),
\end{align}
where $\sigma^\star_R$ is an optimal state such that $R^\FF_\FF(\rho) = \Rmax^\FF(\rho \| \sigma^\star_R)$, and analogously $\sigma^\star_W$ is an optimal state such that $W_\FF(\rho) = \Rmax(\sigma^\star_W \| \rho)^{-1}$, whenever such states exist.
\end{enumerate}
\end{corollary}
\end{boxed}

One could also establish the condition under which $\RW^\FF$ is finite, although it does not take a straightforward form as in the case of $\RW$. Specifically, we have that $\RW^\FF(\rho)$ is finite if and only if there exists a state $\sigma \in \FF$ such that $\supp \sigma \subseteq \supp \rho$ and $\rho + \wt\tau \propto \sigma$ for some $\wt\tau \in \cone(\FF)$.
For our purposes, it will suffice to make note of the fact that $\RW^\FF(\rho)\geq\RW(\rho)$, meaning that diverging projective robustness implies that the free projective robustness also must diverge.

\subsubsection*{Computability}

We have seen in Theorem~\ref{thm:proj_properties} and Corollary~\ref{cor:projf_properties} that, in all convex resource theories, $\RW$ and $\RW^\FF$ are both given as the optimal values of convex optimisation problems.
As long as the constraint $\sigma \in \cone(\FF)$ (or, equivalently, $X \in \cone(\FF)\*$) can be expressed using linear matrix inequalities, the two optimisation problems are computable as semidefinite programs (SDP).

An example is the resource theory of coherence, where $\FF$ denotes the set of all diagonal (incoherent) states in a given basis and
\begin{equation}\begin{aligned}
  \RW(\rho) = \inf \lset \gamma \bar \rho \leq \wt{\sigma} \leq \gamma \rho,\; \wt\sigma \geq 0,\; \wt\sigma = \Delta(\wt\sigma) \rset
\end{aligned}\end{equation}
with $\Delta(\cdot) \coloneqq \sum_i \proj{i} \cdot \proj{i}$ denoting the completely dephasing channel (diagonal map) in the incoherent orthonormal basis $\{\ket{i}\}$. A very similar expression can be obtained also, for instance, in the theory of imaginarity, where the free states $\FF$ are those that only have real coefficients in a given basis --- here, we simply replace $\Delta$ with the map $X \mapsto \Re X$. An SDP expression is also immediate when the set $\FF$ consists only of a single state --- e.g.\ in the theories of athermality (thermodynamics) and purity --- or when $\FF$ is a polytope, that is, a convex combination of a finite number of states, which is the case e.g.\ in the resource theory of magic.

Another important case is the resource theory of non-positive partial transpose, in which case $\FF$ consists of bipartite states whose partial transpose is positive, and we have the SDP expressions
\begin{equation}\begin{aligned}
  \RW(\rho) &= \inf \lset \gamma \bar \rho \leq \wt{\sigma} \leq \gamma \rho,\; \wt\sigma \geq 0,\; \wt\sigma^\Gamma \geq 0 \rset\\
  \RW^\FF(\rho) &=  \inf \lset \gamma \bar \rho \leq \wt\sigma \leq \gamma \rho,\; \rho^\Gamma \leq \wt\sigma^\Gamma,\; \wt\sigma \geq 0,\; \wt\sigma^\Gamma \geq 0 \rset
\end{aligned}\end{equation}
with $(\cdot)^\Gamma$ denoting partial transposition in any chosen basis.

\subsection{Necessary condition for probabilistic transformations}

The usefulness of the projective robustness measures stems from their exceptionally strong monotonicity, which can be used to establish much stronger limitations than commonly encountered monotones. Recall that, in probabilistic transformations of resources, monotonicity is typically required to hold on average. $\RW$ and $\RW^\FF$, however, cannot be increased even if we postselect on a particular outcome of a probabilistic transformation.

\begin{theorem}[\cite{regula_2022}]\label{thm:nogo_monotonicity}
If there exists a probabilistic free transformation $\rho \transf \rho'$, then
\begin{equation}\begin{aligned}
  \RW(\rho) \geq \RW(\rho')
\end{aligned}\end{equation}
and
\begin{equation}\begin{aligned}
  \RW^\FF(\rho) \geq \RW^\FF(\rho').
\end{aligned}\end{equation}
\end{theorem}
\begin{proof}
If there is no $\sigma \in \FF$ such that $\supp \rho = \supp \sigma$, then $\RW(\rho) = \infty$ and the result is trivial, so we shall assume otherwise. Then let $\sigma \in \FF$ be a state such that $\rho \leq \lambda \sigma$ and $\sigma \leq \mu \rho$ with $\RW(\rho) = \lambda\mu$. Assume that there exists a sequence $(\E_n)_n \in \OO$ realising the transformation $\rho \transf \rho'$. Notice that it necessarily holds that $\Tr \E_n(\sigma) > 0 \; \forall n$, since $\Tr \E_n(\sigma) \geq \frac{1}{\lambda}\Tr \E_n(\rho)$. We then have
\begin{equation}\begin{aligned}\label{eq:proof_mono1}
  \frac{\E_n(\rho)}{\Tr \E_n(\rho)} &\leq \lambda \frac{\Tr \E_n(\sigma)}{\Tr \E_n(\rho)} \frac{\E_n(\sigma)}{\Tr \E_n(\sigma)}
\end{aligned}\end{equation}
and
\begin{equation}\begin{aligned}\label{eq:proof_mono2}
  \frac{\E_n(\sigma)}{\Tr \E_n(\sigma)} &\leq \mu \frac{\Tr \E_n(\rho)}{\Tr \E_n(\sigma)} \frac{\E_n(\rho)}{\Tr \E_n(\rho)}
\end{aligned}\end{equation}
using the positivity of $\E_n$. 
Since each $\E_n$ is a free operation, $\frac{\E_n(\sigma)}{\Tr \E_n(\sigma)} \in \FF$. 
This gives that
\begin{equation}\begin{aligned}
  \RW\!\left(\frac{\E_n(\rho)}{\Tr \E_n(\rho)}\right) \leq \left[\lambda \frac{\Tr \E_n(\sigma)}{\Tr \E_n(\rho)}\right]\! \left[\mu \frac{\Tr \E_n(\rho)}{\Tr \E_n(\sigma)} \right] \!= \lambda \mu = \RW(\rho).
\end{aligned}\end{equation}
As $\frac{\E_n(\rho)}{\Tr \E_n(\rho)} \tends \rho'$ by hypothesis, lower semicontinuity of $\RW$ (Theorem~\ref{thm:proj_properties}\eqref{itm:lsc}) implies that $\RW(\rho') \leq \RW(\rho)$ as desired.

The proof for $\RW^\FF$ is completely analogous, using the fact that any $\E_n \in \OO$ preserves the free cone $\cone(\FF)$.
\end{proof}
This yields a general no-go result: it is impossible to transform any state $\rho$ into a state $\rho'$ such that $\RW(\rho) < \RW(\rho')$ or $\RW^\FF(\rho) < \RW^\FF(\rho')$, even probabilistically.

As an immediate consequence, we see that no state $\rho$ with $\RW(\rho) < \infty$ can be converted into a state $\rho'$ with $\RW(\rho') = \infty$, where the latter type includes in particular all resourceful pure states. 
Furthermore, the submultiplicativity of $\RW$ ensures that $\RW(\rho^{\otimes n})$ is finite for all $n$, which means that the transformation $\rho^{\otimes n} \transf \rho'$ is impossible for any number of copies of $\rho$. In resource theories where $\FF$ contains a state of full rank, this precludes the transformation of any full-rank state into a non-free pure state, recovering a result of Ref.~\cite{fang_2020}. However, the result here is strictly stronger, as it extends to any transformation where the target state has a diverging projective robustness, which goes beyond pure states.
 For instance, in the theory of quantum entanglement, rank-deficient states $\rho'$ will have $\RW(\rho') = \infty$ unless their support can be spanned by product vectors. Our findings thus extend previous no-go results that dealt with the impossibility of entanglement purification~\cite{kent_1998,jane_2002,horodecki_1999-1,horodecki_2006-1,regula_2019-2}.
Similarly, for quantum coherence, it holds that $\RW(\rho') = \infty$ whenever the support of $\rho'$ cannot be spanned by incoherent vectors, which also allows us to generalise known restrictions on coherence distillation~\cite{fang_2018,wu_2020}.

We stress that Theorem~\ref{thm:nogo_monotonicity} does not require the convexity of the set $\FF$ and thus holds also for more general types of quantum resource theories, with closedness of $\FF$ being the sole requirement. A weaker version of the result can be shown even without assuming closedness: in this case, we would no longer be able to rely on the lower semicontinuity of $\RW$, but $\RW$ would remain a monotone for all probabilistic protocols that succeed with a non-zero probability $p$ (i.e., non-asymptotic ones).

We also note that we did not actually use the \emph{complete} positivity of operations in $\OO$; the statement of Theorem~\ref{thm:nogo_monotonicity} applies also when $\OO$ consists of maps that are merely positive.

\subsection{Sufficient condition for probabilistic transformations}

In order to show that the non-increase of the projective robustness can serve as a \emph{sufficient} condition for resource convertibility, we will need to consider different types of resource theories separately.

Hereafter in the manuscript, we make the assumption that $\FF$ is convex; as seen in Theorem~\ref{thm:proj_properties}, this is required to make $\RW$ and $\RW^\FF$ convex optimisation problems, the strong duality of which will be a crucial ingredient of our proofs. We define the following classes of quantum resource theories~\cite{gour_2017,regula_2020}:
\begin{enumerate}[(i)]
\item \emph{Affine resource theories} are such that the set $\FF$ is the intersection of some affine subspace with the set of density matrices. In other words, $\FF = \aff(\FF) \cap \DD$, where $\aff$ denotes the affine hull
\begin{equation}\begin{aligned}
  \aff(\FF) = \lset \sum_i c_i \sigma_i \bar \sigma_i \in \FF,\; c_i \in \RR,\; \sum_i c_i = 1 \rset.
\end{aligned}\end{equation}
Examples include the theories of quantum coherence~\cite{baumgratz_2014}, asymmetry~\cite{gour_2008}, thermodynamics~\cite{horodecki_2013}, or imaginarity~\cite{hickey_2018,wu_2021}.
\item \emph{Full-dimensional resource theories} are those in which $\sspan(\FF) = \HH$; in other words, $\FF$ has a non-zero volume as a subset of $\DD$. Examples are the theories of quantum entanglement~\cite{horodecki_2009}, non-positive partial transpose~\cite{horodecki_2009}, or magic (non-stabiliserness)~\cite{veitch_2014,howard_2017}.
\end{enumerate}
The two classes together encompass most known quantum resource theories of relevance.

\subsubsection*{Affine resource theories}

We begin with a characterisation of affine resource theories.

\begin{boxed}{white}
{%
\begin{theorem}\label{thm:nogo_affine_full}
In any affine resource theory, there exists a resource--non-generating probabilistic transformation $\rho \transf \rho'$ if any of the following conditions is satisfied:
\begin{enumerate}[(i)]
\item $\infty > \RW(\rho) \geq \RW(\rho')$,
\item $\infty = \RW(\rho)$ and $R_\F(\rho') < \infty$,
\item $R_\F(\rho) = \infty$.
\end{enumerate}
\end{theorem}
}
\end{boxed}

Let us clarify this result first. As long as $R_\FF(\rho') < \infty$ holds in the given affine resource theory, then conditions (i) and (ii) together with Theorem~\ref{thm:nogo_monotonicity} imply that
\begin{equation}\begin{aligned}
  \rho \transf \rho' \;\iff\; \RW(\rho) \geq \RW(\rho'),
\end{aligned}\end{equation}
as announced in \cite{regula_2022}. As we remarked before, the condition $R_\FF(\rho') < \infty$ is actually satisfied for \emph{all} states in virtually every theory of interest, and we only consider it here explicitly to account for pathological cases.

An important aspect of the above result is that the transformation $\rho \transf \rho'$ might only be possible with asymptotically vanishing probability. Should this be undesirable, one can also study the problem of when there exists a transformation which succeeds with a non-zero probability. We will give a concise sufficient condition in Lemma~\ref{lem:achiev_omega} shortly: the existence of a transformation with a non-zero probability of success can be guaranteed whenever $\rho$ is full rank or a pure state. 

We begin with a useful lemma based on the methods of~\cite{regula_2020}.
\begin{lemma}\label{lem:affine}
In any affine resource theory, the projective robustness as well as its reciprocal can be expressed as
\begin{equation}\begin{aligned}
  \Omega_{\FF}(\rho) &= \sup \lset \< A, \rho \> \bar \< B, \rho \> = 1 ,\; \< B - A, \sigma \> = 0 \; \forall \sigma \in \FF,\; A, B \geq 0 \rset,\\
  \Omega_{\FF}(\rho)^{-1} &= \inf \lset \< A, \rho \> \bar \< B, \rho \> = 1 ,\; \< B - A, \sigma \> = 0 \; \forall \sigma \in \FF,\; A, B \geq 0 \rset.
\end{aligned}\end{equation}
\end{lemma}
\begin{proof}
Let us define $\Omega_{\aff(\FF)}$ as the projective robustness defined with respect to the set $\aff(\FF)$, that is,
\begin{equation}\begin{aligned}\label{eq:proj_aff}
  \Omega_{\aff(\FF)} (\rho) = \inf_{ Z \in \aff(\FF)} \Rmax(\rho \| Z) \, \Rmax(Z \| \rho).
\end{aligned}\end{equation}
Importantly, for density operators, this makes no quantitative difference: any feasible $X \in \aff(\FF)$ which satisfies $\rho \leq \lambda X$ is necessarily positive, which means that all feasible $X$ actually belong to $\FF$, and so $\Omega_{\aff(\FF)} (\rho) = \RW(\rho)$. However, by taking the dual of the program \eqref{eq:proj_aff} as in the proof of Theorem~\ref{thm:proj_properties}, we get
\begin{equation}\begin{aligned}
  \Omega_{\aff(\FF)} (\rho) &= \sup \lset \< A, \rho \> \bar \< B, \rho \> = 1 ,\; B - A \in \aff(\FF)\*,\; A, B \geq 0 \rset.
\end{aligned}\end{equation}
The dual cone $\aff(\FF)\* \coloneqq \lset X \bar \< X, Z \> \geq 0 \; \forall Z \in \aff(\FF) \rset$ can be shown to consist of operators $X$ such that $\<X, Z \>$ is constant on the set $\aff(\FF)$, that is, $\< X, Z \> = k$ for all $Z \in \aff(\FF)$ for some $k \in \RR_+$. This can be seen as follows. Assume that $X \in \aff(\FF)\*$ but $X$ is not constant on $\aff(\FF)$, that is, there exists $Z, Z' \in \aff(\FF)$ such that $\< X, Z \> = k$ but $\<X , Z' \> = k'$ with $k > k'$. Let us then define $\wt{Z} = t Z' + (1-t) Z$ for some $t \in \RR$. Clearly, any such $\wt{Z}$ belongs to $\aff(\FF)$ by definition. But, choosing $t > \frac{k}{k - k'}$, we get
\begin{equation}\begin{aligned}
  \< \wt{Z}, X \> &= t k' + (1-t) k\\
  &= t (k' - k) + k\\
  &< \frac{k}{k - k'} (k'-k) + k\\
  &= 0,
\end{aligned}\end{equation}
which contradicts the assumption that $X \in \aff(\FF)\*$. Thus, we have that
\begin{equation}\begin{aligned}
  \Omega_{\aff(\FF)} (\rho) &= \sup \lset \< A, \rho \> \bar \< B, \rho \> = 1 ,\; \exists k \in \RR_+ \text{ s.t.\ } \< B - A, Z \> = k \; \forall Z \in \aff(\FF),\; A, B \geq 0 \rset.
\end{aligned}\end{equation}
Notice first that it suffices to impose that $ \< B - A, \sigma \> = k \; \forall \sigma \in \FF$, since this is equivalent to $\< B - A, Z \> = k \; \forall Z \in \aff(\FF)$. 
Now, let $A, B$ be feasible solutions such that $\< B - A, \sigma \> = k \; \forall \sigma \in \FF$ for some $k>0$. Then $A' = A + k\id$ and $B$ are also feasible, but with a strictly larger objective function value. Therefore, any optimal choice of $A,B$ (or an optimal sequence of such operators) must satisfy $k=0$.
Putting everything together, we have shown that, for any state $\rho$,
\begin{equation}\begin{aligned}
  \Omega_{\FF}(\rho) = \Omega_{\aff(\FF)} (\rho) &= \sup \lset \< A, \rho \> \bar \< B, \rho \> = 1 ,\; \< B - A, \sigma \> = 0 \; \forall \sigma \in \FF,\; A, B \geq 0 \rset
\end{aligned}\end{equation}
which is precisely the desired form.

For the reciprocal, we notice that $\Rmax(\rho\|\sigma)^{-1} = \sup \lset \lambda \bar \lambda \rho \leq \sigma \rset$. This allows us to express
\begin{equation}\begin{aligned}
  \Omega_{\aff(\FF)} (\rho)^{-1} &= \sup \lset \lambda \mu \bar \lambda \rho \leq X,\; \mu X \leq \rho,\; X \in \aff(\FF) \rset\\
  &= \inf \lset \< A, \rho \> \bar \< B, \rho \> = 1 ,\;  A - B \in \aff(\FF)\*,\; A, B \geq 0 \rset.
\end{aligned}\end{equation}
An analogous series of arguments leads us to the expression in the statement of the Lemma.
\end{proof}

The main proof can be now shown by generalising an approach of Ref.~\cite{reeb_2011} (cf.~\cite{buscemi_2017}).

\begin{proof}[\textbf{\emph{Proof of Theorem~\ref{thm:nogo_affine_full}}}]
(i) Assume that $\Omega_\FF(\rho) > 1$, since otherwise both states are free and the result is trivial. By hypothesis, we have that there exists a state $\sigma \in \FF$ such that $\rho \leq \lambda \sigma$ and $\sigma \leq \mu \rho$ with $\lambda \mu = \RW(\rho)$, and analogously there exists $\sigma' \in \FF$ such that $\rho' \leq \lambda' \sigma'$ and $\sigma' \leq \mu' \rho'$ with $\lambda'\mu' = \RW(\rho')$.  
Since $\lambda\mu \geq \lambda'\mu'$ by assumption, we can always choose a positive factor $k$ such that $k \rho' \leq \lambda \sigma'$ and $\sigma' \leq \mu k \rho'$. To see this, notice that this is trivial if $\lambda \geq \lambda'$ and $\mu \geq \mu'$; if $\mu < \mu'$, then we take $k = \frac{\lambda}{\lambda'}$; and in the case that $\lambda < \lambda'$, it suffices to take $k = \frac{\mu'}{\mu}$.
Denoting $\rho'' \coloneqq k \rho'$, define the sequence of maps
\begin{equation}\begin{aligned}\label{eq:measure_prep_1}
  \E_n(X) \coloneqq \< B_n, X \> \left(\lambda \sigma' - \rho''\right) + \<A_n , X \> \left( \rho'' - \frac{1}{\mu} \sigma'\right),
\end{aligned}\end{equation}
where $(A_n)_n$ and $(B_n)_n$ are optimal sequences of feasible operators for $\RW(\rho)$ in Lemma~\ref{lem:affine}, that is, ones such that $\<A_n, \rho \>\tends \lambda\mu$. 
Because the Choi operator of such a map is given by
\begin{equation}\begin{aligned}
  J_{\E_n} = B_n \otimes \left(\lambda \sigma' - \rho''\right) + A_n \otimes \left( \rho'' - \frac{1}{\mu} \sigma'\right),
\end{aligned}\end{equation}
the positive semidefiniteness of all involved operators ensures that $J_{\E_n} \geq 0$ and thus that the map $\E_n$ is completely positive~\cite{choi_1975}.
For any $\pi \in \FF$, the map furthermore satisfies
  \begin{equation}\begin{aligned}
    \E_n(\pi) = \left(\lambda - \frac{1}{\mu}\right) \< B_n, \pi \> \sigma' \propto \sigma' \in \FF,
  \end{aligned}\end{equation}
  which means that it is a free operation.
Crucially, it holds that
\begin{equation}\begin{aligned}
  \E_n(\rho) = \underbrace{\frac{1}{\mu}}_{\in (0,1]} \underbrace{\vphantom{\frac{1}{\mu}}\left( \lambda \mu - \< A_n, \rho \>\right)}_{\xrightarrow{n\to\infty}\, 0}\sigma' + \underbrace{\vphantom{\frac{1}{\mu}}\left(\< A_n, \rho \> - 1\right)}_{\xrightarrow{n\to\infty} \,\lambda \mu - 1} \rho''.
\end{aligned}\end{equation}
That is, in the limit $n\to\infty$, the sequence of the output operators $\E_n(\rho)$ converges (up to normalisation) to the target state $\rho'$, thus realising the desired free transformation. To complete the proof, we need to ensure that the maps $\E_n$ are valid probabilistic operations by rescaling them to be trace--non-increasing maps as $\E_n' \coloneqq \E_n / \left[ \max_{\omega \in \DD} \Tr \E_n(\omega) \right]$. The sequence of probabilistic free operations $(\E'_n)_n$ is precisely our desired protocol which achieves the transformation $\rho \transf \rho'$.

\begin{remark}Noticing that $\Tr \E_n(\rho) \tends k (\lambda \mu - 1)$, which is non-zero for any $\rho \notin \FF$, one can wonder whether there actually exists a possibility that the protocol only succeeds asymptotically, in the sense of vanishing probability of success as in Eq.~\eqref{eq:prob_asymp}. That is indeed the case, and the culprit is the renormalisation coefficient
\begin{equation}\begin{aligned}
  \max_{\omega \in \DD} \Tr \E_n(\omega) = \norm{(\lambda - k) B_n + \left(k - \frac1\mu\right) A_n}{\infty}
\end{aligned}\end{equation}
which is used to rescale the maps $\E_n$ to be trace non-increasing.
In general, we have no guarantee that this quantity is bounded, and indeed it might be the case that both $\norm{A_n}{\infty}$ and $\norm{B_n}{\infty}$ diverge to infinity as $n \to \infty$; such divergent cases are precisely when we get $\Tr \E'_n(\rho) \tends 0$ after renormalising. On the other hand, when the supremum in the dual program for $\RW(\rho)$ is achieved by some operators $A$ and $B$, then we can choose
\begin{equation}\begin{aligned}
  \E(X) \coloneqq \< B, X \> \left(\lambda \sigma' - \rho''\right) + \<A , X \> \left( \rho'' - \frac{1}{\mu} \sigma'\right)
\end{aligned}\end{equation}
as an optimal operation which satisfies $\E \in \OO$ and $\E(\rho) \propto \rho'$. This is the case, for instance, when there exists optimal sequences $(A_n)_n$ and $(B_n)_n$ which are both bounded: choosing $(n_k)_k$ as a sequence of indices such that both of the subsequences $(A_{n_k})_k$ and $(B_{n_k})_k$ converge, it suffices to take $A$ and $B$ defined as the limits of such subsequences. We will return to this problem in Lemma~\ref{lem:achiev_omega}.\end{remark}

(ii) For the case $\RW(\rho)=\infty$ and $R_\FF(\rho') \eqqcolon \lambda' <\infty$, we will find it easier to work with the reciprocal function $\RW(\rho)^{-1}$, as we have $\RW(\rho)^{-1} = 0$ and the optimal value of the dual optimisation problem is finite. Using the dual form of $\RW(\rho)^{-1}$ in Lemma~\ref{lem:affine}, there exists sequences $(A_n)_n, (B_n)_n$ such that $\< A_n, \rho \> \tends 0$, $\< B_n, \rho \> = 1 \; \forall n$, and $\< A_n , \sigma \> = \< B_n, \sigma \> \; \forall \sigma \in \FF, n \in \NN$. This allows us to construct the maps
\begin{equation}\begin{aligned}
  \E_n(X) \coloneqq \< A_n, X \> (\lambda' \sigma' - \rho') + \< B_n, X \> \rho'
\end{aligned}\end{equation}
where $\sigma' \in \FF$ is a state such that $\rho' \leq \lambda' \sigma'$, which exists by assumption. Clearly, $\E_n(\rho) \tends \rho'$ and $\E_n(\sigma) = \< B_n, \sigma \> \lambda' \sigma'$ for any $\sigma \in \FF$. Renormalising the operations to be trace non-increasing, we get the desired transformation.

As before, when there exists a choice of $A$, $B$ optimal for $\RW(\rho)^{-1}$, then it suffices to take
\begin{equation}\begin{aligned}
  \E(X) \coloneqq \< A, X \> (\lambda' \sigma' - \rho') + \< B, X \> \rho'.
\end{aligned}\end{equation}

(iii) The case of $R_\FF(\rho) = \infty$ implies that $\rho$ must have a non-trivial overlap with the subspace orthogonal to all states in $\FF$, i.e., that there exists a vector $\ket\psi$ such that $\< \psi, \rho \> > 0$ but $\< \psi, \sigma \> = 0$ for all $\sigma \in \FF$. Were this not the case, by convexity of $\FF$ we could take a free state whose support is the union of the supports of all free states, and the support of $\rho$ would have to be contained in this subspace, contradicting the assumption that $R_\FF(\rho) = \infty$. We then construct the map
\begin{equation}\begin{aligned}
  \E(X) \coloneqq \< \psi, X \> \rho',
\end{aligned}\end{equation}
which achieves the desired transformation with probability $\< \psi, \rho \>$. The fact that $0 \leq \psi \leq \id$ implies that $\< \psi, \rho \> \in [0,1]$ for any state $\rho$, and together with the fact that $\E(\sigma) = 0$ for any $\sigma \in \FF$, we indeed have that $\E \in \OO$.
\end{proof}

\subsubsection*{General resource theories}

Let us now consider the case of non-affine resource theories. Here, the issue we encounter is that the simplified form of $\RW$ in Lemma~\ref{lem:affine} is no longer applicable, preventing us from following the same approach. To circumvent this, we will employ the free projective robustness $\RW^\FF$.


\begin{boxed}{white}
{%
\begin{theorem}\label{thm:nogo_sufficient_fulldim_full}
In any convex resource theory, there exists a resource--non-generating probabilistic transformation $\rho \transf \rho'$ if any of the following conditions is satisfied:
\begin{enumerate}[(i)]
\item $\infty > \RW(\rho) \geq \RW^\FF(\rho')$,
\item $\infty = \RW(\rho)$ and $R^\FF_\F(\rho') < \infty$,
\item $R_\F(\rho) = \infty$.
\end{enumerate}
\end{theorem}
}
\end{boxed}

Combining this with the monotonicity result of Theorem~\ref{thm:nogo_monotonicity}, as long as $R_\FF^\FF(\rho') < \infty$ (which is true for all states in any full-dimensional resource theory), we have the following:
\begin{equation}\begin{aligned}
  \rho \transf\rho' &\;\Rightarrow\; \RW(\rho) \geq \RW(\rho'),\; \RW^\FF(\rho) \geq \RW^\FF(\rho'),\\
 \rho \transf\rho' &\;\Leftarrow\; \RW(\rho) \geq \RW^\FF(\rho').
\end{aligned}\end{equation}
We see here that this reduces to an `if and only if' statement only when $\RW(\rho') = \RW^\FF(\rho')$. Although this might not always be true, we will shortly see that states obeying this requirement can be defined in relevant cases (such as the problem of resource distillation), allowing us to use this result as a necessary and sufficient condition for resource conversion also in non-affine theories.

\begin{proof}
(i) The basic idea is very similar to Theorem~\ref{thm:nogo_affine_full}. We construct the maps
\begin{equation}\begin{aligned}\label{eq:measure_prep_2}
  \E_n(X) \coloneqq \< B_n, X \> \left(\lambda \sigma' - \rho''\right) + \<A_n , X \> \left( \rho'' - \frac{1}{\mu} \sigma'\right),
\end{aligned}\end{equation}
where $(A_n)_n, (B_n)_n$ are optimal sequences of operators for $\RW(\rho)$ as before, but now $\sigma'$ is a state such that $\rho'' \leq_\FF \lambda \sigma'$ and $\sigma' \leq \mu \rho''$, with $\rho''$ being the state $\rho'$ which has been rescaled if needed. The crucial difference now is that $\lambda \sigma' - \rho'' \in \cone(\FF)$. What this means is that, for any $\sigma \in \FF$,
\begin{equation}\begin{aligned}
  \E_n(\sigma) &= \< B_n, \sigma \> \left(\lambda \sigma' - \rho''\right) + \<A_n, \sigma \> \left( \rho'' - \frac{1}{\mu} \sigma'\right)\\
  &\geq_\FF \<A_n, \sigma \> \left(\lambda \sigma' - \rho''\right) + \<A_n, \sigma \> \left( \rho'' - \frac{1}{\mu} \sigma'\right)\\
  &= \left(\lambda - \frac{1}{\mu}\right) \< A_n, \sigma \> \sigma' \,\in \cone(\FF)
\end{aligned}\end{equation}
which, by the convexity of $\cone(\FF)$, means that $\E_n(\sigma) \in \cone(\FF)$ and hence that $\E_n \in \OO$. 
Here, we used that $\<B_n, \sigma\> \geq \< A_n,\sigma\>$ for any $\sigma\in\FF$, which is a property obeyed by any feasible dual solution to $\RW(\rho)$ (cf.\ Theorem~\ref{thm:proj_properties}). Clearly, we also have
\begin{equation}\begin{aligned}
  \E_n(\rho) = \frac{1}{\mu} \left( \lambda \mu - \< A_n, \rho \>\right)\sigma' + \left(\< A_n, \rho \> - 1\right) \rho'' \, \tends (\lambda \mu - 1 ) \rho'' \propto \rho'
\end{aligned}\end{equation}
as before. The operations $(\E_n)_n$ therefore realise the desired transformation, and with a suitable rescaling they form a valid probabilistic transformation.

(ii) Taking the dual form of the reciprocal function $\RW(\rho)^{-1}$, we can write
\begin{equation}\begin{aligned}
  \RW(\rho)^{-1} = \inf \lset \< A, \rho \> \bar \< B, \rho \> = 1 ,\; \< A - B, \sigma \> \geq 0 \; \forall \sigma \in \FF,\; A, B \geq 0 \rset.
\end{aligned}\end{equation}
Since $\RW(\rho)^{-1} = 0$ by assumption, we can take the optimal choice of such operators satisfying $\< A_n, \rho \> \tends 0$ and construct the maps
\begin{equation}\begin{aligned}
  \E_n(X) \coloneqq \< A_n, X \> (\lambda' \sigma' - \rho') + \< B_n, X \> \rho'
\end{aligned}\end{equation}
where $\sigma' \in \FF$ is a state such that $\rho' \leq_\FF \lambda' \sigma'$ with $\lambda' \coloneqq R_{\FF}^\FF(\rho')$. We can see that $\E_n(\rho) \tends \rho'$ and
\begin{equation}\begin{aligned}
  \E_n(\sigma) &\geq_\FF \< B_n, X \> (\lambda' \sigma' - \rho') + \< B_n, X \> \rho'\\
  &= \< B_n, X \> \lambda' \sigma'\, \in \cone(\FF)
\end{aligned}\end{equation}
for any $\sigma \in \FF$, as was to be shown.

(iii) Exactly the same as in Theorem~\ref{thm:nogo_affine_full}.
\end{proof}

We now give sufficient conditions which strengthen the result to show that not only can the transformation $\rho \transf \rho'$ be realised probabilistically, it is actually achieved with a non-zero probability of success (that is, no asymptotic protocols are necessary).

\begin{boxed}{white}
\begin{lemma}\label{lem:achiev_omega}
Let $\rho,\rho'$ be two states such that $\rho \transf \rho'$ is possible probabilistically (as per Theorem~\ref{thm:nogo_affine_full} or~\ref{thm:nogo_sufficient_fulldim_full}). If additionally either of the following is true:
\begin{enumerate}[(a)]
 \item $\rho$ is full rank and $\RW(\rho) < \infty$ (i.e.\ there exists a full-rank free state), or
 \item there is no free state $\sigma$ such that $\supp \sigma \subseteq \supp \rho$ (i.e.\ $W_\FF(\rho) = 0$), 
\end{enumerate}
 then there exists $\E \in \OO$ such that $\E(\rho) = p \rho'$ for some $p>0$.
\end{lemma}
\end{boxed}
Note in particular that condition (b) is satisfied for any resourceful pure state.
\begin{proof}
As discussed in the proof of Theorem~\ref{thm:nogo_affine_full}, the existence of such a map $\E \in \OO$ can be ensured by showing that the dual optimisation problems for $\RW(\rho)$ or $\RW(\rho)^{-1}$ are achieved, in the sense that there exist feasible $A, B$ such that $\RW(\rho)$ (or $\RW(\rho)^{-1}$) equals $\< A, \rho \>$.

(a) When $\rho$ is a full-rank state with $\RW(\rho) < \infty$, then any $A \geq 0$ feasible for $\RW(\rho)$ satisfies $\<A, \omega\> \leq \Rmax(\omega\|\rho) \<A,\rho\> \leq \lambda_{\min}(\rho)^{-1} \RW(\rho)$ for any state $\omega$ by Lemma~\ref{lem:weight_lambda}. This means that $\norm{A}{\infty} = \max_{\omega \in \DD} \< A, \omega\>$ is upper bounded by a constant independent of $A$. Similarly, the condition $\<B, \rho \> =1$ ensures that $\norm{B}{\infty} \leq \lambda_{\min}(\rho)^{-1}$. Any optimal sequences $(A_n)_n$, $(B_n)_n$ are therefore bounded, and taking suitable convergent subsequences ensures that the supremum in the dual form of $\RW(\rho)$ is achieved.

(b) Note first that $W_\FF(\rho) = 0$ implies $\RW(\rho) = \infty$, so we consider the quantity $\RW(\rho)^{-1}$. In this case, we have that $\<\Pi_\rho, \rho \> = 1$ but $\<\Pi_\rho, \sigma\> < 1 \; \forall \sigma \in \FF$, where  $\Pi_\rho$ denotes the projection onto $\supp \rho$. Let $\zeta = \max_{\sigma \in \FF} \< \Pi_{\rho}, \sigma \>$ and define
\begin{equation}\begin{aligned}
  A \coloneqq \frac{\zeta}{1-\zeta} (\id - \Pi_\rho), \quad B \coloneqq \Pi_\rho.
\end{aligned}\end{equation}
We see that $\<A, \rho \> = 0$, $\<B, \rho\> = 1$, and $\<A - B, \sigma \> \geq 0 \; \forall \sigma \in \FF$. This means that $A$ and $B$ are feasible solutions for $\RW(\rho)^{-1}$ which achieve the optimal value $\RW(\rho)^{-1} = 0$.
\end{proof}

A general consequence of the results in Theorems~\ref{thm:nogo_affine_full} and \ref{thm:nogo_sufficient_fulldim_full} is that states with infinite projective robustness can always be interconverted probabilistically; that is,
\begin{equation}\begin{aligned}
  \RW(\rho) = \RW(\rho') = \infty \;\Rightarrow\; \rho \transf \rho',\;\, \rho' \transf \rho.
\end{aligned}\end{equation}
This includes, in particular, all pure resourceful states, and thus implies that no such state can be isolated: any pure state can be obtained from another resourceful pure state with some non-zero probability under resource--non-generating transformations. 
This contrasts with some more restricted theories, such as entanglement manipulated by LOCC: in the bipartite setting, pure states can be interconverted only if they share the same Schmidt rank~\cite{bennett_1996-1,lo_2001}, and in the multipartite setting, there exist inequivalent types of entanglement that cannot be interconverted~\cite{dur_2000}. We thus see that using the operations $\OO$ allows one to overcome such limitations in the regime of pure-state transformations. 
In the context of entanglement theory, similar results about the convertibility of pure states have been shown previously through other approaches~\cite{ishizaka_2004-1,ishizaka_2005,contreras-tejada_2019}; our approach extend this to general resource theories and to states beyond pure.
Here we remark the rather interesting fact that, when \emph{deterministic} transformations are concerned, the operational capabilities provided by the resource--non-generation operations $\OO$ are not always significantly larger than those of smaller, physically motivated operations --- for example, in the concentration of entanglement from pure states, the maximal operations $\OO$ do not even outperform one-way LOCC~\cite{regula_2019-2}. 

{%
\renewcommand\addcontentsline[3]{}
\subsection{Projective robustness in channel discrimination}
}

Another application of the projective robustness $\RW$ can be obtained by connecting it with the advantage that the given resource provides in a class of channel discrimination tasks. This is essentially an extension of similar properties that were recently shown to hold for the robustness~\cite{takagi_2019-2} and weight~\cite{uola_2020-1,ducuara_2020} measures. Since this is not related to our main discussion concerned with the uses of $\RW$ in constraining resource manipulation, we defer the details to Appendix~\ref{app:proj_chandisc}.


\mbox{}\\
\mbox{}

\section{Probabilistic resource distillation}\label{sec:prob_dist_RW}

We will now show how our results can be employed to establish strong bounds on the performance of resource distillation tasks. In this setting, we aim to transform a given noisy state $\rho$ into some pure resource state $\phi$, often taken to be a maximally resourceful state or a particularly useful state that is necessary in the practical utilisation of a given resource. The representative example of such a task is entanglement distillation, whose probabilistic characterisation attracted significant attention~\cite{lo_2001,linden_1998,kent_1998,horodecki_1999-1,vidal_1999-1,horodecki_2006-1}.

However, we already know from Theorem~\ref{thm:nogo_monotonicity} that transformations $\rho \transf \phi$ are never possible when $\rho$ is a generic (full-rank) noisy state. This can be circumvented by allowing a small error in the conversion, which means that we will instead aim to achieve a transformation $\rho \transf \tau$, where $\tau$ is a state such that $F(\tau, \phi) \geq 1-\ve$ with $F(\rho,\sigma)\coloneqq \|\sqrt{\rho}\sqrt{\vphantom{\rho}\sigma}\|_1^2$ being the fidelity.  We then endeavour to understand and bound the error $\ve$ necessarily incurred in such protocols.

Let us begin by recalling some bounds that have been previously established in this context for general resource theories. The robustness $R_\FF$ can provide both one-shot~\cite{regula_2020} and many-copy bounds~\cite{regula_2021-1} for distillation, with the latter particularly applicable to asymptotic transformation rates. When it comes to bounds on resource overheads --- that is, the number of copies of a state required to realise a desired distillation protocol --- the resource weight $W_\FF$ was found to provide strong constraints~\cite{regula_2021-1,fang_2022}. However, such results were mostly concerned with deterministic transformations, and probabilistic conversion protocols significantly complicate the characterisation of distillation errors~\cite{regula_2021-1,fang_2022}. We will see that the projective robustness $\Omega_\FF$, which can be thought of as combining the properties of $R_\FF$ and $W_\FF$, will allow us to establish broadly applicable bounds that overcome the limitations of previous approaches.

Another approach to constraining the distillation of general resources, including in a probabilistic setting, was introduced in Ref.~\cite{fang_2020} in the form of the \emph{eigenvalue bound}. This restriction, applicable to the distillation of pure states from full-rank input states, was claimed to establish a trade-off relation between the errors and the probability of success of transformations --- although there are issues with this latter claim (as will be shown in Section~\ref{sec:eigenvalue}), the result can nevertheless be used to lower bound transformation errors in general probabilistic protocols as considered here. However, we will see in several examples that it typically does not yield good restrictions: the bound obtained from the projective robustness $\RW$ can never be worse than the eigenvalue bound, and in most cases it performs significantly better.


\subsection{Improved bounds on probabilistic distillation errors and overheads}\label{app:proj_distillation}

We show that the projective robustness can be used to establish a tight bound on distillation error in general probabilistic transformations, giving a non-trivial restriction for all full-rank states as well as rank-deficient states that satisfy $\RW(\rho) < \infty$. The resourcefulness of the target state $\phi$ will be quantified using the overlap $F_\FF(\phi) \coloneqq \max_{\sigma \in \FF} \braket{\phi|\sigma|\phi}$, which effectively determines how difficult a given state is to distill. Note that $F_\FF$ is a reverse monotone, in the sense that it gets smaller the more resourceful a state is.

\begin{boxed}{white}
\begin{theorem}[\cite{regula_2022}]\label{thm:prob_error_full}
If there exists a free transformation $\rho\!\transf\!\tau$ such that $\tau$ is a state satisfying $F(\tau,\phi) \geq 1-\ve$ for some resourceful pure state $\phi$, then
\begin{equation}\begin{aligned}
  \ve \geq \left( \frac{F_\FF(\phi)}{1-F_\FF(\phi)} \,\RW(\rho)\, + 1 \right)^{-1}
\end{aligned}\end{equation}
where $F_\FF(\phi) = \max_{\sigma \in \FF} \< \phi, \sigma \>$. The statement is valid for all values of $\ve \in [0,1]$ and all $F_\FF(\phi) \in [0,1)$ as long as one understands $0 \infty \coloneqq \infty$.
\end{theorem}
\end{boxed}

The result can be viewed either as a restriction on achievable distillation fidelity $1-\ve$ for a fixed target state $\phi$, or as a constraint on how much of the given resource (quantified by $F_\FF(\phi)$) can be distilled from $\rho$ up to given accuracy $\ve$.

We note that such thresholds were studied in the theory of quantum entanglement~\cite{kent_1998,horodecki_1999-1}, but even in that case no general quantitative bound was previously known.

\begin{proof}
Assume first that $\ve > 0$ and $F_\FF(\phi) > 0$. The stated relation is trivial when $\RW(\rho) = \infty$, so assume otherwise. Then also $\RW(\tau) < \infty$ by Theorem~\ref{thm:nogo_monotonicity}, so let $\sigma$ be a free state such that $\tau \leq \lambda \sigma$ and $\sigma \leq \mu \tau$ with $\RW(\tau) = \lambda\mu$. We then get
\begin{equation}\begin{aligned}\label{eq:lambdamu_lowerbounds}
  \< \tau, \phi \> &\leq \lambda \< \sigma, \phi \> \leq  \lambda F_\FF(\phi),\\
  1-F_\FF(\phi) &\leq \< \sigma, \id - \phi \> \leq \mu \< \tau, \id - \phi \>.
\end{aligned}\end{equation}
Using the assumption $\< \tau, \phi \> \geq 1-\ve$ and multiplying the above inequalities together gives
\begin{equation}\begin{aligned}\label{eq:prob_error_bound}
  \RW(\rho) \geq \RW(\tau) = \lambda \mu \geq \frac{(1-\ve)\left(1-F_\FF(\phi)\right)}{\ve F_\FF(\phi)},
\end{aligned}\end{equation}
where we used the monotonicity of $\RW$. Rearranging, we get the claimed relation.

To conclude, let us clarify the cases when one or both of $\ve$ and $F_\FF(\phi)$ is zero.  The Theorem is consistent with the no-go result of Theorem~\ref{thm:nogo_monotonicity} when $\ve=0$, as it implies that $\RW(\rho) = \infty$ is necessary for the transformation to be possible. In the case $F_\FF(\phi) = 0$, when $\RW(\rho) < \infty$, we get $\ve \geq 1$ and the transformation is not possible. If, however, $F_\FF(\phi) = 0$ and $\RW(\rho) = \infty$, then we can only give the trivial bound $\ve \geq 0$; indeed, if also $R_\FF(\rho) = \infty$, then $\ve = 0$ is actually achievable (see Theorem~\ref{thm:nogo_sufficient_fulldim_full}) so no better bound that depends only on $\RW(\rho)$ can be given.
\end{proof}

\begin{remark}
The result can be alternatively seen by using the dual form of $\RW$ from Eq.~\eqref{eq:thm_convex_progB}. Taking the feasible dual solutions
\begin{equation}\begin{aligned}
  A = \frac{1 - F_\FF(\phi)}{\ve  F_\FF(\phi)}\phi, \qquad B =  \frac{\id - \phi}{\ve}
\end{aligned}\end{equation}
allows us to lower bound $\RW(\tau)$ as in Eq.~\eqref{eq:prob_error_bound}.
\end{remark}

In probabilistic protocols used in practice, one is often able to improve the transformation accuracy by `gambling' with the resource. This idea stems from entanglement distillation~\cite{bennett_1996,lo_2001,horodecki_1999-1,rozpedek_2018} and is based on the fact that transformation probability can often be traded for accuracy: a higher output fidelity can be obtained by sacrificing some probability. Since our bound holds for all values of $p$ and even in the limit as $p \to 0$, it establishes restrictions which are impossible to beat through such trade-offs.

The bound can tightly characterise the limitations of probabilistic transformation protocols. We explicitly demonstrate this in Figure~\ref{fig:iso}, where we plot bounds on the error incurred in entanglement distillation. 
By choosing the set of positive partial transpose (PPT) states as the free states $\FF$, the achievable performance of distillation protocols is computable exactly with an SDP (see the forthcoming Section~\ref{sec:tradeoff}; see also~\cite{rozpedek_2018}), which we can then use to benchmark the performance of our bound. Not only does the projective robustness bound of Theorem~\ref{thm:prob_error_full} significantly improve over the eigenvalue bound~\eqref{eq:eig_bound}, we will shortly prove (see Corollary~\ref{thm:prob_error_tight}) that $\RW$ gives the tightest possible restriction on the achievable error of probabilistic distillation schemes --- the whole region in Figure~\ref{fig:iso} that is not forbidden by the projective robustness bound is actually achievable by probabilistic protocols, and such a feature extends to more general resource theories.
As detailed in the Figure, the applications of our bound provide new insights into the advantages that can be obtained by employing probabilistic protocols in entanglement transformations, strengthening and extending previous results where entanglement manipulation under LOCC~\cite{lo_2001,linden_1998,kent_1998,horodecki_1999-1,vidal_2000-1,horodecki_2006-1} was considered.

\begin{figure}
 \centering
        \vspace*{-1.8\baselineskip}
    \subfloat[][$\displaystyle \rho = (1-\gamma) \phi_2 + \gamma \id/4$]{%
    \includegraphics[width=7.55cm]{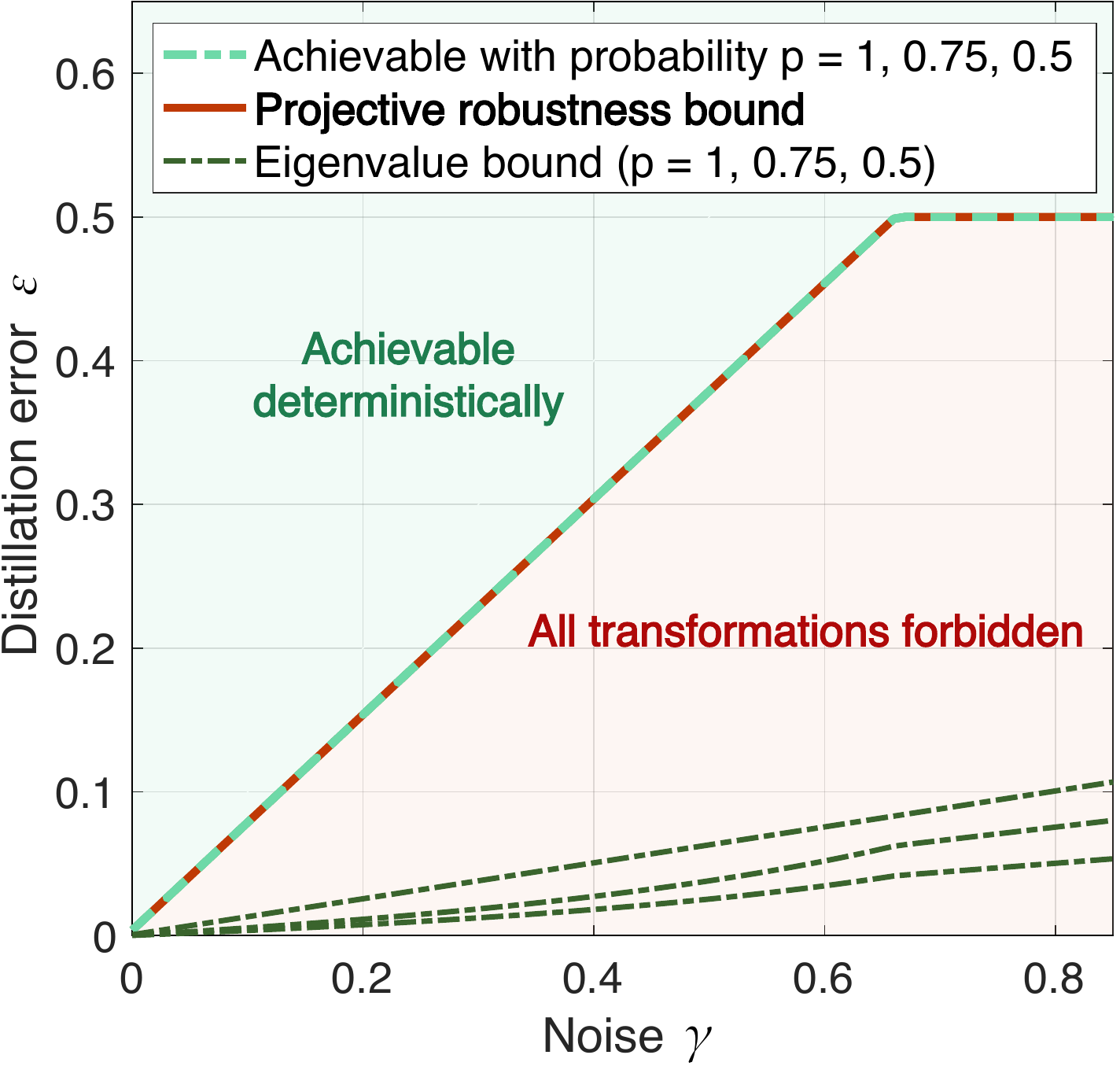}
    \hspace*{20pt}
    }%
    \subfloat[][$\displaystyle \rho = \left\lbrack(1-\gamma) \phi_2 + \gamma \id/4\right\rbrack^{\otimes 2}$]{%
    \includegraphics[width=7.55cm]{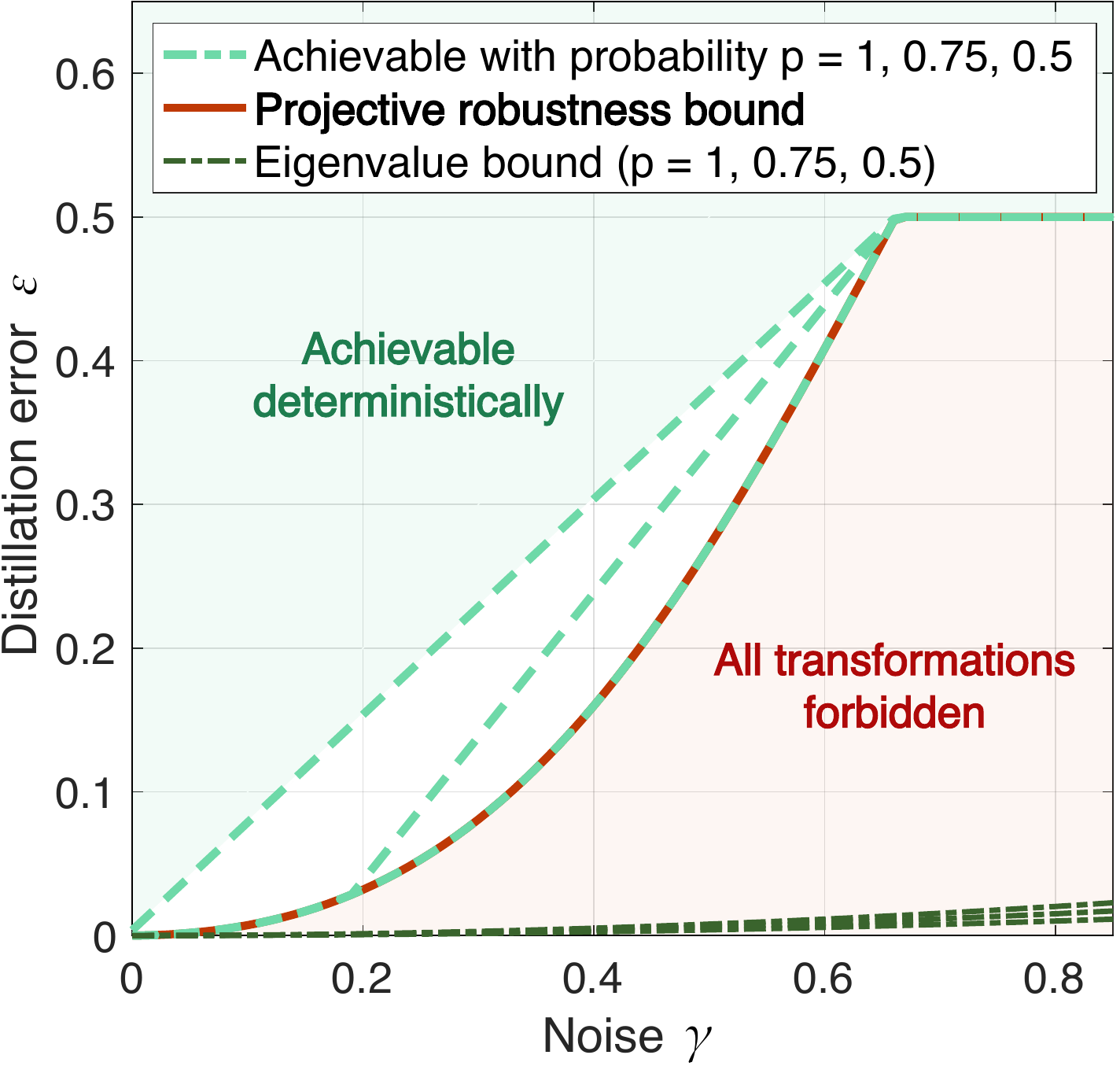}
    }%
    \\
    \subfloat[][$\displaystyle \rho = \left\lbrack(1-\gamma) \phi_2 + \gamma \id/4\right\rbrack^{\otimes 3}$]{%
    \includegraphics[width=7.55cm]{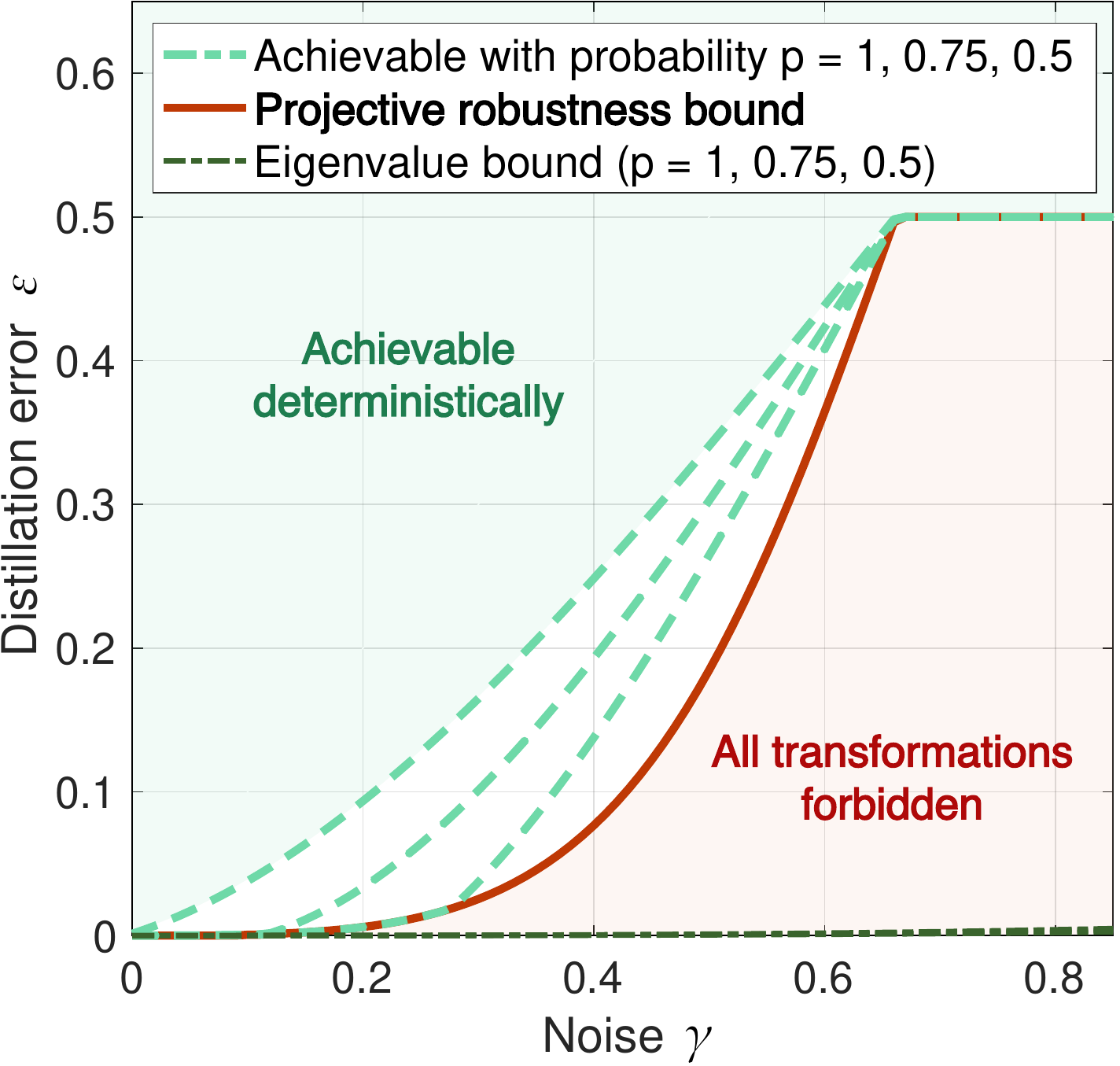}
    }%
\caption{\textbf{Bounding achievable error in entanglement distillation.} We consider the task of distilling a single copy of a maximally entangled state $\ket{\phi_2}$ under operations $\OO$, which in this case correspond to PPT-preserving operations. The projective robustness bound of Theorem~\ref{thm:prob_error_full}, which lower bounds the achievable error $\ve$ under any physical transformation, is compared with the eigenvalue bound of \cite{fang_2020} (see Eq.~\eqref{eq:eig_bound}) as well as with the achievable performance of distillation protocols. The achievable error is plotted in both the deterministic ($p=1$) case and the probabilistic case for a choice of $p \in \{0.75, 0.5\}$ (from top to bottom, respectively).\\
We see in \textbf{(a)} that, for a single copy of an isotropic state, probabilistic protocols offer no improvement in distillation error --- the projective robustness bound is actually achievable deterministically, and indeed it is achievable by simply leaving the state $\rho$ as it is and not performing any operations whatsoever. This provides a strengthening of the result of \cite{linden_1998}, where it was shown that probabilistic LOCC cannot increase the distillation fidelity of a single isotropic state; we extend this to all PPT-preserving (or non-entangling) operations. (We will later see an analytical proof of this fact in Corollary~\ref{cor:isotropic_nogo}.) Note that the eigenvalue bound is too weak to reach the same conclusion.\\
In \textbf{(b)}, where two copies of an isotropic state are used, we see that probabilistic protocols offer an enhanced performance over deterministic ones. We also see that the protocol for $p=0.5$ achieves the projective robustness bound, which means that no further improvement in fidelity can be obtained by choosing $p < 0.5$.\\
The plot in \textbf{(c)}, on the other hand, suggests that an improved distillation fidelity can be obtained by lowering the probability below $p = 0.5$. Note the extremely poor performance of the eigenvalue bound in comparison with the projective robustness and the achievable protocols.
}
\label{fig:iso}
\end{figure}

We remark here that yet another approach to bounding achievable error in probabilistic transformations was introduced in Ref.~\cite{regula_2021-1} based on the strong monotonicity of the robustness and weight measures. Although accurate for $p \approx 1$, these bounds quickly become trivial as $p$ decreases, which is why we do not include them in our comparison.

An important consequence of the universal bound of Theorem~\ref{thm:prob_error_full} is that we can use the subadditivity of the projective robustness to provide limitations on the overhead required in many-copy transformations. Specifically, consider a protocol where one aims to minimise the transformation error by using more input copies of the given state. Given access to $\rho^{\otimes n}$, we are then concerned with the question of how many copies are required to achieve a desired accuracy.
\begin{boxed}{white}
{
\begin{corollary}\label{cor:overhead}
Assume that the given convex resource theory is closed under tensor product, in the sense that $\sigma \in \FF \Rightarrow \sigma^{\otimes n} \in \FF$. 
If there exists a free transformation $\rho^{\otimes n}\!\transf\!\tau$ such that $\tau$ is a state satisfying $F(\tau,\phi) \geq 1-\ve$ for some resourceful pure state $\phi$, then
\begin{equation}\begin{aligned}
  n \geq \log_{\RW(\rho)} \frac{(1-\ve)\left(1-F_\FF(\phi)\right)}{\ve F_\FF(\phi)}.
\end{aligned}\end{equation}
\end{corollary}
}
\end{boxed}
\begin{proof}
Follows from Eq.~\eqref{eq:prob_error_bound} and the submultiplicativity of $\RW$ (Theorem~\ref{thm:proj_properties}\eqref{itm:submult}).
\end{proof}
The practical performance of this bound has already been demonstrated in Ref.~\cite{regula_2022}, where it was shown in particular that $\RW$ significantly improves on the eigenvalue bound.
We note that, although we considered the target state to be a single system $\phi$, one can equivalently study the overhead in transformations of the type $\rho^{\otimes n} \transf \phi^{\otimes m}$. Importantly, the bound only depends on the overlap $F_\FF(\phi)$ of the target state, which in many resource theories satisfies $F_\FF(\phi^{\otimes m}) = F_\FF(\phi)^m$ and is therefore simple to compute for many-copy target states (this includes the theories of entanglement~\cite{shimony_1995}, coherence, athermality, and magic when $\phi$ is a state of up to 3 qubits~\cite{bravyi_2019}).

Our bound here shows that the total overhead $n$ of any distillation protocol must scale as $\Omega\left(\log\frac{1-\ve}{\ve}\right)= \Omega\left(\log \frac{1}{\ve}\right)$, even if the transformation probability is allowed to vanish asymptotically. This strengthens the previous observation of Ref.~\cite{fang_2020} where this was shown to hold for constant, non-vanishing probability of success.

Let us comment briefly on the applicability and usefulness of the restrictions introduced here.
The key property of the bounds of this section is to constrain all possible probabilistic transformations, regardless of their probability of success. 
In some cases, rather than a constraint on \emph{all} probabilistic protocols, one may wish to obtain bounds on transformations achievable for a certain fixed probability --- such restrictions are the subject of the forthcoming Section~\ref{sec:tradeoff}.
If constraints only on deterministic protocols are desired, there is no need to consider the projective robustness: in such a case, either the robustness-~\cite{regula_2020,regula_2021-1} or the weight-based~\cite{regula_2021-1,fang_2022} bound will perform better than $\RW$ in the context of resource distillation.


\subsection{Comparison with the eigenvalue bound}\label{sec:eigenvalue}

The eigenvalue bound from Ref.~\cite{fang_2020} states that, for any full-rank input state $\rho$ and any $\ve$-error probabilistic protocol that succeeds with probability $p$, it holds that
\begin{equation}\begin{aligned}\label{eq:eig_bound}
  \frac{\ve}{p} \geq \frac{\lambda_{\min}(\rho)\left(1-F_\FF(\phi)\right)}{R_\FF(\rho)},
\end{aligned}\end{equation}
where $\lambda_{\min}$ denotes the smallest eigenvalue. 
We have already seen in Figure~\ref{fig:iso} (cf.\ also~\cite{regula_2022}) that the restrictions obtained using the projective robustness can significantly improve on this bound; we now explicitly prove that the $\RW$-based bounds are never worse than the eigenvalue bound. 
In fact, we will also show that the dependence of the eigenvalue bound on the probability of success $p$ is superficial --- we can replace $p$ with the constant $1$ and the bound still holds for any probabilistic transformation, including protocols that succeed with asymptotically vanishing probability.
\begin{boxed}{white}
\begin{lemma}\label{lem:vacuous}
If there exists a probabilistic transformation $\rho \transf \tau$ such that $\tau$ is a state satisfying $F(\tau,\phi) \geq 1-\ve$ for some resourceful pure state $\phi$, then
\begin{equation}\begin{aligned}
\ve \geq \frac{1-F_\FF(\phi)}{\RW(\rho)} \geq \frac{\lambda_{\min}(\rho)\left(1-F_\FF(\phi)\right)}{R_\FF(\rho)}.
\end{aligned}\end{equation}
\end{lemma}
\end{boxed}
\begin{proof}
The first inequality is essentially a weaker version of Theorem~\ref{thm:prob_error_full}; here, we write $\RW(\tau) = \lambda \mu$, recall from Eq.~\eqref{eq:lambdamu_lowerbounds}  that
\begin{equation}\begin{aligned}
  1-F_\FF(\phi) &\leq \mu \< \tau, \id - \phi \> \leq \mu \ve,
\end{aligned}\end{equation}
and use the trivial lower bound $\lambda \geq 1$ to obtain
\begin{equation}\begin{aligned}
  \RW(\rho) \geq \RW(\tau) \geq \mu \geq \frac{1-F_\FF(\phi)}{\ve}
\end{aligned}\end{equation}
using the monotonicity of $\RW$. Finally, the upper bound $\RW(\rho) \leq  R_\FF(\rho) \, \lambda_{\min}(\rho)^{-1}$ (see Theorem~\ref{thm:proj_properties}\eqref{itm:bounds}) gives the statement of the Lemma.
\end{proof}

Let us discuss a consequence of this finding. One of the stated aims of the eigenvalue bound was to understand trade-offs between the error $\ve$ and probability $p$~\cite{fang_2020}; in this sense, it can be viewed as a bound on the probability of success as
\begin{equation}\begin{aligned}\label{eq:vacuous}
  p \leq \frac{\ve R_\FF(\rho)}{\lambda_{\min}(\rho)\! \left(1 - F_\FF(\phi)\right)} \eqqcolon v(\rho, \ve),
\end{aligned}\end{equation}
However, our result implies that this bound is \emph{unable} to give any non-trivial restrictions on the probability, since it holds that $\displaystyle v(\rho,\ve) \geq \RW(\rho)\,/\,\RW(\tau)$.
What this means is that, when $v(\rho,\ve) < 1$, this implies that $\RW(\tau) > \RW(\rho)$ and hence the transformation is actually impossible with any non-zero probability (as per Theorem~\ref{thm:nogo_monotonicity}), while if $v(\rho,\ve) \geq 1$, then the bound trivialises. Indeed, the stronger statement of the eigenvalue bound that we derived in Lemma~\ref{lem:vacuous}, and the ensuing 
superficiality of the dependence of Eq.~\eqref{eq:vacuous} on the probability $p$, means that the eigenvalue bound applies in the same way regardless of the probability of success of the distillation protocol.

\begin{remark}
Ref.~\cite{fang_2020} concluded from the eigenvalue bound that probabilistic transformations from a full-rank state $\rho$ to a resourceful pure state $\phi$ is impossible in \textit{any} resource theory. This is not strictly correct, as can be seen from the example where $\FF = \{ \proj{0} \}$ and $\rho$ is any full-rank one-qubit state. The state $\ket{1}$ can then be probabilistically obtained from $\rho$ by simply projecting onto it --- such a projection cannot generate any resource when acting on $\proj{0}$, so it is a free probabilistic operation. The missing assumption in Ref.~\cite{fang_2020} is that the robustness $R_\FF(\rho)$ is finite. (Note also that this no-go result is significantly strengthened by our Theorem~\ref{thm:nogo_monotonicity}, as we discussed previously).
\end{remark}


\subsection{Achievable fidelity}

We now investigate the tightness of the general bound of Theorem~\ref{thm:prob_error_full} by constructing achievable protocols for the distillation of resources.

\begin{boxed}{white}
\begin{theorem}\mbox{}\label{thm:RW_dist_achiev}
Let $\phi$ be a resourceful pure state in any convex resource theory. For any $\ve$ satisfying
\begin{equation}\begin{aligned}\label{eq:ve_choice_achiev}
  \ve \geq \left( \frac{\RW(\rho)}{R_\FF^\FF(\phi)-1} \, + 1 \right)^{-1}\!,
\end{aligned}\end{equation}
there exists a resource non-generating protocol $\rho \transf \tau_\ve$ where $\tau_\ve$ is a state such that $F(\tau_\ve,\phi) \geq 1-\ve$.

If the resource theory is affine, it suffices to take
\begin{equation}\begin{aligned}\label{eq:ve_choice_achiev_aff}
  \ve \geq \left( \frac{\RW(\rho)}{R_\FF(\phi)-1} \, + 1 \right)^{-1}\!.
\end{aligned}\end{equation}
\end{theorem}
\end{boxed}
\begin{proof}
If $\RW(\rho) = \infty$, then Theorems~\ref{thm:nogo_affine_full}--\ref{thm:nogo_sufficient_fulldim_full} already tell us that this transformation is possible with $\ve = 0$, so we assume that $\RW(\rho)<\infty$. Let us then consider the two non-trivial cases: (i) $R^\FF_\FF(\phi) < \infty$, and (ii) the resource theory is affine and $R_\FF(\phi) < \infty$.
Notice that the chosen value of $\ve$ in (i) Eq.~\eqref{eq:ve_choice_achiev} or (ii) Eq.~\eqref{eq:ve_choice_achiev_aff} entails that
\begin{equation}\begin{aligned}\label{eq:onevef}
  \mathrm{(i)}\quad \RW(\rho) &\geq \frac{1-\ve}{\ve} \left(R^\FF_\FF(\phi)-1\right),\\
    \mathrm{(ii)}\quad \RW(\rho) &\geq \frac{1-\ve}{\ve} \left(R_\FF(\phi)-1\right).
\end{aligned}\end{equation}
Let $\sigma \in \FF$ be a state such that (i) $\phi \leq_\FF \lambda \sigma$ with $\lambda = R_\FF^\FF(\phi) $ or (ii) $\phi \leq \lambda \sigma$ with $\lambda = R_\FF(\phi)$, and define the state
\begin{equation}\begin{aligned}\label{eq:tauve_def}
  \tau_\ve \coloneqq&\; (1-\ve) \phi + \ve \frac{\lambda \sigma - \phi}{\lambda - 1}\\
  =& \left[(1-\ve)-\ve \frac{1}{\lambda-1}\right] \phi + \ve \frac{\lambda}{\lambda - 1} \sigma.
\end{aligned}\end{equation}
Let us first assume that $(1-\ve) \geq \ve \frac{1}{\lambda-1}$, i.e.\ $\ve \leq \frac{\lambda-1}{\lambda}$. 
We then have
\begin{equation}\begin{aligned}\label{eq:tauve_nontrivial}
  \mathrm{(i)}\quad \tau_\ve &\leq_\FF \left[(1-\ve)-\ve \frac{1}{\lambda-1}\right] \lambda \sigma + \ve \frac{\lambda}{\lambda - 1} \sigma\\
  &= (1-\ve) \lambda \sigma,\\
  \mathrm{(ii)}\quad \tau_\ve &\leq \left[(1-\ve)-\ve \frac{1}{\lambda-1}\right] \lambda \sigma + \ve \frac{\lambda}{\lambda - 1} \sigma\\
  &= (1-\ve) \lambda \sigma,\\
  \mathrm{(i), (ii)}\quad \tau_\ve &\geq \ve\frac{\lambda}{\lambda-1} \sigma.
\end{aligned}\end{equation}
This gives
\begin{equation}\begin{aligned}\label{eq:twoomegas}
  \mathrm{(i)}\quad \RW^\FF(\tau_\ve) &\leq \frac{1-\ve}{\ve}(\lambda-1)\\
  \mathrm{(ii)}\quad \RW(\tau_\ve) &\leq \frac{1-\ve}{\ve}(\lambda-1),
\end{aligned}\end{equation}
which in particular implies that (i) $\RW^\FF(\tau_\ve) \leq \RW(\rho)$ or (ii) $\RW(\tau_\ve) \leq \RW(\rho)$, as the right-hand sides of~\eqref{eq:onevef} and~\eqref{eq:twoomegas} coincide. Invoking (i) Theorem~\ref{thm:nogo_sufficient_fulldim_full} or (ii) Theorem~\ref{thm:nogo_affine_full} shows the existence of a probabilistic transformation $\rho \transf \tau_\ve$. Since $F(\tau_\ve, \phi) \geq 1-\ve$ by construction, we have thus shown the desired achievability.

We need to consider separately the case of large error, $(1-\ve) < \ve \frac{1}{\lambda-1}$. Here, it holds that
\begin{equation}\begin{aligned}\label{eq:tauve_trivial}
  \mathrm{(i)}\quad &\tau_\ve = (1-\ve) \phi + (1-\ve) (\lambda \sigma - \phi) + \left[  \ve \frac{1}{\lambda-1} - (1-\ve)\right] (\lambda \sigma - \phi)\\
  &\hphantom{\tau_\ve} = (1-\ve) \lambda \sigma +  \left[  \ve \frac{1}{\lambda-1} - (1-\ve)\right] (\lambda \sigma - \phi)\\
  &\Rightarrow \tau_\ve \in \FF
\end{aligned}\end{equation}
since $\phi \leq_\FF \lambda \sigma$ by assumption. Hence, the transformation into $\tau_\ve$ trivialises. In the case (ii), we cannot make the exact same argument, but we can simply change the target state into
\begin{equation}\begin{aligned}
  \tau'_\ve \coloneqq (1-\ve) \phi + (1-\ve) (\lambda \sigma - \phi) +  \left[  \ve \frac{1}{\lambda-1} - (1-\ve)\right] \sigma \in \FF
\end{aligned}\end{equation}
so that $\rho \transf \tau'_\ve$ is again trivially possible.
\end{proof}


Combining Theorem~\ref{thm:RW_dist_achiev} and Theorem~\ref{thm:prob_error_full} gives the bounds
\begin{equation}\begin{aligned}\label{eq:err_bounds_1}
   \left( \frac{1}{R_\FF^\FF(\phi)-1} \,\RW(\rho) + 1 \right)^{-1} \geq \inf_{\E \in \OO} \left[ 1 -  F\left(\frac{\E(\rho)}{\Tr \E(\rho)}, \phi\right)\right] \geq \left( \frac{F_\FF(\phi)}{1-F_\FF(\phi)} \,\RW(\rho)\, + 1 \right)^{-1},
\end{aligned}\end{equation}
which are valid in any convex resource theory, and
\begin{equation}\begin{aligned}\label{eq:err_bounds_2}
   \left( \frac{1}{R_\FF(\phi)-1} \,\RW(\rho) + 1 \right)^{-1} \geq \inf_{\E \in \OO} \left[ 1 -  F\left(\frac{\E(\rho)}{\Tr \E(\rho)}, \phi\right)\right] \geq \left( \frac{F_\FF(\phi)}{1-F_\FF(\phi)} \,\RW(\rho)\, + 1 \right)^{-1},
\end{aligned}\end{equation}
which hold for affine resources. We stress once again that the converse (lower) bounds are valid in any resource theory with any set of free operations contained in $\OO$, while the achievable (upper) bounds are valid specifically for resource--non-generating operations $\OO$.

A natural question then is: when can these bounds coincide, meaning that $\RW$ exactly characterises the least error achievable in distillation of $\phi$?
Comparing the bounds in Eqs.~\eqref{eq:err_bounds_1}--\eqref{eq:err_bounds_2}, we see that a necessary and sufficient condition for the upper and lower bounds to be equal is that $R_\FF^\FF(\phi) = F_\FF(\phi)^{-1}$, or $R_\FF(\phi) = F_\FF(\phi)^{-1}$ in affine resource theories. Remarkably, states satisfying the latter equality exist in \emph{every} convex quantum resource theory --- these are the maximally resourceful (so-called `golden') states~\cite{regula_2020}. Throughout this paper, we understand `maximally resourceful' to mean specifically any state which maximises the robustness $R_\FF$, the motivation for this choice being precisely that such states satisfy $R_\FF(\phi) = F_\FF(\phi)^{-1}$. We formalise the consequences of such a choice as follows.
\begin{boxed}{white}
\begin{corollary}[\cite{regula_2022}]\label{thm:prob_error_tight}
Let $\phim$ be a pure state which maximises the robustness $R_\FF$ among all states of the same dimension. Then, as long as either: 
\begin{enumerate}[(i)]
\item the given resource theory is affine, or 
\item it holds that $R_\FF(\phim) = R^\FF_\FF(\phim)$,
\end{enumerate}
then there exists a resource--non-generating protocol that achieves the bound of Theorem~\ref{thm:prob_error_full}. Specifically, for any input state $\rho$ it holds that
\begin{equation}\begin{aligned}\label{eq:exact_error}
  \inf_{\E \in \OO} \left[ 1 -  F\left(\frac{\E(\rho)}{\Tr \E(\rho)}, \phim\right)\right] = \left( \frac{F_\FF(\phim)}{1-F_\FF(\phim)} \,\RW(\rho)\, + 1 \right)^{-1}\!\!.
\end{aligned}\end{equation}
\end{corollary}
\end{boxed}
The choice of a maximally resourceful state as the distillation target is a natural one in most quantum resources. In the fundamental example of quantum entanglement, $\phim$ can be understood as any maximally entangled state $\sum_{i=1}^m \frac{1}{\sqrt{m}} \ket{ii}$ of some dimension; indeed, these states satisfy $R^\FF_\FF(\phim) = F_\FF(\phim)^{-1}$~\cite{shimony_1995,vidal_1999}, and so we obtain an exact quantitative characterisation of probabilistic distillation of entanglement. The same property holds true in multi-level bipartite entanglement of higher Schmidt rank~\cite{johnston_2018}, genuine multipartite entanglement~\cite{contreras-tejada_2019}, as well as in the resource theory of multi-level quantum coherence~\cite{johnston_2018}. For other theories such as magic-state quantum computation, it is not always the case that $R^\FF_\FF(\phim) = F_\FF(\phim)^{-1}$~\cite{regula_2018}, but one can nevertheless find examples of states for which the bounds in \eqref{eq:err_bounds_1} coincide --- for example, the three-qubit Hoggar state, or the single-qutrit Strange and Norell states~\cite{takagi_2022-1}. For any such target state, \eqref{eq:exact_error} is satisfied. 
We also stress that Corollary~\ref{thm:prob_error_tight} gives an exact expression for the least achievable distillation error under resource--non-generating operations in \emph{any} affine resource theory.

\begin{proof}
The key property that we use is that $R_\FF(\phim) = F_\FF(\phim)^{-1}$~\cite[Theorem~4]{regula_2020}, which holds for any state $\phim$ that maximises the robustness in some dimension (noting that the input $\rho$ and the target $\phim$ can act on spaces of different dimension). If the resource theory is affine, Theorem~\ref{thm:RW_dist_achiev} immediately gives us our desired achievability. For non-affine theories, before we invoke Theorem~\ref{thm:RW_dist_achiev}, we additionally need to ensure that $R^\FF_\FF(\phim) = F_\FF(\phim)^{-1}$ in order to match the bound of Theorem~\ref{thm:prob_error_full}, which is why we require that $R_\FF(\phim) = R^\FF_\FF(\phim)$.
\end{proof}

We can also give an alternative interpretation of the above bounds. For a fixed $\ve > 0$, they can be understood as establishing a limitation on the resourcefulness of any state $\ket{\phi}$ that can be distilled from $\rho$ up to the given error.
\begin{boxed}{white}
\begin{corollary}
  If there exists a transformation $\rho \transf \tau$ such that $\tau$ is a state satisfying $F(\tau,\phi) \geq 1-\ve$ for some resourceful pure state $\phi$, then
\begin{equation}\begin{aligned}
  F_\FF(\phi)^{-1} \leq \frac{\ve}{1-\ve} \,\RW(\rho)\, + 1.
\end{aligned}\end{equation}
Conversely, let $\phim$ be as in Corollary~\ref{thm:prob_error_tight}. Then, for any $\ve$ such that
\begin{equation}\begin{aligned}
  F_\FF(\phim)^{-1} \leq \frac{\ve}{1-\ve} \,\RW(\rho)\, + 1,
\end{aligned}\end{equation}
there exists a resource--non-generating protocol $\rho \transf \tau$ where $\tau$ is a state satisfying $F(\tau,\phim) \geq 1-\ve$.
\end{corollary}
\end{boxed}
To understand this interpretation of the bound, it is helpful to consider the explicit example of the resource theory of entanglement. As mentioned previously, the maximally resourceful states $\ket{\Phi_m} \coloneqq \sum_{i=1}^m \frac{1}{\sqrt{m}} \ket{ii}$ satisfy $R^\FF_\FF(\Phi_m) = F_\FF(\Phi_m)^{-1} = m$, which means that the two directions of the above bound match, and we obtain that
\begin{equation}\begin{aligned}
  \sup \lset m \in \NN \bar \rho \transf \tau,\; F(\tau, \Phi_m) \geq 1-\ve \rset = \floor{ \frac{\ve}{1-\ve} \,\RW(\rho)\, + 1 },
\end{aligned}\end{equation}
with the $\floor{\cdot}$ needed because of the discrete character of the states $\{\Phi_m\}_m$. The above result, stated first in~\cite{regula_2022}, can be understood as the value of entanglement that can be distilled probabilistically up to error $\ve$ --- as we see, the projective robustness $\RW$ of the input state determines this value exactly.

\subsubsection*{The fidelity of an isotropic state cannot be improved}

The construction in the proof of Theorem~\ref{thm:RW_dist_achiev} allows us to show the following no-go result.

\begin{boxed}{white}
\begin{corollary}\label{cor:isotropic_nogo}
Let $\phi$ be a pure state such that $R_\FF(\phi) = F_\FF(\phi)^{-1} \eqqcolon \lambda$, for example any maximally resourceful state. Fix $\ve \in (0,1)$ and define $\rho_\ve$ as the state
\begin{equation}\begin{aligned}\label{eq:isotropic}
  \rho_\ve \coloneqq (1-\ve) \phi + \ve \frac{\lambda \sigma - \phi}{\lambda - 1},
\end{aligned}\end{equation}
where $\sigma \in \FF$ is a state such that $\phi \leq \lambda \sigma$.

Then, there does not exist any free transformation protocol --- probabilistic or deterministic --- such that $\rho_\ve \transf \tau$, where $\tau$ is any state satisfying $F(\tau, \phi) \geq 1-\ve' > 1-\ve$. In other words, the fidelity of $\rho_\ve$ with the state $\phi$ cannot be increased by any free operation.
\end{corollary}
\end{boxed}
\begin{proof}
Following the proof of Theorem~\ref{thm:RW_dist_achiev}, we see that
\begin{equation}\begin{aligned}
  \RW(\rho_\ve) &\leq \frac{1-\ve}{\ve}(\lambda-1).
\end{aligned}\end{equation}
Now, if there existed a transformation $\rho_\ve \transf \tau$, then by Theorem~\ref{thm:prob_error_full} we would have
\begin{equation}\begin{aligned}
  \RW(\rho_\ve) &\geq \frac{(1-\ve')(1- F_\FF(\phi))}{\ve' \, F_\FF(\phi)}\\
  &= \frac{1-\ve'}{\ve'} (\lambda-1)\\
  &> \frac{1-\ve}{\ve}(\lambda-1),
\end{aligned}\end{equation}
which is a contradiction.
\end{proof}

The class of states $\rho_\ve$ constructed in~\eqref{eq:isotropic} includes in particular the isotropic states~\cite{horodecki_1999-1} encountered in the theories of entanglement and coherence, and indeed also more general `isotropic-like' constructions that can be found in other resource theories~\cite{takagi_2022-1}. Corollary~\ref{cor:isotropic_nogo} gives an explicit proof of a fact that we have already seen in Figure~\ref{fig:iso} --- namely, that the fidelity of entanglement distillation of an isotropic state cannot be improved --- and provides an extension of the results of Refs.~\cite{linden_1998,kent_1998} that showed an equivalent property for entanglement manipulation under LOCC.

States $\rho_\ve$ satisfying the conditions of the Corollary exist in any convex resource theory, since we can always take a maximally resourceful golden state $\phim$ of any dimension. 
But note that golden states are not the only ones for which it may hold that $R_\FF(\phi) = F_\FF(\phi)^{-1}$. For instance, in the resource theory of magic, the so-called Clifford magic states~\cite{bravyi_2019} also satisfy this property; this includes many relevant magic states of interest such as the $T$ state.

We remark that the assumptions on $\phi$ imply that $\phi$ and $\frac{\lambda \sigma - \phi}{\lambda - 1}$ are necessarily orthogonal to each other.


\section{Bounding probability and error trade-offs}\label{sec:tradeoff}

We have seen that the projective robustness can be used to establish general limitations on probabilistic resource transformations. However, since $\RW$ is monotonic under all probabilistic protocols, the restrictions obtained in this way can be considered too strong in some applications: one might instead be interested in understanding the least error achievable with a high probability of success, or in an exact quantification of the best achievable probability in a given transformation.
To address this, we consider a complementary approach for the characterisation of probabilistic transformations, with the explicit aim of tightly bounding the achievable performance in transforming two quantum states probabilistically under some constraints on the achievable probability or error. Although seemingly different, we will see that the approach is in fact closely related to the methods we used to study the projective robustness.

\subsection{General bounds on achievable probability}

We first aim to bound the maximal probability of success in the transformation between two states,
\begin{equation}\begin{aligned}
  P(\rho \transf \rho') \coloneqq \max \lset p \bar \E (\rho) = p \rho',\; \E \in \OO \rset.
\end{aligned}\end{equation}
As before, any protocol $\E$ here can be thought of as being a part of a larger instrument $\{\E_i\}$, with each $\E_i$ a free operation.

Our main idea will be to construct measure-and-prepare maps that can perform a given transformation $\rho \transf \rho'$, taking suitable care to exactly quantify the probability with which the transformation is realised. Specifically, we will employ a class of transformations that we already encountered in the proof of Theorems~\ref{thm:nogo_affine_full} and \ref{thm:nogo_sufficient_fulldim_full}, namely, maps of the form
\begin{equation}\begin{aligned}\label{eq:meas_prep}
   \E(X) = \< W, X \> \phi + \< Q, X \> \frac{\lambda\, \sigma - \phi}{\lambda - 1}
 \end{aligned}\end{equation}
 for some operators $W$ and $Q$.
Using such operations, we will be able to provide tight bounds on the achievable probability and fidelity of probabilistic transformations. Corresponding upper bounds will then be used to verify the tightness of the achievable bounds for relevant cases.

The methods studied here will be particularly suited to target states $\rho'$ that satisfy $\RW(\rho')=\infty$, such as resource distillation or conversion into states of low rank. Whether the approach can be extended also to other types of transformations, such as conversion between general states of full rank, remains an open question. 

To begin, we need to define another resource measure, the overlap
 \begin{equation}\begin{aligned}
   V_\FF(\rho) \coloneqq \max_{\sigma \in \FF}\, \< \Pi_\rho,  \sigma \>
 \end{aligned}\end{equation}
 with $\Pi_\rho$ denoting the projection onto $\supp \rho$. This quantity is a non-logarithmic variant of the min-relative entropy~\cite{datta_2009}, and for pure states it reduces to the quantity $F_\FF$ that we encountered before, that is $V_\FF(\phi) = F_\FF(\phi)$. The second measure that we will use is the standard robustness $R_\FF^\FF(\rho) = \min_{\sigma \in \FF} \Rmax^\FF (\rho \| \sigma)$ that we encountered before.

As a motivating example, define the map
\begin{equation}\begin{aligned}
   \E(X) \coloneqq \< \mu \Pi_\rho, X \> \rho' + \<\id - \Pi_\rho, X \> \frac{R^\FF_\FF(\rho')\, \sigma - \rho'}{R^\FF_\FF(\rho') - 1},
 \end{aligned}\end{equation}
 where $\sigma$ is a free state such that $\rho' \leq_\FF R^\FF_\FF(\rho')\, \sigma$, and $\mu$ is some suitable parameter. Notice that this is a valid probabilistic quantum operation for any $\mu \in [0,1]$. In particular, choosing $\mu = \frac{1}{R^\FF_\FF(\rho')-1}\frac{1 - V_\FF(\rho)}{V_\FF(\rho)}$ (or $\mu=1$, whichever is smaller) gives for any $\pi \in \FF$ with $\< \Pi_\rho, \pi \> \neq 0$ that
 \begin{equation}\begin{aligned}
   \E(\pi) &\propto\,  \rho' + \frac{1}{\mu} \,\frac{1 - \<\Pi_\rho,\pi\>}{\< \Pi_\rho, \pi\>} \,\frac{R^\FF_\FF(\rho')\sigma-\rho'}{R^\FF_\FF(\rho')-1}\\
   &\geq_\FF \rho' + \frac{1}{\mu} \,\frac{1 - V_\FF(\rho)}{V_\FF(\rho)} \,\frac{R^\FF_\FF(\rho')\sigma-\rho'}{R^\FF_\FF(\rho')-1}\\
   &\geq_\FF \rho' + R^\FF_\FF(\rho')\sigma - \rho'\\
   &\propto\, \sigma \in \FF
 \end{aligned}\end{equation}
where we used that $R^\FF_\FF(\rho')\sigma-\rho' \geq_\FF 0$ by definition of $R^\FF_\FF$, and we note that the third line is actually an equality when $\mu\neq 1$. In the case that $\< \Pi_\rho, \pi \> = 0$, we also get that $\E(\pi) \geq_\FF 0$ from the aforementioned fact that $R^\FF_\FF(\rho')\sigma-\rho' \geq_\FF 0$.  This means that $\E \in \OO$, implying that
\begin{equation}\begin{aligned}\label{eq:Pbound_lower}
  P(\rho \transf \rho') \geq \min \left\{ 1, \, \frac{1}{R^\FF_\FF(\rho')-1} \frac{1 - V_\FF(\rho)}{V_\FF(\rho)} \right\}.
\end{aligned}\end{equation}
Note that the quantifier $V_\FF(\rho)$ satisfies $V_\FF(\rho) = 1$ if and only if there exists a free state $\sigma$ such that $\supp \sigma \subseteq \supp \rho$; this is equivalent to $W_\FF(\rho) > 0$.
This provides an alternative way to prove a result that we have already seen in Lemma~\ref{lem:achiev_omega}: for any state $\rho$ s.t.\ $W_\FF(\rho)=0$ and any target $\rho' \notin \FF$ such that $R^\FF_\FF(\rho')<\infty$, it holds that $P(\rho \transf \rho') > 0$, and so any such state $\rho'$ is reachable with some probability. The bound of Eq.~\eqref{eq:Pbound_lower} was recently shown in other works for the case of pure-state inputs $\rho = \proj\psi$ (which always have vanishing $W_\FF$), first in the resource theory of entanglement~\cite{contreras-tejada_2019} and later extended to multilevel quantum coherence~\cite{zhang_2021}.

To approach this problem more generally, let us consider the optimisation program
 \begin{align}
   \HP(\rho \gbar t)\coloneqq \max \Big\{ \<\rho, W \>& \;\Big|\; 0 \leq W \leq Z \leq \id,\nonumber \\
   &\< W, \rho \> = \< Z, \rho \>,\\
   & \<W, \sigma \> \leq \frac{1}{t} \< Z, \sigma \> \; \forall \sigma\in\F \Big\},\nonumber
 \end{align}
 where $t \in \RR_+ \cup \{+\infty\}$ is some chosen parameter. The problem is clearly feasible for any $t$ (take $W=Z=0$) and the feasible sets are both compact, justifying the use of $\max$ instead of $\sup$. Once again, in many theories of interest this quantity can be computed as an SDP. A quantity of this form previously appeared in~\cite{rozpedek_2018} for the resource theory of entanglement and in~\cite{fang_2018} for the case of quantum coherence, and a related optimisation problem was used in~\cite{buscemi_2017} in the characterisation of transformations of pairs of quantum states. We then have the following.

 \begin{boxed}{white}
 \begin{theorem}\label{thm:Hbound}
 In any convex resource theory, the probability of transforming $\rho$ into another state $\rho'$ with resource--non-generating operations is bounded as
 \begin{equation}\begin{aligned}\label{eq:H_bounds}
   \HP\!\left(\rho\gbar V_\FF(\rho')^{-1} \right) \geq P(\rho \transf \rho') \geq \HP\!\left(\rho\gbar R^\FF_\FF(\rho') \right).
 \end{aligned}\end{equation}
 \end{theorem}
 \end{boxed}
The optimisation problem $\HP(\rho, t)$ can therefore be used to provide both lower and upper bounds on achievable probabilities. We will give a proof of this result shortly.

Note that the lower bound in Theorem~\ref{thm:Hbound} above relies on the standard robustness $R^\FF_\FF$. An issue which arises here is that this quantity is not guaranteed to be finite in all resource theories, and in particular it is \emph{never} finite for resourceful states in affine theories, trivialising Theorem~\ref{thm:Hbound}. This suggests that another approach is needed for such resources. Indeed, we consider a modified optimisation problem,
 \begin{align}\label{eq:HPaff}
   \HP^{\aff}(\rho \gbar t)\coloneqq \max \Big\{ \<\rho, W \>& \;\Big|\; 0 \leq W \leq Z \leq \id,\nonumber \\
   &\< W, \rho \> = \< Z, \rho \>,\\
   & \<W, \sigma \> = \frac1t \< Z, \sigma \> \; \forall \sigma\in\F \Big\}.\nonumber
 \end{align}
Note that the only difference between $\HP$ and $\HP^{\aff}$ is the equality in the last line of~\eqref{eq:HPaff}; similarly to the case of $\Omega_{\aff(\FF)}$ encountered in the proof of Theorem~\ref{thm:nogo_affine_full}, $\HP^{\aff}(\rho, t)$ can be understood as $\HP(\rho, t)$ where the optimisation over the set $\FF$ is replaced with the affine hull $\aff(\FF)$. Using this modification, we will be able to obtain a bound that uses the robustness $R_\FF$ instead of the standard robustness $R^\FF_\FF$, making it applicable to resources where  $R^\FF_\FF$ diverges.
 \begin{boxed}{white}
 \begin{theorem}\label{thm:Hbound_aff}
 In any convex resource theory, the probability of transforming $\rho$ into another state $\rho'$ with resource--non-generating operations is bounded as
 \begin{equation}\begin{aligned}
   \HP^{\aff}\!\left(\rho\gbar V_{\aff(\FF)}(\rho')^{-1} \right) \geq P(\rho \transf \rho') \geq \HP^{\aff}\!\left(\rho\gbar R_\FF(\rho') \right)\!.
 \end{aligned}\end{equation}
 \end{theorem}
 \end{boxed}
 Here, we defined
\begin{equation}\begin{aligned}
  \Qmaxa(\rho') \coloneqq& \sup_{X \in \aff(\FF)} \< \Pi_{\rho'}, X \>\\
=& \begin{cases} c & \text{ if } \< \Pi_{\rho'}, \sigma \> = c \; \forall \sigma \in \FF\\
\infty & \text{ otherwise},\end{cases}
\end{aligned}\end{equation}
where the second line can be seen as in the proof of Lemma~\ref{lem:affine}. 
The bounds of this Theorem are useful particularly for affine resource theories, as they will typically trivialise when the theory is not affine.

We will find it convenient to split the proofs of Theorems~\ref{thm:Hbound} and \ref{thm:Hbound_aff} into two Propositions that establish upper and lower bounds.
\begin{proposition}\label{prop:Hbound_upper}
  For any states $\rho, \rho'$, it holds that
  \begin{equation}\begin{aligned}
    P(\rho \to \rho') &\leq \HP\left(\rho \gbar \Qmax(\rho')^{-1}\right),\\
    P(\rho \to \rho') &\leq \HP^{\aff}\left(\rho \gbar \Qmaxa(\rho')^{-1}\right).
  \end{aligned}\end{equation}
\end{proposition}
\begin{proof}
Assume that there exists a probabilistic free operation such that $\E(\rho) = p \rho'$ for some $p \in (0,1]$. Then
\begin{equation}\begin{aligned}
  p = \< \Pi_{\rho'}, \E(\rho) \> = \< \E^\dagger(\Pi_{\rho'}), \rho \>,
\end{aligned}\end{equation}
where $\E^\dagger$ is the adjoint map of $\E$. Using that $\E^\dagger(\id) \leq \id$ since $\E$ is positive and trace--non-increasing, we have $\E^\dagger(\Pi_{\rho'}) \leq \E^\dagger(\id) \leq \id$. Additionally, we can notice that $\<\E^\dagger(\id - \Pi_{\rho'}), \rho \> = 0$. Now, for any $\sigma \in \F$, it holds that
\begin{equation}\begin{aligned}
  \< \E^\dagger(\Pi_{\rho'}), \sigma \> &= \< \Pi_{\rho'}, \E(\sigma) \>\\
  &= \Tr\E(\sigma) \< \Pi_{\rho'}, \frac{\E(\sigma)}{\Tr\E(\sigma)} \>\\
  & \leq \Tr\E(\sigma) \Qmax(\rho').
\end{aligned}\end{equation}
Together with $\< \E^\dagger(\id), \sigma \> = \Tr\E(\sigma)$, this gives
\begin{equation}\begin{aligned}
  \< \E^\dagger(\Pi_{\rho'}), \sigma \> \leq \Qmax(\rho') \< \E^\dagger(\id), \sigma \>,
\end{aligned}\end{equation}
which means that the operators $\E^\dagger(\Pi_{\rho'})$ and $\E^\dagger(\id)$ are feasible solutions to the optimisation problem $\displaystyle \HP\left(\rho\gbar \Qmax(\rho')^{-1}\right)$ with optimal value $p$, which is the stated result.

An analogous argument holds for $\HP^{\aff}$. If $\Qmaxa(\rho') < \infty$, then
\begin{equation}\begin{aligned}
 \< \E^\dagger(\Pi_{\rho'}), \sigma \> &= \Tr\E(\sigma) \Qmaxa(\rho'),\\
\< \E^\dagger(\id), \sigma \> &= \Tr\E(\sigma)
\end{aligned}\end{equation}
for all $\sigma \in \FF$ due to the definition of $\Qmaxa$, so we can use these operators as feasible solutions for $\HP^{\aff}$. In the case that $\Qmaxa(\rho')$ diverges, we can understand this result as the trivial bound $P(\rho \to \rho') \leq 1$.
\end{proof}

\begin{proposition}\label{prop:Hbound_lower}
  For any states $\rho, \rho'$, it holds that
  \begin{equation}\begin{aligned}
    P(\rho \to \rho') &\geq \HP\left(\rho \gbar R^\FF_\FF(\rho') \right),\\
    P(\rho \to \rho') &\geq \HP^{\aff}\left(\rho \gbar R_\FF(\rho') \right).
  \end{aligned}\end{equation}
\end{proposition}
\begin{proof}
Consider the case of $\HP$ first, and let $W, Z$ be optimal operators in the definition of $\HP\left(\rho \gbar R^\FF_\FF(\rho')\right)$. If $R^\FF_\FF(\rho') = \infty$, then $\<W, \sigma\> = 0 \; \forall \sigma \in \FF$ and the map $\E(X) \coloneqq \< W, X \> \rho'$ satisfies $\E \in \OO$ trivially, so we are done. Otherwise, define the completely positive and trace--non-increasing map
\begin{equation}\begin{aligned}
  \E(X) \coloneqq \< W, X \> \rho' + \< Z-W, X \> \frac{\lambda \sigma - \rho'}{\lambda - 1}
\end{aligned}\end{equation}
where $\sigma$ is an optimal state such that $\rho' \leq_\FF \lambda \sigma$ with $R^\FF_\FF(\rho') = \lambda$. Then
\begin{equation}\begin{aligned}
  \E(\rho) &= \< W, \rho \> \rho'\\
  &= \HP\!\left(\rho \gbar R^\FF_\FF(\rho')\right) \rho'
\end{aligned}\end{equation}
since $\<Z -W , \rho \> = 0$, and for any $\sigma \in \FF$ it holds that
\begin{equation}\begin{aligned}
  \E(\sigma) &= \< W, \sigma \> \rho' + \< Z - W, \sigma \> \frac{\lambda \sigma - \rho'}{\lambda - 1}\\
  &\geq_\FF \< W, \sigma \> \rho' + (\lambda - 1) \< W, \sigma \> \frac{\lambda \sigma - \rho'}{\lambda - 1}\\
  &= \< W, \sigma \> \lambda \sigma \in \cone(\FF),
\end{aligned}\end{equation}
where we used that $\frac{1}{\lambda} \<Z, \sigma\> \geq \< W, \sigma\>$ by definition of $\HP$. This shows that $\E$ is a free probabilistic operation achieving the desired transformation.

An analogous argument holds for $\HP^{\aff}$, where we now construct
\begin{equation}\begin{aligned}
  \E(X) \coloneqq \< W, X \> \rho' + \< Z - W, X \> \frac{\lambda \sigma - \rho'}{\lambda - 1}
\end{aligned}\end{equation}
with $\rho' \leq \lambda \sigma$, or $\E(X) \coloneqq \< W, X \> \rho'$ if $R_\FF(\rho')=\infty$.
\end{proof}

\subsection{Trade-offs in probabilistic resource distillation}

The fact that both our upper and lower bounds for transformation probabilities are obtained in terms of the same optimisation problem, $\HP(\rho, t)$ or $\HP^{\aff}(\rho, t)$, raises the question: when can the bounds be tight, characterising the achievable performance exactly? Similarly to how we established the tightness result in Section~\ref{sec:prob_dist_RW}, we will show that such an equality can be obtained in the problem of resource distillation.

We can first consider the purifying transformation $\rho \transf \phi$ to be achieved exactly, and ask about the best possible probability of success.
Recall that any convex resource theory admits a meaningful notion of a maximally resourceful state: a golden state $\phim$, which maximises the robustness $R_\FF$ among all states, also minimises the overlap $V_\FF$ and satisfies $R_\FF(\phim) = V_\FF(\phim)^{-1}$~\cite{regula_2020}.
We immediately see that, for such states, the upper and lower bounds of Theorems~\ref{thm:Hbound} and \ref{thm:Hbound_aff} can become equal, establishing an exact expression for the probability $P(\rho \transf \phim)$. Specifically, this is the case when either $R_\FF(\phim) = R^\FF_\FF(\phim)$ or $V_\FF(\phim) = V_{\aff(\FF)}(\phim)$, which are conditions obeyed in theories such as quantum entanglement, coherence, or athermality.

We can extend this result to approximate distillation, which we encountered before in Section~\ref{sec:prob_dist_RW}. Here, the output of the protocol is a state $\tau$ that is not necessarily pure itself, but is close to the pure distillation target $\phi$. We thus define
\begin{equation}\begin{aligned}
  P(\rho \transf \phi, \ve) &\coloneqq \max \big\{ p \;\big|\; \E(\rho) = p \tau,\; \E \in \OO, \; F(\tau, \phi) \geq 1-\ve \big\}.
\end{aligned}\end{equation}
Alternatively, turning the problem around, we can ask: for a given probability $p$, what is the best achievable distillation error $\ve$? Let us denote this quantity as
\begin{equation}\begin{aligned}
  E(\rho \transf \phi,\, p) \coloneqq \min \big\{ \ve \;\big|\;& \E (\rho) = p \tau,\; \E \in \OO,\; F(\tau,\phi) \geq 1-\ve \big\}.
\end{aligned}\end{equation}

In order to characterise this setting, we will need to consider a modification of the quantities $\HP$ and $\HP^{\aff}$. First, for a fixed $\ve$, we define
\begin{equation}\begin{aligned}
  \HP_\ve(\rho \gbar t)\coloneqq \max \Big\{ \< Z, \rho \> \;\Big|&\; 0 \leq W \leq Z \leq \id,\\
   &  \< W, \rho \> = (1-\ve)\< Z, \rho \>,\\
   & \<W, \sigma \> \leq \frac1t \< Z, \sigma \> \; \forall \sigma\in\F \big\},\\
   \HP_\ve^{\aff}(\rho \gbar t)\coloneqq \max \Big\{ \< Z, \rho \> \;\Big|&\; 0 \leq W \leq Z \leq \id,\\
   & \< W, \rho \> = (1-\ve) \< Z, \rho \> ,\\
   & \<W, \sigma \> = \frac1t \< Z, \sigma \> \; \forall \sigma\in\F \Big\}.
\end{aligned}\end{equation}
As expected, for $\ve = 0$ we recover $\HP(\rho, k)$ and $\HP^{\aff}(\rho, k)$. Optimisation problems of this form have been previously used to characterise probabilistic distillation of entanglement~\cite{rozpedek_2018} and quantum coherence~\cite{fang_2018}. The other family of programs that we will consider is defined, for a fixed $p$, as
\begin{equation}\begin{aligned}\label{eq:thetap_def}
  \Theta_p(\rho, t) \coloneqq \max \Bigg\{ \frac{\<W,\rho\>}{p} \;\Bigg|&\; 0 \leq W \leq Z \leq \id,\\
   &  \< Z, \rho \> = p,\\
   & \<W, \sigma \> \leq \frac1t \< Z, \sigma \> \; \forall \sigma\in\F \Bigg\},\\
     \Theta_p^{\aff}(\rho, t) \coloneqq \max \Bigg\{ \frac{\<W,\rho\>}{p} \;\Bigg|&\; 0 \leq W \leq Z \leq \id,\\
   &  \< Z, \rho \> = p,\\
   & \<W, \sigma \> = \frac1t \< Z, \sigma \> \; \forall \sigma\in\F \Bigg\}.\\
\end{aligned}\end{equation}
Both $\HP_\ve$ and $\Theta_p$ are, once again, convex optimisation problems, and they reduce to semidefinite programs in resource theories of coherence, non-positive partial transpose, athermality, and similar.

Of note is the fact that, were we to constrain the considered operations to be deterministic (trace-preserving), this would lead to fixing $Z = \id$ in our constructions, in which case \eqref{eq:thetap_def} would reduce to the quantity $G(\rho,t)$ previously defined in~\cite{brandao_2010,regula_2020} and closely related to the hypothesis testing relative entropy~$D^\ve_H$~\cite{buscemi_2010,wang_2012}.

Our final result then uses these quantities to provide tight bounds on the probability and error of approximate probabilistic distillation in general resource theories.
\begin{boxed}{white}
\begin{theorem}\label{thm:prob_ve_bounds}
For any input state $\rho$, any resourceful pure state $\phi$, and any error $\ve \in [0,1]$ in any convex resource theory, the maximal probability achievable with resource--non-generating operations is bounded as
\begin{align}\label{eq:prob_ve_bounds}
  \HP_\ve\!\left(\rho\gbar V_\FF(\phi)^{-1} \right) &\geq P(\rho \transf \phi, \ve) \geq \HP_\ve\!\left(\rho,\, R^\FF_\FF(\phi) \right),\\
\label{eq:prob_ve_bounds_aff}
   \HP_\ve^{\aff}\!\left(\rho\gbar V_{\aff(\FF)}(\phi)^{-1} \right) &\geq P(\rho \transf \phi, \ve) \geq \HP_\ve^{\aff}\!\left(\rho\gbar R_\FF(\phi) \right)\!.
\end{align}
On the other hand, the smallest achievable error for a given probability $p \in (0,1]$ is bounded as
\begin{align}\label{eq:ve_bounds}
  \Theta_p\!\left(\rho\gbar V_\FF(\phi)^{-1} \right) &\geq 1-E(\rho\transf\phi, p) \geq \Theta_p\!\left(\rho,\, R^\FF_\FF(\phi) \right)\\
\label{eq:ve_bounds_aff}
   \Theta_p^{\aff}\!\left(\rho\gbar V_{\aff(\FF)}(\phi)^{-1} \right) &\geq 1-E(\rho\transf\phi, p) \geq \Theta_p^{\aff}\!\left(\rho\gbar R_\FF(\phi) \right)\!.
\end{align}

In particular, when $\phi = \phim$ is a golden state (i.e.\ it maximises $R_\FF$ among all states of the same dimension) and it holds that either
\begin{enumerate}[(i)]
\item $R_\FF(\phim) = R^\FF_\FF(\phim)$, or
\item $V_\FF(\phim) = V_{\aff(\FF)}(\phim)$,
\end{enumerate}
then the inequalities in (i) Eqs.~\eqref{eq:prob_ve_bounds}, \eqref{eq:ve_bounds} or in (ii) Eqs.~\eqref{eq:prob_ve_bounds_aff}, \eqref{eq:ve_bounds_aff} are all equalities.
\end{theorem}
\end{boxed}
This gives an exact expression for the best achievable probability $P(\rho\transf\phim, \ve)$ or error $E(\rho\transf\phim, p)$ of approximate one-shot distillation in all resource theories satisfying the specified conditions. For coherence theory, we then recover a result of~\cite{fang_2018}; for quantum entanglement, we establish a general expression for approximate distillation using all non-entangling operations (or, similarly, all operations preserving the set of states with positive partial transpose (PPT), which closely resembles a result of~\cite{rozpedek_2018}).

The Theorem allows us to study exactly the trade-offs between the error $\ve$ incurred in distillation and the probability $p$ with which the process is realised. We demonstrate this in Figure~\ref{fig:ent}, where we characterise the distillation of entanglement under PPT-preserving operations. Since condition (i) is satisfied in this theory (see e.g.~\cite{regula_2020}), we have that $1 - E(\rho \transf \phi_m, p) = \Theta_p(\rho, m)$ for any input state $\rho$ and $m$-dimensional maximally entangled state $\phi_m$. In the Figure, we show how the result of Theorem~\ref{thm:prob_ve_bounds} can lead to insights about the practical manipulation of quantum states by explicitly studying the distillability of two states that were found to be difficult (or impossible) to distill under LOCC in Refs.~\cite{horodecki_1999-2,horodecki_1999-1}.

Note that we have already encountered the quantity $E(\rho \transf \phi, p)$ in Figure~\ref{fig:iso}, where we plotted the best achievable error of entanglement distillation for fixed choices of probabilities. As Theorem~\ref{thm:prob_ve_bounds} applies to this theory, we were justified in using the quantity $\Theta_p$ to exactly quantify the performance of probabilistic distillation protocols there.

The proof of Theorem~\ref{thm:prob_ve_bounds} will be established in a similar manner as before. We also split it into two propositions.

\begin{figure*}[t]
 \centering
        \hspace*{-2pt}
    \subfloat[][$\displaystyle \rho_{a} = \frac{3}{4} \phi_3 + \frac14 \proj{01}$]{%
    \includegraphics[width=8cm]{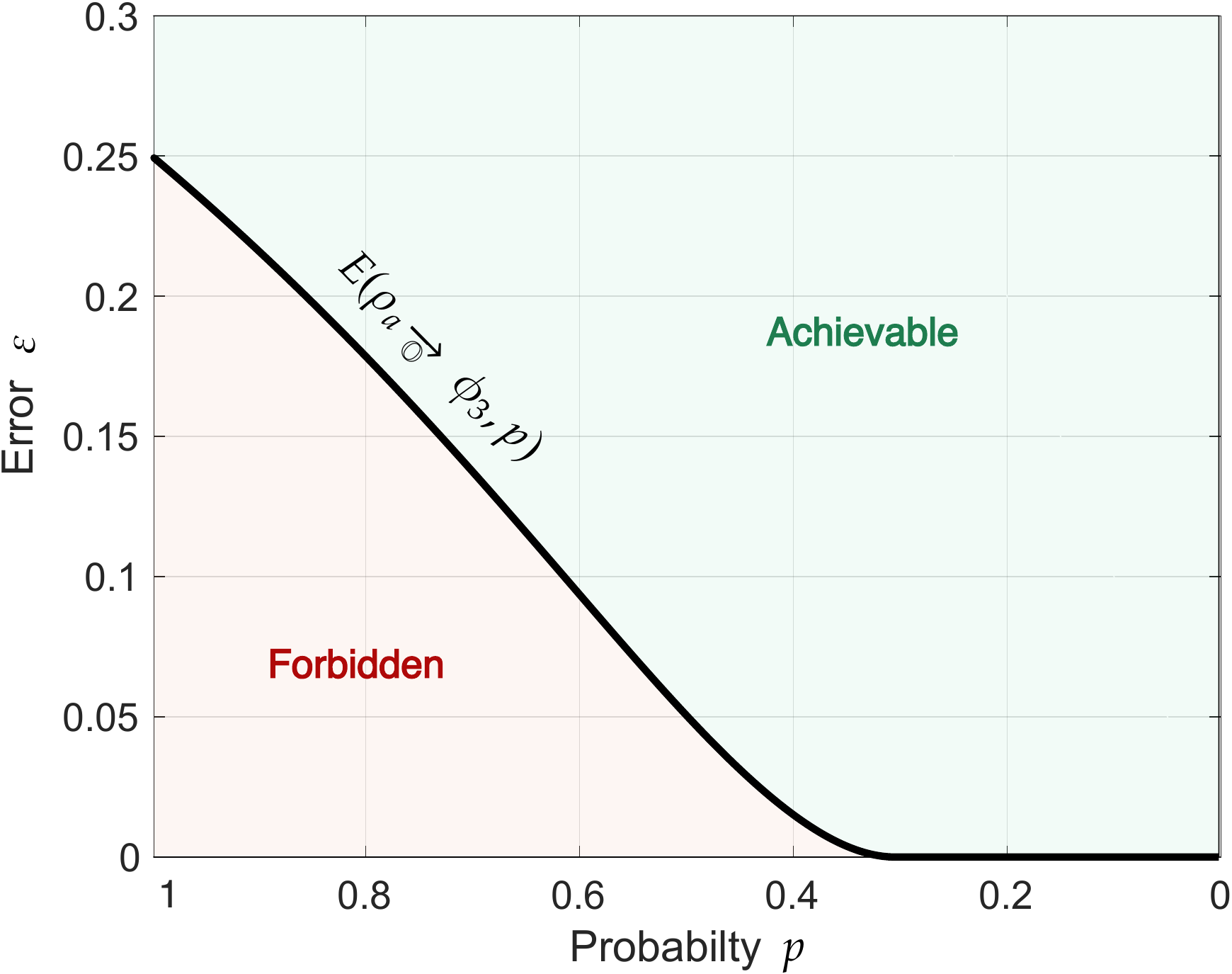}
    \hspace*{10pt}
    }%
    \subfloat[][$\displaystyle \rho_{b} = \frac{3}{4} \phi_3 + \frac{1}{12} \big(\!\proj{01} + \proj{12} + \proj{20}\!\big)$]{%
    \includegraphics[width=7.95cm]{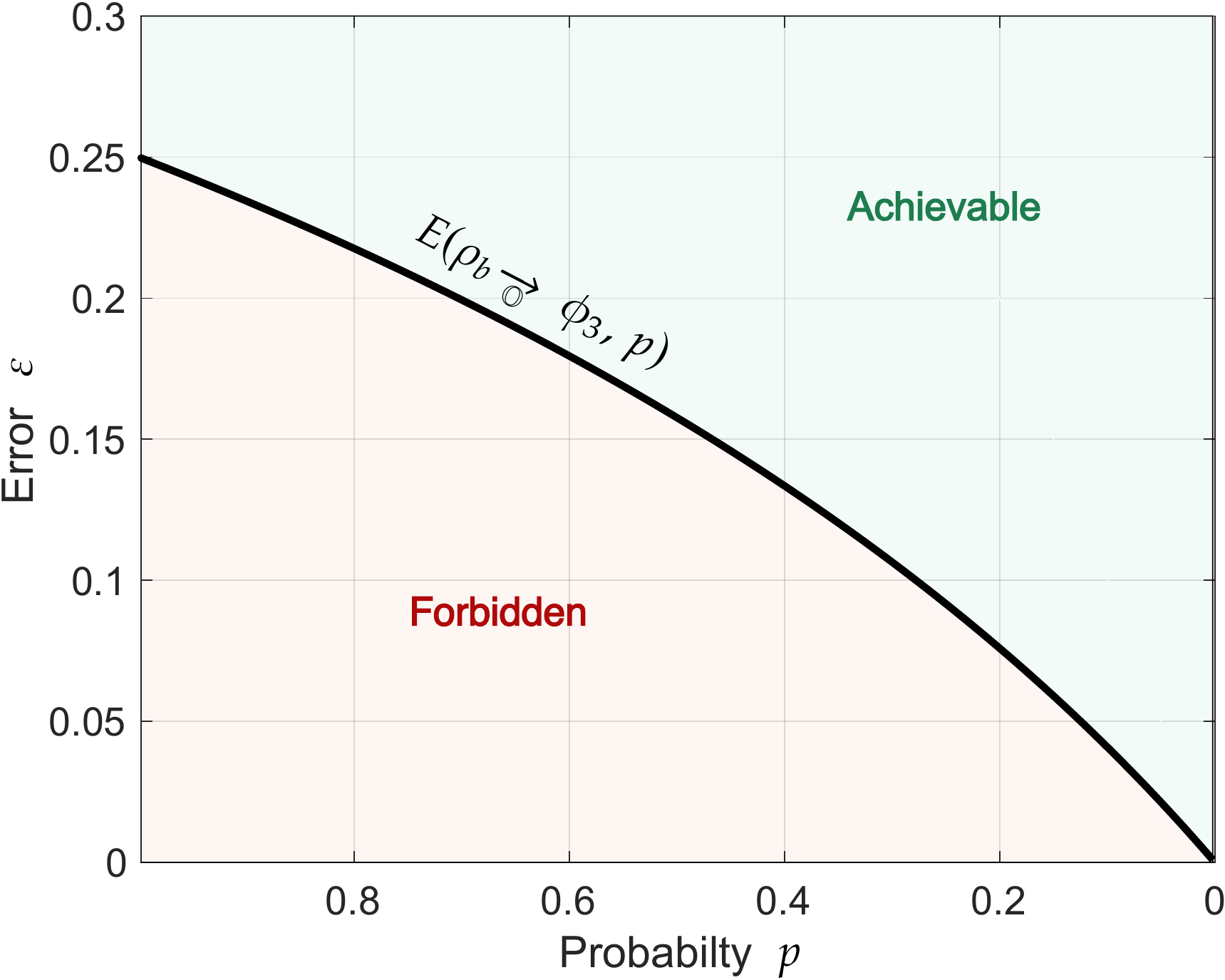}
    }%
\caption{\textbf{Trade-offs in entanglement distillation.} We plot the exact achievable error $E(\rho \to \phi_3, p)$ in the probabilistic one-shot distillation under PPT-preserving operations, that is, the resource theory where the free states $\FF$ are all PPT states. We focus on two types of two-qutrit states that found use in the characterisation of entanglement distillation in Refs.~\cite{horodecki_1999-2,horodecki_1999-1}. The state $\rho_{a}$ in (a) was found to be impossible to distill under LOCC with error $\ve=0$ unless the probability $p$ goes to 0~\cite{horodecki_1999-1}; here, we see that under PPT-preserving operations, it suffices to take $p=0.3$ to achieve perfect distillation. The state $\rho_b$ in (b) is an example of a state that is impossible to distill under LOCC, even probabilistically, in the sense that $\ve$ is bounded away from 0~\cite{horodecki_1999-2}. However, it is already known that assisting the distillation process with PPT states (or, equivalently, PPT operations) makes it possible to distill $\phi_3$ from this state with $\ve=0$~\cite{horodecki_1999-1}. We indeed see this to be the case, but only when $p \to 0$. The state $\rho_b$ is thus an example of a state such that $\RW(\rho_b) = \infty$, but the transformation $\rho_b \to \phi_3$ is only possible with asymptotically vanishing probability, even under the class $\OO$ of all PPT-preserving operations.
}
\label{fig:ent}
\end{figure*}

\begin{proposition}
For any input state $\rho$, any resourceful pure state $\phi$, and any error $\ve \in [0,1]$, it holds that
\begin{equation}\begin{aligned}
  P(\rho \transf \phi, \ve) &\leq \HP_\ve\!\left(\rho\gbar V_\FF(\phi)^{-1} \right)\\
  P(\rho \transf \phi, \ve) &\leq \HP_\ve^{\aff}\!\left(\rho\gbar V_{\aff(\FF)}(\phi)^{-1} \right)\\
  1-E(\rho \transf \phi, p) &\leq \Theta_p\!\left(\rho\gbar V_\FF(\phi)^{-1} \right)\\
  1-E(\rho \transf \phi, p) &\leq \Theta^{\aff}_p\!\left(\rho\gbar V_{\aff(\FF)}(\phi)^{-1} \right).
\end{aligned}\end{equation}
\end{proposition}
\begin{proof}
We begin by noting that the condition $\< W, \rho \> = (1-\ve) \< Z, \rho \>$ in the definition of $\HP_\ve$ can be relaxed to $\< W, \rho \> \geq (1-\ve) \< Z, \rho \>$ without loss of generality, as the optimal value will necessarily be the same. To see this, note that if there existed optimal $Z$ and $W$ with $\< W, \rho \> > (1-\ve) \< Z, \rho\>$, then we could choose another feasible $Z'$ with $\< Z', \rho\> > \<Z,\rho\>$, which would contradict the optimality of $Z$.

Now, let $\E \in \OO$ be any protocol such that $\E(\rho) = p \tau$ for some $\tau$ such that $\< \phi, \tau \> \geq 1-\ve$. This gives
\begin{equation}\begin{aligned}
  \< \E^\dagger(\phi), \rho \> &\geq p (1-\ve) = \< \E^\dagger(\id), \rho \> (1-\ve),\\
  \< \E^\dagger(\id), \rho \> &= p.
\end{aligned}\end{equation}
Using also the fact that
\begin{equation}\begin{aligned}
  \< \E^\dagger(\phi), \sigma \> &= \Tr \E(\sigma) \< \phi, \frac{\E(\sigma)}{\Tr \E(\sigma)} \>\\
  &\leq \Tr \E(\sigma) V_{\FF}(\phi),
\end{aligned}\end{equation}
we then have that $W = \E^\dagger(\phi)$ and $Z = \E^\dagger(\id)$ constitute feasible solutions for $\HP_\ve\!\left(\rho\gbar V_\FF(\phi)^{-1} \right)$ and $\Theta_p\!\left(\rho\gbar V_\FF(\phi)^{-1} \right)$, as was to be shown. An analogous argument holds for $\HP_\ve^{\aff}\!\left(\rho\gbar V_{\aff(\FF)}(\phi)^{-1} \right)$ and $\Theta_p\!\left(\rho\gbar V^{\aff}_\FF(\phi)^{-1} \right)$.
\end{proof}

\begin{proposition}
For any input state $\rho$, any resourceful pure state $\phi$, and any error $\ve \in [0,1]$, it holds that
\begin{equation}\begin{aligned}
  P(\rho \transf \phi, \ve) &\geq \HP_\ve\!\left(\rho,\, R^\FF_\FF(\phi) \right)\\
  P(\rho \transf \phi, \ve) &\geq \HP_\ve^{\aff}\!\left(\rho\gbar R_\FF(\phi) \right)\\
    1-E(\rho \transf \phi, p) &\geq \Theta_p\!\left(\rho\gbar R^\FF_\FF(\phi) \right)\\
  1-E(\rho \transf \phi, p) &\geq \Theta^{\aff}_p\!\left(\rho\gbar R_\FF(\phi) \right).
\end{aligned}\end{equation}
\end{proposition}
\begin{proof}
Let us begin with $\HP_\ve$ and $\Theta_\ve$. Construct the operation
\begin{equation}\begin{aligned}
  \E(X) \coloneqq \< W, X \> \phi + \< Z - W, X \> \frac{\lambda \sigma - \phi}{\lambda - 1}
\end{aligned}\end{equation}
where $\lambda = R_\FF^\FF(\phi)$ and $\phi \leq_\FF \lambda \sigma$. This gives
\begin{equation}\begin{aligned}
  \E(\rho) = \< Z , \rho \> \left[ \frac{\< W, \rho \>}{\< Z, \rho \>} \phi + \left(1 - \frac{\< W, \rho \>}{\< Z, \rho \>}\right) \frac{\lambda \sigma - \phi}{\lambda - 1} \right],
\end{aligned}\end{equation}
where we can see that the right-hand side is of the form $\< Z, \rho\> \tau$ for a normalised state $\tau$, explicitly showing that the operation achieves the approximate transformation $\rho \to \phi$ with probability $\< Z, \rho\>$ and error (at most) $1-\frac{\< W, \rho \>}{\< Z, \rho \>}$. To verify that this is a free operation, we see that
\begin{equation}\begin{aligned}
  \E(\sigma) &\geq_\FF \<W, \sigma\> \lambda \sigma \in \cone(\FF)
\end{aligned}\end{equation}
for any $\sigma \in \FF$. Taking $W, Z$ optimal for $\HP_\ve$ or $\Theta_p$ shows the desired achievability result.

Alternatively, if $R_\FF^\FF(\phi) = \infty$, then $\HP_\ve(\rho, \infty) = \HP(\rho, \infty)$ and it actually suffices to use the operation $\E(X) \coloneqq \< Z, X\> \phi$ which achieves the transformation exactly.

The argument for $\HP_\ve^{\aff}$ and $\Theta_p^{\aff}$ proceeds analogously, simply replacing $R_\FF^\FF$ with $R_\FF$ in the definition of~$\E$.
\end{proof}

\subsubsection*{Some thoughts about the methods of Theorem~\ref{thm:prob_ve_bounds}}

An interesting fact about the result of Theorem~\ref{thm:prob_ve_bounds}, closely related to known properties of isotropic states in entanglement theory~\cite{horodecki_1999-1} and observed also in the resource theory of coherence~\cite{fang_2018}, is that the state $\tau_\ve$ that we already encountered in Section~\ref{sec:prob_dist_RW} is always an optimal choice for the approximate distillation of a golden state $\phim$. Specifically, for $\phim$ obeying either of the conditions (i), (ii) of Theorem~\ref{thm:prob_ve_bounds}, it holds that
\begin{equation}\begin{aligned}
  P(\rho \transf \phim, \ve) = P(\rho \transf \tau_\ve)
\end{aligned}\end{equation}
with $\tau_\ve$ defined as in Eq.~\eqref{eq:tauve_def}. Intuitively, the conditions (i) and (ii) can be understood as ensuring that the state $\phim$ can be optimally distilled using measure-and-prepare operations as in \eqref{eq:meas_prep}, which makes it sufficient to use the `isotropic' state $\tau_\ve$. This is a priori much less obvious in general resource theories than in settings such as quantum entanglement, where the maximally resourceful state enjoys a number of symmetries that naturally lead to the concept of an isotropic state. (We refer the interested reader to Ref.~\cite{takagi_2022-1} where the properties of such states are explained from the perspective of a class of twirling maps applicable to more general quantum resources.)
We note, however, a difference between Theorem~\ref{thm:prob_ve_bounds} and the finding of Section~\ref{sec:prob_dist_RW} where the states $\tau_\ve$ also appear --- in Theorem~\ref{thm:prob_error_tight}, we only needed the resource theory to be affine to ensure the optimality of $\tau_\ve$, while (ii) here is a more demanding criterion.

Let us conclude by clarifying the connection between the optimisation problems considered in this section and the projective robustness $\RW$. 
Taking the map $\E(X) \coloneqq \< W, X \> \phi + \<Q, X \> \frac{\lambda \sigma - \phi}{\lambda - 1}$ as before, recall that the main idea of the proof of Theorem~\ref{thm:prob_ve_bounds} is to optimise the value of $\< W + Q, \rho \> = p$ for a fixed error $\ve$, or $\< W, \rho\> / \< W + Q , \rho \> = 1-\ve$ for a fixed $p$. If, instead, one wishes to minimise the error $\ve$ while disregarding the probability with which it is achieved, then this can be expressed as a maximisation of
\begin{equation}\begin{aligned}
  \frac{\<W, \rho \>}{\< Q, \rho \>} = \frac{1-\ve}{\ve}
\end{aligned}\end{equation}
under the same constraints. We can then get rid of the constraint $W + Q \leq \id$, since any non-zero feasible operators can be rescaled as $W \mapsto W/\norm{W+Q}{\infty}, Q \mapsto Q/\norm{W+Q}{\infty}$ with the same optimal value. The last constraint that we imposed in $\HP$ and $\Theta$ is $(t-1)\<W, \sigma \> \leq \<Q, \sigma\>$, so by defining $W' = (t-1) W$ we end up with the optimisation problem
\begin{equation}\begin{aligned}
  \sup \lset \frac{\<W', \rho \>}{\< Q, \rho \>} \bar W', Q \geq 0, \; \< W', \sigma \> \leq \< Q, \sigma \> \; \forall \sigma \in \FF \rset,
\end{aligned}\end{equation}
which is precisely $\RW$. This is consistent with our results in Section~\ref{sec:prob_dist_RW}, where we identified $\RW$ with the error $E(\rho \transf \phi, p)$ when $p$ is allowed to be arbitrarily small. Our upper and lower bounds in Theorem~\ref{thm:prob_ve_bounds} can then be understood to be analogous to the upper and lower bounds in Theorems \ref{thm:prob_error_full}--\ref{thm:prob_error_tight}. We thus see that the two approaches to characterising probabilistic distillation are closely related. 

Note, however, that the simpler form of the projective robustness $\RW$ and its properties make it useful as a bound even in settings where the approach of this section is difficult to apply: for instance, bounding the probability of many-copy transformations $\rho^{\otimes n} \transf \phi$ using Theorem~\ref{thm:prob_ve_bounds} would require evaluating the quantities $\HP_\ve$ or $\Theta_p$ for the tensor product state $\rho^{\otimes n}$, while the submultiplicativity of $\RW$ makes this significantly easier.

\section{Conclusions}

We introduced two general approaches to the characterisation of probabilistic resource transformations. 

We first studied the projective robustness $\RW$~\cite{regula_2022} --- a monotone under all probabilistic transformations of a given quantum resource. We proved that it yields broadly applicable no-go theorems for the manipulation of quantum states in general resource theories. In many relevant contexts, we established that the constraints obtained through $\RW$ are in fact necessary and sufficient, endowing the projective robustness with a fundamental operational meaning as a single quantity that fully determines the probabilistic convertibility between states with resource--non-generating operations. In resource distillation tasks, we showed that $\RW$ gives tight thresholds on the performance of all physical distillation schemes, expressed as universal bounds on the achievable accuracy and overhead in purifying quantum resources under the most general probabilistic manipulation protocols.

In our second approach, we developed precise quantitative restrictions on the achievable performance of one-shot probabilistic transformations. The bounds are all expressed in terms of convex conic optimisation problems, making them efficiently computable in many relevant cases. We showed that our restrictions become tight in broad types of resources when applied to the problem of distilling a maximally resourceful state, enabling us to provide an exact expression for the efficacy of distillation schemes in theories such as entanglement and coherence.

The importance of probabilistic manipulation protocols, and in particular the purification of noisy quantum resources, means that our results are bound to find use in many fundamental problems that underlie the practical exploitation of quantum mechanical phenomena. Here and in~\cite{regula_2022}, we presented explicit applications of our methods in the distillation of quantum coherence, entanglement, and magic states; we stress, however, that the bounds are fully general, and can be applied in settings beyond the ones considered here. An understanding of the properties of $\RW$ in any specific theory of interest will then provide additional insight into the manipulation of that resource.

A point of note is that many of our methods straightforwardly extend to settings beyond quantum theory. 
An examination of the proofs of the necessary (Theorem~\ref{thm:nogo_monotonicity}) and sufficient (Theorems~\ref{thm:nogo_affine_full} and~\ref{thm:nogo_sufficient_fulldim_full}) conditions for resource convertibility reveals that we did not actually use any properties specific to quantum mechanics. Indeed, the results can be adapted exactly as stated to the broad setting of general probabilistic theories (GPTs)~\cite{ludwig_1985,hartkamper_1974,davies_1970,lami_2017-1}. The study of resource theories can indeed be formulated in such contexts~~\cite{delrio_2015,coecke_2016,takagi_2019}, and we invite the interested reader to~\cite{takagi_2019,lami_2021} for a discussion of how to define quantities such as $R_\FF$ and $R^\FF_\FF$ in GPTs. The principal idea here is that the inequalities with respect to the cone of positive semidefinite operators are replaced with inequalities with respect to a cone $\mathcal{K}$ that defines the state space of the given GPT; for instance, the projective robustness can be defined as
\begin{equation}\begin{aligned}
  \Omega_{\FF_\mathcal{K}}^\mathcal{K} (\omega) \coloneqq \inf \lset \frac{\lambda}{\nu} \bar \omega -  \nu \sigma \in \mathcal{K},\;  \lambda \sigma - \omega \in \mathcal{K},\; \sigma \in \FF_\K \rset,
\end{aligned}\end{equation}
where $\FF_\K$ is now a subset of states in the given GPT. 
For simplicity, our results have focused on applications to \emph{quantum} resources, and we did not consider GPTs in more detail here. We note also that the results concerning distillation, namely Sections~\ref{sec:prob_dist_RW} and parts of Section~\ref{sec:tradeoff}, appear to be much more difficult to generalise to GPTs, as they make heavy use of the properties of pure states in quantum theory. We expect the possibility of a complete extension of all our bounds beyond quantum theory to be highly dependent on the specifics of the given GPT.

Returning to quantum resource theories, an intriguing open question is the possibility to establish tighter bounds by exploiting known structures of the allowed set of free operations $\OO$, since we know that the behaviour of probabilistic resource transformations can vary significantly between different sets of operations (as evidenced, for instance, by LOCC in entanglement theory). In particular, although our bounds are very suitable for the characterisation of high-rank noisy states, one could hope that another approach could yield better results when pure input states are considered --- the price we pay for the generality of our results is that constraints such as those based on the Schmidt rank~\cite{lo_2001,terhal_2000} are not recovered in our work. 
Another follow-up question is how the phenomenon of catalysis~\cite{jonathan_1999} can improve achievable probabilities --- it is known that catalysts often enhance the operational capabilities in resource manipulation, and it would be interesting to make this advantage precise in general probabilistic settings. Furthermore, following methods considered in~\cite{regula_2021-1}, one could extend our approach to establish analogous limitations on the probabilistic manipulation of quantum channels, leading to a generalisation of the framework considered here.

\begin{acknowledgments}
I am grateful to the anonymous referees for careful corrections and insightful comments that helped improve this work. Useful comments of Ludovico Lami and an anonymous referee of~\cite{regula_2022} are also acknowledged. I further thank Alexander Streltsov and Ryuji Takagi for discussions in the early stages of this work. 

This work was supported by the Japan Society for the Promotion of Science (JSPS) KAKENHI Grant No.\ 21F21015 and the JSPS Postdoctoral Fellowship for Research in Japan.
\end{acknowledgments}

\bibliographystyle{apsc}
\bibliography{../main}

  \newcounter{thmtemp}
  \setcounter{thmtemp}{\value{theorem}}

\appendix

\makeatletter
\renewcommand{\p@subsection}{\thesection.}
\makeatother

  \makeatletter
\def\renewtheorem#1{%
  \expandafter\let\csname#1\endcsname\relax
  \expandafter\let\csname c@#1\endcsname\relax
  \gdef\renewtheorem@envname{#1}
  \renewtheorem@secpar
}
\def\renewtheorem@secpar{\@ifnextchar[{\renewtheorem@numberedlike}{\renewtheorem@nonumberedlike}}
\def\renewtheorem@numberedlike[#1]#2{\newtheorem{\renewtheorem@envname}[#1]{#2}}
\def\renewtheorem@nonumberedlike#1{  
\def\renewtheorem@caption{#1}
\edef\renewtheorem@nowithin{\noexpand\newtheorem{\renewtheorem@envname}{\renewtheorem@caption}}
\renewtheorem@thirdpar
}
\def\renewtheorem@thirdpar{\@ifnextchar[{\renewtheorem@within}{\renewtheorem@nowithin}}
\def\renewtheorem@within[#1]{\renewtheorem@nowithin[#1]}
\makeatother

\section{Projective robustness in channel discrimination}\label{app:proj_chandisc}

Let us consider a setting in which an ensemble of quantum channels is given as $\{p_i, \N_i\}$, and one of the channels from the ensemble is randomly applied to a chosen input state $\rho$ with corresponding probability $p_i$. By measuring the output of this process with a chosen POVM $\{A_i\}$, we can attempt to distinguish which of the channels was applied to $\rho$. For a given ensemble, input state, and measurement, the average probability of successfully discriminating the channels is given by
\begin{equation}\begin{aligned}\label{eq:disc_prob}
  p_{\succ}(\rho, \{p_i, \N_i\}, \{A_i\}) = \sum_{i} p_i \Tr\!\left[ \N_i(\rho) A_i \right].
\end{aligned}\end{equation}
We then distinguish two scenarios. One is standard channel discrimination, in which our aim is to choose a measurement $\{A_i\}$ which maximises the probability of successfully distinguishing which channel was applied. The other scenario is channel exclusion (anti-distinguishability)~\cite{bandyopadhyay_2014}, where our aim is instead to exclude a channel, i.e.\ determine to the best of our ability which of the channels was \emph{not} applied; we then understand Eq.~\eqref{eq:disc_prob} as the average probability of guessing incorrectly, and we aim to minimise this quantity. In channel discrimination, the robustness $R_\FF(\rho)$ is known to quantify the advantage that a given input state $\rho$ can provide over all free states $\sigma$~\cite{takagi_2019-2}, while in channel exclusion, it is the weight $W_\FF(\rho)$ which quantifies the advantage~\cite{uola_2020-1,ducuara_2020}.
Combining the two, we can show that $\RW(\rho)$ quantifies the advantage when channel discrimination and channel exclusion tasks are considered \emph{at the same time}. 

\begin{boxed}{white}
\begin{theorem}
The projective robustness quantifies the maximal advantage that a given state $\rho$ gives over all free states $\sigma \in \FF$ in simultaneous channel discrimination and exclusion of a fixed channel ensemble, as quantified by the ratio
\begin{equation}\begin{aligned}
  \cfrac{p_{\succ}(\rho, \{p_i, \N_i\}, \{A_i\})}{p_{\succ}(\rho, \{p_i, \N_i\}, \{B_i\})}
\end{aligned}\end{equation}
where we aim to perform channel discrimination with measurement $\{A_i\}$ and channel exclusion with measurement $\{B_i\}$.

Specifically,
\begin{equation}\begin{aligned}\label{eq:channel_disc_statement}
  \sup_{\substack{\{p_i, \N_i\}\\\{A_i\}, \{B_i\}}} \, \cfrac{  \cfrac{p_{\succ}(\rho, \{p_i, \N_i\}, \{A_i\})}{p_{\succ}(\rho, \{p_i, \N_i\}, \{B_i\})} }{ \max_{\sigma \in \FF} \cfrac{p_{\succ}(\sigma, \{p_i, \N_i\}, \{A_i\})}{p_{\succ}(\sigma, \{p_i, \N_i\}, \{B_i\})}  } = \RW(\rho),
\end{aligned}\end{equation}
where the maximisation is over all finite channel ensembles $\{ p_i, \N_i\}_{i=1}^n$ and all POVMs $\{A_i\}_{i=1}^n, \{B_i\}_{i=1}^n$ for which the expression in Eq.~\eqref{eq:channel_disc_statement} is well defined.
\end{theorem}
\end{boxed}
\begin{proof}
We begin by showing that $\RW(\rho)$ always upper bounds the LHS of Eq.~\eqref{eq:channel_disc_statement}, so let us assume that it is finite and take $\rho \leq \lambda \sigma$, $\sigma \leq \mu\rho$ such that $\RW(\rho) = \lambda\mu$. Let $\{p_i, \N_i\}$ be any channel ensemble, and $\{A_i\}, \{B_i\}$ be any feasible POVMs. Then
\begin{equation}\begin{aligned}
  \cfrac{p_{\succ}(\rho, \{p_i, \N_i\}, \{A_i\})}{p_{\succ}(\rho, \{p_i, \N_i\}, \{B_i\})} &\leq  \cfrac{\lambda \, p_{\succ}(\sigma, \{p_i, \N_i\}, \{A_i\})}{\frac{1}{\mu} \, p_{\succ}(\sigma, \{p_i, \N_i\}, \{B_i\})}\\
  &\leq \lambda\mu \, \max_{\sigma' \in \FF} \cfrac{p_{\succ}(\sigma', \{p_i, \N_i\}, \{A_i\})}{p_{\succ}(\sigma', \{p_i, \N_i\}, \{B_i\})},
\end{aligned}\end{equation}
where the first inequality follows by the positivity of $A_i$ and $B_i$ and linearity of $p_{\succ}$ in $\rho$. Since this holds for any feasible channel discrimination task, we get the desired inequality.

To show the other inequality, we will use a binary channel discrimination task based on a construction of Ref.~\cite{takagi_2019}.
Let us then consider the channel ensemble $\{p_i, \N_i\}_{i=1}^2$ where $\N_1$ is the identity channel and $\N_2$ is arbitrary, and the probabilities are chosen as $p_1 = 1, p_2 = 0$. Now take any operators $A, B$ feasible for the dual formulation of $\RW(\rho)$ in Eq.~\eqref{eq:thm_convex_progC}, and define the POVMs $\{A_i\}, \{B_i\}$ as $A_1 = A / \norm{A}{\infty}$, $A_2 = \id - A_1$ and analogously for $B$. This gives
\begin{equation}\begin{aligned}
  \cfrac{  \cfrac{p_{\succ}(\rho, \{p_i, \N_i\}, \{A_i\})}{p_{\succ}(\rho, \{p_i, \N_i\}, \{B_i\})} }{ \max_{\sigma \in \FF} \cfrac{p_{\succ}(\sigma, \{p_i, \N_i\}, \{A_i\})}{p_{\succ}(\sigma, \{p_i, \N_i\}, \{B_i\})}  } &= \cfrac{ \cfrac{ \;\cfrac{\<A, \rho \>}{\norm{A}{\infty}}\; }{ \cfrac{\<B, \rho \>}{\norm{B}{\infty}} } }{ \max_{\sigma \in \FF} \cfrac{ \;\cfrac{\<A, \sigma \>}{\norm{A}{\infty}}\; }{ \cfrac{\<B, \sigma \>}{\norm{B}{\infty}}} }
  = \frac{ \<A, \rho \> }{ \< B, \rho \> } \min_{\sigma \in \FF} \frac{ \<B, \sigma \> }{ \< A, \sigma \> }
  \geq \frac{ \<A, \rho \> }{ \< B, \rho \> },
\end{aligned}\end{equation}
where the last inequality holds since any feasible $A, B$ satisfy $\< A, \sigma \> \leq \< B, \sigma \> \; \forall \sigma \in \FF$. Here we have constrained ourselves to operators $A, B$ such that $\< A, \rho \> \neq 0 \neq \< B, \rho \>$ and $\< A, \sigma \> \neq 0 \neq \< B, \sigma \> \; \forall \sigma \in \FF$ to ensure that Eq.~\eqref{eq:channel_disc_statement} is well defined. 
Since the supremum over all feasible $A, B$ is precisely $\RW(\rho)$, this establishes the projective robustness as a lower bound on the LHS of Eq.~\eqref{eq:channel_disc_statement}, concluding the proof.
\end{proof}


\end{document}